\documentclass[11pt,reqno]{amsart}
\usepackage{amssymb,latexsym}
\usepackage{url}
\usepackage[dvipdfm]{color}
\usepackage{graphicx, amsmath, amssymb}
\usepackage{epsfig}
\usepackage[pdfstartview=FitH,pdfpagemode=None,backref,colorlinks=true,citecolor=bluetxt,linkcolor=bluetxt]{hyperref} % DVI pink links!
\usepackage{fullpage}

\theoremstyle{plain}
\newtheorem{theorem}{Theorem}
\newtheorem{lemma}[theorem]{Lemma}
\newtheorem{corollary}[theorem]{Corollary}
\newtheorem{proposition}{Proposition}
\newtheorem{conjecture}{Conjecture}
\theoremstyle{definition}
\newtheorem{definition}{Definition}[section]
\newtheorem{fact}{Fact}
\newtheorem{claim}{Claim}
\newtheorem{remark}{Remark}
\newtheorem{observation}{Observation}
\newtheorem{example}{Example}
\newtheorem{note}{Note}
\newtheorem{comment}{Comment}

\renewenvironment{proof}{\QuadSpace\par\noindent{\bf
Proof}:}{\EndProof\HalfSpace}

\def\Res#1{{\mbox{Res}(#1)}}

\def\RL0{{\mbox{R(lin)}}}
\def\RZ0{{\mbox{R$^0$(lin)}}}
\def\RC0{R(lin) with constant coefficients}

\def\clique#1#2#3{{\mbox{$\neg$\textsc{clique}}$^{#1}_{#2,#3}$}}
\def\Tse0{{\mbox{$\neg$\textsc{Tseitin}$_{G,p}$}}}

\usepackage[dvipdfm]{color}
\definecolor{bluetxt}{rgb}{0,0,.5}
\definecolor{myred}{rgb}{0.6,0.0,0.1}
\definecolor{greentxt}{rgb}{0,.5,0}

\newenvironment{convention}{\QuadSpace\par\noindent{\bf Convention}:}{\HalfSpace}
\newenvironment{notation}{\QuadSpace\par\noindent{\bf Notation}:}{\HalfSpace}
\newcommand {\para}[1] {\smallskip
\paragraph{\textbf{#1}}}

\renewcommand{\l}{\ell}

\newcommand{\rst}{\!\!\upharpoonright}

\newcommand{\abs}[1]{\left\vert#1\right\vert}
\newcommand{\set}[1]{\left\{#1\right\}}

\newcommand{\F}{\mathbb F}
\newcommand{\eps}{\varepsilon}

\newcommand {\cd}{\cdot}
\newcommand {\zo}{\set{0,1}}
\newcommand {\OrDots} {\Or\cdots\Or}
\newcommand {\PlusDots} {+\ldots +}
\newcommand {\CommaDots} {,\ldots,}
\newcommand {\ML}[1] {\mathbf{M}\!\left[#1\right]} % multilinearize operator
\newcommand {\sm} {\setminus}
\newcommand {\se}{\subseteq}

\newcommand {\ind} {\noindent}

\newcommand {\bigs} {\bigskip}

\newcommand {\bx}   {\bar{x}}

\def\vr#1#2#3{{\mbox{$#1_{#2,#3}$}}} % e.g., for x_{i,j} write \vr x i j
\def\ps0{polynomial-size}

\newcommand{\BPP}{\textbf{BPP}}

\font\sf=cmss10
\newcommand{\Nats}{{\hbox{\sf I\kern-.13em\hbox{N}}}}   % Natural numbers
\newcommand{\Reals}{{\hbox{\sf I\kern-.14em\hbox{R}}}}  % Real numbers
\newcommand{\Ints}{{\hbox{\sf Z\kern-.43emZ}}}          % Integers
\newcommand{\CC}{{\hbox{\sf C\kern -.48emC}}}           % Complex numbers
\newcommand{\QQ}{{\hbox{\sf C\kern -.48emQ}}}           % Rational numbers

\newcommand{\such}{\;|\;}

\renewcommand{\And}{\land}
\newcommand{\Or}{\lor}
\newcommand{\Not}{\neg}

\newcommand{\BigOr}{\bigvee}
\newcommand{\BigAnd}{\bigwedge}
\ifx\theorem\undefined

\newtheorem{theorem}{Theorem}
\newtheorem{lemma}[theorem]{Lemma}
\newtheorem{proposition}[theorem]{Proposition}
\newtheorem{corollary}[theorem]{Corollary}
\newtheorem{definition}{Definition}[section]

\newtheorem{claim}{Claim}

\newenvironment{notation}{\QuadSpace\par\noindent{\bf Notation}:}{\HalfSpace}

\newenvironment{convention}{\QuadSpace\par\noindent{\bf Convention}:}{\HalfSpace}

\newenvironment{remark}{\HalfSpace\par\noindent{\bf Remark}:}{\HalfSpace}

\fi

\newenvironment{proofclaim}{\QuadSpace\par\noindent{\bf Proof of claim}:}
{\vrule width 1ex height 1ex depth 0pt\HalfSpace}

\newcommand{\QuadSpace}{\vspace{0.25\baselineskip}}
\newcommand{\HalfSpace}{\vspace{0.5\baselineskip}}
\newcommand{\FullSpace}{\vspace{1.0\baselineskip}}
\newcommand{\EndProof}{ \hfill \vrule width 1ex height 1ex depth 0pt }

\def\Res#1{{\mbox{Res}(#1)}}

\def\RL0{{\mbox{\rm R(lin)}}}
\def\RZ0{{\mbox{\rm R$^0$(lin)}}}
\def\RC0{R(lin) with constant coefficients}

\def\clique#1#2#3{{\mbox{$\neg$\textsc{clique}}$^{#1}_{#2,#3}$}}
\def\Tse0{{\mbox{$\neg$\textsc{Tseitin}$_{G,p}$}}}

\usepackage[dvipdfm]{color}
\definecolor{bluetxt}{rgb}{0,0,.5}
\definecolor{myred}{rgb}{0.6,0.0,0.1}
\definecolor{greentxt}{rgb}{0,.5,0}

\ifx\proof\undefined
\newenvironment{proof}{

\smallskip
\noindent\emph{Proof.}}{\hfill\(\Box\)
\bigskip
} \fi

\newlength{\defbaselineskip}
\def\Krajicek0{Kraj\'{i}\v{c}ek}
\def\af0{arithmetic formulas}
\def\fls0{formulas}
\def\ncf0{non-commutative formulas}
\def\CR0{Cook-Reckhow}
\def\ml0{multilinear}
\def\LB0{lower bound}

\title[Resolution over Linear Equations and Multilinear Proofs]{Resolution over Linear Equations \\ and Multilinear Proofs}
\author{Ran Raz}
\address{Department of Applied Mathematics and Computer Science,
Weizmann Institute, Rehovot 76100, Israel} \email{ranraz@wisdom.weizmann.ac.il}
\author{\;Iddo Tzameret}
\address{School of Computer
Science, Tel Aviv University, Tel Aviv 69978, Israel}
\email{tzameret@tau.ac.il}
\thanks{The first author was
supported by The Israel Science Foundation and The Minerva Foundation.\\
The second author was supported by The Israel Science Foundation (grant
no.~250/05).}

\begin{document}

 \maketitle

\begin{abstract}
We develop and study the complexity of propositional proof systems of varying
strength extending resolution by allowing it to operate with disjunctions of
linear equations instead of clauses.
 We demonstrate polynomial-size refutations for hard tautologies like
the pigeonhole principle, Tseitin graph tautologies and the clique-coloring
tautologies in these proof systems.
 Using the (monotone) interpolation by a communication game technique
we establish an exponential-size lower bound on refutations in a certain,
considerably strong, fragment of resolution over linear equations, as well as a
general polynomial upper bound on (non-monotone) interpolants in this fragment.

We then apply these results to extend and improve previous results on
multilinear proofs (over fields of characteristic $0$), as studied in
\cite{RT05}. Specifically, we show the following:

\begin{itemize}
\item Proofs operating with depth-$3$ multilinear formulas
polynomially simulate a certain, considerably strong, fragment of resolution
over linear equations.

\item Proofs operating with depth-$3$ multilinear formulas
admit polynomial-size refutations of the pigeonhole principle and Tseitin graph
tautologies. The former improve over a previous result that established small
multilinear proofs only for the \emph{functional} pigeonhole principle. The
latter are different than previous proofs, and apply to multilinear proofs of
Tseitin mod $p$ graph tautologies over any field of characteristic $0$.
\end{itemize}

We conclude by connecting resolution over linear equations with extensions of
the cutting planes proof system.
\end{abstract}

%\newpage
\tableofcontents
%\newpage

%====================================================
\section{Introduction}\label{sec-Intro}
%====================================================
This paper considers two kinds of proof systems.
The first kind are extensions of resolution that operate with disjunctions
of linear equations with integral coefficients instead of clauses.
The second kind are algebraic proof systems
operating with multilinear arithmetic formulas.
Proofs in both kinds of systems establish the unsatisfiability of
formulas in conjunctive normal form (CNF). We are primarily concerned with
connections between these two families of proof systems and with extending and
improving previous results on \ml0 proofs.

The resolution system is a popular propositional proof system that establishes
the unsatisfiability of CNF formulas (or equivalently, the truth of tautologies
in disjunctive normal form) by operating with clauses (a clause is a
disjunction of propositional variables and their negations).
It is well known that resolution cannot provide small (that is, polynomial-size)
proofs for many basic counting arguments.
The most notable example of this are the strong
exponential lower bounds on the resolution refutation size of the pigeonhole
principle and its different variants (Haken \cite{Hak85} was the first to establish
such a lower bound; see also \cite{Razb02} for a survey on the proof complexity
of the pigeonhole principle). Due to the popularity of resolution both in
practice, as the core of many automated theorem provers, and as a theoretical
case-study in propositional proof complexity, it is natural to consider weak
extensions of resolution that can overcome its inefficiency in providing proofs
of counting arguments.
The proof systems we present in this paper are extensions of resolution,
of various strength, that are suited for this purpose.

Propositional proof systems of a different nature that also attracted much
attention in proof complexity theory are \emph{algebraic proof systems}, which
are proof systems operating with (multivariate) polynomials over a field. In
this paper, we are particularly interested in algebraic proof systems that
operate with multilinear polynomials represented as multilinear arithmetic
formulas, called by the generic name \emph{multilinear proofs} (a polynomial is
\emph{multilinear} if the power of each variable in its monomials is at most one).
The investigation into such proof systems was initiated in \cite{RT05}, and here we
continue this line of research.
This research is motivated on the one hand by
the apparent considerable strength of such systems; and on the other hand, by the known
super-polynomial size lower bounds on multilinear formulas computing certain
important functions \cite{Raz04a,Raz04b}, combined with the general working
assumption that establishing lower bounds on the size
of \emph{objects} a proof system manipulates (in this case, multilinear formulas)
is close to establishing lower bounds on the size
of the \emph{proofs} themselves.

The basic proof system we shall study is denoted R(lin).
The proof-lines\footnotemark ~in R(lin) proofs are disjunctions of
linear equations with integral coefficients over the variables $\vec x =
x_1\CommaDots x_n$.
It turns out that (already proper subsystems of) R(lin) can handle very
elegantly basic counting arguments.
\footnotetext{Each element (usually a formula) of a proof-sequence
is referred to as a \emph{proof-line}.}
The following defines the R(lin) proof system.
Given an initial CNF, we translate every
clause $\BigOr_{i\in I} x_i \Or \BigOr_{j\in J} \Not x_j$ (where $I$ are the
indices of variables with positive polarities and $J$ are the indices of
variables with negative polarities) pertaining to the CNF, into the disjunction
$\BigOr_{i\in I} (x_i=1) \Or \BigOr_{j\in J} (x_j=0)$.
Let $A$ and $B$ be two disjunctions of linear equations,
and let $\vec a\cd\vec x =a_0$ and $\vec b\cd\vec x =b_0$
be two linear equations
(where $\vec a,\vec b$ are two vectors of $n$ integral coefficients,
and $\vec a\cd\vec x$ is the scalar product $\sum_{i=1}^n a_i x_i$;
and similarly for $\vec b\cd\vec x$).
The rules of inference belonging to R(lin) allow to derive
$A\Or B\Or((\vec a + \vec b)\cd\vec x =a_0+b_0)$
from $A\Or (\vec a\cd\vec x =a_0)$ and $B\Or (\vec b\cd\vec x=b_0)$
(or similarly, to derive
$A\Or B\Or((\vec a - \vec b)\cd\vec x =a_0-b_0)$
from $A\Or (\vec a\cd\vec x =a_0)$ and $B\Or (\vec b\cd\vec x=b_0)$).
%(added or subtracted in the obvious way).
We can also simplify disjunctions by discarding
(unsatisfiable) equations of the form $(0=k)$, for $k\ne 0$.
In addition, for every variable $x_i$, we shall add an axiom $(x_i=0)\Or(x_i=1)$,
which forces $x_i$ to take on only Boolean values.
A derivation of the empty disjunction (which stands for \textsc{false})
 from the (translated) clauses of a
CNF is called an \emph{R(lin) refutation} of the given CNF.
This way, every unsatisfiable CNF
has an R(lin) refutation (this can be proved by a straightforward simulation of
resolution by R(lin)).

The basic idea connecting resolution operating with disjunctions of linear
equations and multilinear proofs is this: Whenever a disjunction of linear
equations is simple enough --- and specifically,
when it is close to a symmetric
function, in a manner made precise --- then it can be represented by a small size
and small depth \ml0 arithmetic formula over fields of characteristic $0$. This
idea was already used (somewhat implicitly) in \cite{RT05} to obtain
polynomial-size \ml0 proofs operating with depth-$3$ multilinear formulas
of the functional pigeonhole principle (this principle is weaker than
the pigeonhole principle).
In the current paper we generalize previous results on multilinear proofs
by fully using this idea:
We show how to polynomially simulate with
multilinear proofs, operating with small depth \ml0 formulas,
certain short proofs carried inside resolution over linear equations.
This enables us to provide new polynomial-size multilinear proofs for
certain hard tautologies, improving results from \cite{RT05}.

More specifically,
we introduce a certain fragment of R(lin), which can be
polynomially simulated by depth-$3$ multilinear proofs
(that is, multilinear proofs operating with depth-$3$ multilinear
formulas).
 On the one hand this fragment of resolution over linear equations already
is sufficient to formalize in a transparent way basic counting arguments,
and so it admits small proofs of the pigeonhole principle and the Tseitin
mod $p$ formulas (which yields some new upper bounds on multilinear proofs);
and on the other hand we can use the (monotone) interpolation
technique to establish an exponential-size lower bound on refutations in this fragment
as well as demonstrating a general (non-monotone) polynomial upper bound
on interpolants for this fragment.
The possibility that multilinear proofs (possibly, operating with depth-$3$
multilinear formulas) possess the feasible
monotone interpolation property
(and hence, admit exponential-size lower bounds) remains open.

Another family of propositional proof systems we discuss in relation to
the systems mentioned above are the \emph{cutting planes} system and its
extensions. The cutting planes proof system operates with linear
\emph{inequalities} with integral coefficients, and this system is very close
to the extensions of resolution we present in this paper.
In particular, the following simple observation can be used to polynomially simulate
cutting planes proofs with polynomially bounded coefficients
(and some of its extensions) inside resolution over
linear equations:
The truth value of a linear inequality $\vec a\cd \vec x \ge a_0$
(where $\vec a$ is a vector of $n$ integral coefficients and
$\vec x$ is a vector of $n$ \emph{Boolean} variables)
%and $\vec a\cd\vec x$ is their scalar product $\sum_{i=1}^n a_i x_i$)
is equivalent to the truth value of the
following disjunction of linear equalities:
\[
\left(\vec a\cd\vec x = a_0\right)\Or \left(\vec a\cd\vec x = a_0+1\right)
\OrDots
\left(\vec a\cd\vec x = a_0+k\right)\,,
\]
where $a_0+k$ equals the sum of all positive coefficients in $\vec a$
(that is, $a_0+k = \max\limits_{\vec x\in\zo^n}\left(\vec a\cd\vec x\right)$).

\para{Note on terminology.}
All the proof systems considered in this paper intend to prove the
\emph{unsatisfiability} over $0,1$ values of collections of clauses (possibly,
of translation of the clauses to disjunctions of linear equations). In other
words, proofs in such proof systems intend to \emph{refute} the collections of
clauses, which is to validate their negation. Therefore, throughout this paper
we shall sometime speak about refutations and proofs interchangeably, always intending
refutations, unless otherwise stated.

%============================================
\subsection{Comparison to Earlier Work}
%============================================

To the best of our knowledge this paper is the first that considers resolution
proofs operating with disjunctions of linear \emph{equations}.
Previous works considered extensions of resolution over linear
\emph{inequalities} augmented
with the cutting planes inference rules (the resulting proof system denoted
R(CP)).
In full generality, we show that resolution over linear equations can
polynomially simulate R(CP) when the coefficients in all the inequalities are
polynomially bounded
(however, the converse is not known to hold).
On the other hand, we shall consider a certain fragment of resolution over
linear equations,
in which we do not even know how to polynomially simulate
cutting planes proofs with polynomially bounded coefficients in inequalities
(let alone R(CP) with polynomially bounded coefficients in inequalities).
We now shortly discuss the previous work on R(CP) and related proof systems.

Extensions of resolution to disjunctions of linear \emph{inequalities} were
first considered by Kraj{\'{\i}}{\v{c}}ek \cite{Kra98-Discretely} who developed
the proof systems LK(CP) and R(CP). The LK(CP) system is a first-order
(Gentzen-style) sequent calculus that operates with linear inequalities instead
of atomic formulas and augments the standard first-order sequent calculus
inference rules with the cutting planes inference rules. The R(CP) proof system
is essentially resolution over linear inequalities, that is, resolution that
operates with disjunctions of linear inequalities instead of clauses.

The main motivation of \cite{Kra98-Discretely} is to extend the feasible
interpolation technique and consequently the lower bounds
results, from cutting planes and resolution to stronger proof systems. That
paper establishes an exponential-size lower bound on a restricted version of
R(CP) proofs, namely, when the number of inequalities in each proof-line is
$O(n^\varepsilon)$, where $n$ is the number of variables of the initial
formulas, $\varepsilon$ is a small enough constant and the coefficients in the
cutting planes inequalities are polynomially bounded.

Other papers considering extensions of resolution over linear
inequalities are the more recent papers by Hirsch \& Kojevnikov \cite{HK06}
and Kojevnikov \cite{Koj07}.
The first paper \cite{HK06} considers a combination of resolution with LP
(an incomplete subsystem of cutting planes based on simple linear
programming reasoning),
with the `lift and project' proof system (L\&P),
and with the cutting planes proof system.
The second paper \cite{Koj07} deals with improving the parameters of the tree-like
R(CP) lower-bounds obtained in \cite{Kra98-Discretely}.

Whereas previous results concerned primarily with extending the cutting planes
proof system, our foremost motivation is to extend and improve previous results
on algebraic proof systems operating with multilinear formulas obtained in
\cite{RT05}. In that paper the concept of multilinear proofs was introduced and
several basic results concerning multilinear proofs were proved. In particular,
polynomial-size proofs of two important combinatorial principles were
demonstrated: the functional pigeonhole principle and the Tseitin (mod $p$)
graph tautologies. In the current paper we improve both these results.

As mentioned above, motivated by relations with multilinear proofs operating with
depth-$3$ multilinear formulas, we shall consider a certain subsystem
of resolution over linear equations.
For this subsystem we apply twice the interpolation by a communication
game technique.
The first application is of the \emph{non}-monotone version of the technique,
and the second application is of the \emph{monotone} version.
Namely, the first application provides a general (non-monotone)
interpolation theorem that demonstrates a polynomial (in the size of refutations)
upper bound on interpolants;
The proof uses the general method of transforming a
refutation into a Karchmer-Wigderson communication game for two players,
from which a Boolean circuit is then attainable.
In particular, we shall apply the interpolation theorem of \Krajicek0 from
\cite{Kra97-Interpolation}.
 The second application of the (monotone) interpolation by a communication
game technique is implicit and proceeds by using the lower bound criterion
of Bonet, Pitassi \& Raz in \cite{BPR97}.
This criterion states that (semantic) proof systems (of a certain natural and
standard kind) whose proof-lines (considered as Boolean functions)
have low communication complexity cannot prove efficiently a
certain tautology (namely, the clique-coloring tautologies).

%====================================
\subsection{Summary of Results}
%====================================

This paper introduces and connects several new concepts and ideas
with some known ones.
It identifies new extensions of resolution operating with linear
equations, and relates (a certain) such extension to multilinear proofs.
 The upper bounds for the pigeonhole principle and Tseitin mod $p$ formulas in fragments
of resolution over linear equations are new.
 By generalizing the machinery developed in \cite{RT05},
these upper bounds yield new and improved results concerning multilinear proofs.
 The lower bound for the clique-coloring formulas in a fragment of
resolution over linear equations employs the
standard monotone interpolation by a communication game technique,
and specifically utilizes the theorem of Bonet, Pitassi \& Raz from \cite{BPR97}.
The general (non-monotone) interpolation result
for a fragment of resolution over linear equations
employs the theorem of Kraj{\'i}\v{c}ek from \cite{Kra97-Interpolation}.
The upper bound in (the stronger variant of --  as described in the introduction)
resolution over linear equations
of the clique-coloring formulas follows that of Atserias, Bonet \&
Esteban \cite{ABE02}.
We now give a detailed outline of the results in this paper.

\para{The proof systems.}
In Section \ref{sec-systems-definitions}
we formally define two extensions of resolution of decreasing strength
allowing resolution to operate with disjunctions of linear equations.
The size of a linear equation $a_1 x_1\PlusDots a_n x_n = a_0$ is
the sum of all $a_0,\ldots, a_n$ written in \emph{unary notation}.
The size of a disjunction of linear equations is the total size of all linear equations in
the disjunction.
The size of a proof operating with disjunctions of linear equations
is the total size of all the disjunctions in it.
\HalfSpace

\textit{\RL0}: This is the stronger proof system (described in the
introduction) that operates with disjunctions of linear equations with integer
coefficients.\QuadSpace

\textit{\RZ0}: This is a (provably proper) fragment of R(lin).
It operates with disjunctions of (arbitrarily many) linear equations
whose variables have constant coefficients,
under the restriction that every disjunction can be partitioned
into a constant number of sub-disjunctions, where each sub-disjunction either
consists of linear equations that differ only in their free-terms or is a
(translation of a) clause. %(whose variables do not appear in other

Note that any single linear \emph{inequality} with Boolean variables can be represented
by a disjunction of linear equations that differ only in their
free-terms (see the example in the introduction section).
So the \RZ0 proof system is close to a proof system operating with
disjunctions of constant number of linear inequalities (with constant integral
coefficients). In fact, disjunctions of linear
equations varying only in their free-terms, have more (expressive)
strength than a single inequality. For instance, the \textsc{parity}
function can be easily represented by a disjunction of linear equations,
while it cannot be represented by a single linear inequality (or even by a
disjunction of linear inequalities).

As already mentioned, the motivation to consider the restricted proof system \RZ0
comes from its relation to multilinear proofs operating with depth-$3$
multilinear formulas (in short, depth-$3$ multilinear proofs):
\RZ0 corresponds roughly to the subsystem of R(lin) that we know how to
simulate by depth-$3$ multilinear proofs via the technique in \cite{RT05} (the
technique is based on converting disjunctions of linear forms into symmetric
polynomials, which are known to have small depth-$3$ multilinear formulas).
 This simulation is then applied in order to improve over known upper bounds
for depth-$3$ multilinear proofs, as \RZ0 is already sufficient to efficiently
prove certain ``hard tautologies".
 Moreover, we are able to establish an exponential lower bound on \RZ0
refutations size (see below for both upper and lower bounds on \RZ0 proofs).
 We also establish a super-polynomial separation of \RL0 from \RZ0
(via the clique-coloring principle, for a certain choice of parameters; see below).

\para{Short refutations.}
We demonstrate the following short refutations in \RZ0 and \RL0:%\vspace{-6pt}
\begin{enumerate}
\item Polynomial-size refutations of the pigeonhole principle in \RZ0;
%\vspace{-7pt}

\item Polynomial-size refutations of Tseitin mod $p$ graph formulas in \RZ0;
%\vspace{-7pt}
%In fact the R(lin) with constant coefficients

\item Polynomial-size refutations of the clique-coloring formulas in \RL0
(for certain parameters). The refutations here follow by direct simulation
of the \Res 2 refutations of clique-coloring formulas from \cite{ABE02}.
\end{enumerate}

All the three families of formulas above are prominent ``hard tautologies" in
proof complexity literature, which means that strong size lower bounds on
proofs in various proof systems are known for them (for the exact formulation
of these families of formulas see Section \ref{sec-hard}).

\para{Interpolation results.}
We provide a polynomial upper-bound on (non-monotone)
interpolants corresponding to \RZ0 refutations;
Namely, we show that any \RZ0-refutation of a given formula
can be transformed into a (non-monotone) Boolean circuit computing
the corresponding interpolant function of the formula (if there exists such
a function),
with at most a polynomial increase in size.
We employ the general interpolation theorem
of Kraj{\'i}\v{c}ek \cite{Kra97-Interpolation} for semantic proof systems.

\para{Lower bounds.}
We provide the following exponential lower bound:
% on the size of
%\RZ0 refutations of the clique-coloring formulas.
\begin{theorem}
\RZ0 does not have sub-exponential refutations for the clique-coloring
formulas.
\end{theorem}
This result is proved by applying a result of Bonet,
Pitassi \& Raz \cite{BPR97},
that (implicitly) use the monotone interpolation by a communication
game technique for establishing an exponential-size lower bound on refutations
of general semantic proof systems operating with proof-lines of
low communication complexity.
\QuadSpace

\para{Applications to multilinear proofs.}
%As mentioned in the introduction,
Multilinear proof systems are (semantic) refutation systems
operating with multilinear polynomials over a fixed field,
where every multilinear polynomial is represented by a multilinear
arithmetic formula.
In this paper we shall consider multilinear formulas over fields of characteristic
$0$ only.
The \emph{size} of a multilinear proof
(that is, a proof in a multilinear proof system)
is the total size of all multilinear formulas in the
proof (for formal definitions concerning multilinear proofs see Section
\ref{sec-multilinear}).

We shall first connect multilinear proofs with resolution over linear equations
by the following result:

\begin{theorem}
Multilinear proofs operating with depth-$3$ multilinear formulas over
characteristic $0$ polynomially-simulate \RZ0.
\end{theorem}

An immediate corollary of this theorem and the upper bounds in \RZ0 described
above are polynomial-size multilinear proofs for the pigeonhole principle and
the Tseitin mod $p$ formulas.

\begin{enumerate}
\item Polynomial-size depth-$3$ multilinear refutations for the pigeonhole
principle over fields of characteristic $0$. This improves over \cite{RT05}
that shows a similar upper bound for a weaker principle, namely, the
\emph{functional} pigeonhole principle.

\item Polynomial-size depth-$3$ multilinear refutations for the Tseitin mod $p$
graph formulas over fields of characteristic $0$. These refutations are
different than those demonstrated in \cite{RT05},
and further they establish short
multilinear refutations of the Tseitin mod $p$ graph formulas over \emph{any field
of characteristic $0$} (the proof in \cite{RT05} showed how to refute the
Tseitin mod $p$ formulas by multilinear refutations only over fields
that contain a primitive $p$th root of unity).
\end{enumerate}

\para{Relations with cutting planes proofs.} As mentioned in the
introduction, a proof system combining resolution with cutting planes was
presented by Kraj{\'{\i}}{\v{c}}ek in \cite{Kra98-Discretely}.
The resulting system is denoted R(CP)
(see Section \ref{sec-CP} for a definition).
When the coefficients in the linear inequalities inside R(CP) proofs are
polynomially bounded,
the resulting proof system is denoted R(CP*).
We establish the following simulation result:

\begin{theorem}
\emph{\RL0} polynomially simulates resolution over cutting planes inequalities
with polynomially bounded coefficients {\rm R(CP*)}.
\end{theorem}
We do not know if the converse also holds.

%====================================================
\section{Notation and Background on Propositional Proof Systems}\label{sec-notation}
%====================================================

For a natural number $n$, we use $[n]$ to denote $\set{1,\ldots, n}$.
For a vector of $n$ (integral) coefficients $\vec a$ and a
vector of $n$ variables $\vec x$, we denote by $\vec a \cd\vec x$ the scalar
product $\sum_{i=1}^n a_i x_i$.
If $\vec b$ is another vector (of length $n$),
then $\vec a+\vec b$ denotes the addition of $\vec a$ and $\vec b$ as vectors,
and $c\vec a$ (for an integer $c$) denotes the product of the scalar $c$ with $\vec a$
(where, $-\vec a$ denotes $-1\vec a$).
For two linear equations $L_1: \vec a \cd\vec x = a_0$ and
$L_2: \vec b\cd\vec x = b_0$, their addition
$(\vec a + \vec b)\cd \vec x=a_0+b_0$ is denoted $L_1+L_2$
(and their subtraction $(\vec a - \vec b)\cd \vec x=a_0-b_0$ is denoted $L_1-L_2$).
For two Boolean assignments (identified as $0,1$ strings)
$\alpha,\alpha'\in\zo^n$ we write $\alpha'\ge \alpha\,$ if
$\,\alpha'_i\ge \alpha_i$, for all $i\in[n]$ (where
$\alpha_i$,\,$\alpha'_i$ are the $i$th bits of $\alpha$ and $\alpha'$, respectively).

We now recall some basic concepts on propositional proof systems.
For background on algebraic proof systems (and specifically \ml0 proofs) see Section
\ref{sec-multilinear}.

\para{Resolution.}
In order to put our work in context, we need to define the resolution
refutation system.

A CNF formula over the variables $x_1, \ldots, x_n$ is defined as follows. A
\emph{literal} is a variable $x_i$ or its negation $\neg x_i$. A \emph{clause }
is a disjunction of literals. A \emph{CNF formula} is a conjunction of clauses.
The \emph{size of a clause} is the number of literals in it.

Resolution is a complete and sound proof system for unsatisfiable CNF formulas.
Let $C$ and $D$ be two clauses containing neither $x_i$ nor $\neg x_i$, the
\emph{resolution rule} allows one to derive $C\vee D$ from $C\vee x_i$ and
$D\vee \neg x_i$. The clause $C\vee D$ is called the \emph{resolvent} of the
clauses $C\vee x_i$ and $D\vee \neg x_i$ on the variable $x_i$, and we also say
that $C\vee x_i$ and $D\vee \neg x_i$ were \emph{resolved over $x_i$}. The
\emph{weakening rule} allows to derive the clause $C\vee D$ from the clause
$C$, for any two clauses $C,D$.

\begin{definition}[\textbf{Resolution}] \label{defnRes}
A \emph{resolution proof of the clause $D$ from a CNF formula $K$} is a
sequence of clauses $D_1,D_2,\ldots,D_\ell\,$, such that: (1) each clause $D_j$
is either a clause of $K$ or a resolvent of two previous clauses in the
sequence or derived by the weakening rule from a previous clause in the
sequence; (2) the last clause $D_\ell=D$.
The \emph{size} of a resolution proof is the sum of all the sizes of the
clauses in it.
%The width of a resolution proof is the maximal width of a clause in it.
A \emph{resolution refutation} of a CNF formula $K$ is
a resolution proof of the empty clause $\Box$ from $K$ (the empty clause stands
for \textsc{false}; that is, the empty clause has no satisfying assignments).
\end{definition}

A proof in resolution (or any of its extensions) is called also
a \emph{derivation} or a \emph{proof-sequence}.
Each sequence-element in a proof-sequence is called also a \emph{proof-line}.
A proof-sequence containing the proof-lines
$D_1\CommaDots D_\ell$ is also said to be a
\emph{derivation of $D_1\CommaDots D_\ell$}.

\para{Cook-Reckhow proof systems.}
Following \cite{CR79}, a \emph{Cook-Reckhow proof system} is a polynomial-time
algorithm $A$ that receives a Boolean formula $F$ (for instance, a CNF) and a string
$\pi$ over some finite alphabet (``the (proposed) refutation" of $F$), such
that there exists a $\pi$ with $A(F,\pi)=1$ if and only if $F$ is
unsatisfiable.
The \emph{completeness} of a (Cook-Reckhow) proof system
(with respect to the set of all unsatisfiable Boolean formulas;
or for a subset of it, e.g. the set of unsatisfiable CNF formulas)
stands for the fact that
every unsatisfiable formula $F$ has a string $\pi$
(``the refutation of $F$") so that $A(F,\pi)=1$.
The \emph{soundness} of a (Cook-Reckhow) proof system stands for the fact that
every formula $F$ so that
$A(F,\pi)=1$ for some string $\pi$ is unsatisfiable
(in other words, no satisfiable formula has a refutation).

For instance, resolution is a Cook-Reckhow proof system,
since it is complete and sound for the set of unsatisfiable CNF formulas, and
given a CNF formula $F$ and a string $\pi$ it is easy to check in polynomial-time
(in \emph{both $F$ and $\pi$}) whether $\pi$ constitutes a resolution refutation of $F$.

We shall also consider proof systems that are not necessarily
(that is, not known to be) Cook-Reckhow proof systems.
Specifically, multilinear proof systems (over large enough fields)
meet the requirements in the definition of Cook-Reckhow proof systems,
\emph{except} that the condition on $A$ above is relaxed:
we allow $A$ to be in \emph{probabilistic} polynomial-time \BPP
~(which is not known to be equal to deterministic polynomial-time).
%(For multilinear proof systems, see Section \ref{sec-multilinear} and \cite{RT05}.)

\para{Polynomial simulations of proof systems.}
When comparing the strength of different proof systems we shall confine
ourselves to CNF formulas only. That is, we consider propositional proof
systems as proof systems for the set of unsatisfiable CNF formulas. For that
purpose, if a proof system does not operate with clauses directly,
then we fix a (direct) translation from clauses to
the objects operated by the proof system. This is done for both resolution over
linear equations (which operate with disjunctions of linear
equations) and its fragments,
and also for multilinear proofs
(which operate with multilinear polynomials, represented as multilinear formulas);
see for example Subsection
\ref{sec-clauses-of-linear-equations} for such a direct translation.

\begin{definition}\label{defnSim} Let $\mathcal P_1, \mathcal P_2$ be
two proof systems for the set of unsatisfiable CNF formulas
(we identify a CNF formula with its
corresponding translation, as discussed above). We say that $\mathcal P_2$
\emph{polynomially simulates} $\mathcal P_1$
if given a $\mathcal P_1$ refutation $\pi$ of a CNF formula $F$,
then there exists a refutation of $F$ in $\mathcal P_2$ of size polynomial
in the size of $\pi$.
 In case $\mathcal P_2$ polynomially simulates $\mathcal P_1$ while $\mathcal P_1$
does not polynomially simulates $\mathcal P_2$ we say that $\mathcal P_2$
is \emph{strictly stronger} than $\mathcal P_1$.% Given
%an unsatisfiable CNF formula $F$, we say that $P_2$ \emph{has an exponential
%gap over} $P_1$ \emph{for} $F$ if there exists a polynomial size $P_2$
%refutation of $F$ and the smallest $P_1$ refutation of $F$ is of exponential
%size.
\end{definition}

%====================================================
\section{Resolution over Linear Equations and its Subsystems}
%====================================================
\label{sec-systems-definitions}

The proof systems we consider in this section are extensions of resolution.
Proof-lines in resolution are clauses. Instead of this,
the extensions of resolution we consider here operate with disjunctions of
linear equations with integral coefficients.
For this section we use the convention that all the formal variables
in the propositional proof systems
considered are taken from the set $X:=\set{x_1\CommaDots x_n}$.

%=============================================
\subsection{Disjunctions of Linear Equations}
\label{sec-clauses-of-linear-equations}
%=============================================

For $L$ a linear equation $a_1 x_1\PlusDots a_n x_n = a_0$, the right hand
side $a_0$ is called the \emph{free-term} of $L$
and the left hand side $a_1 x_1\PlusDots a_n x_n$ is called the
\emph{linear form} of $L$ (the linear form can be $0$).
 A \emph{disjunction of linear equations} is of
the following general form:

\begin{equation}\label{eq-example-of-disjunction-of-linear-equations}
\left(a^{(1)}_1 x_1+\ldots + a^{(1)}_{n} x_n=a^{(1)}_0\right) \OrDots
 \left(a^{(t)}_1 x_1+\ldots + a^{(t)}_{n} x_n=a^{(t)}_0\right)\,,
\end{equation}
where $t\ge 0$ and the coefficients $a^{(j)}_i$ are integers (for all $0\le
i\le n,\; 1\le j\le t$).
We discard duplicate linear equations from a disjunction of linear equations.
The semantics of such a disjunction is the natural
one: We say that an assignment of integral values to the variables
$x_1,...,x_n$ \emph{satisfies}
(\ref{eq-example-of-disjunction-of-linear-equations}) if and only if there
exists $j\in[t]$ so that the equation $a^{(j)}_1 x_1+\ldots + a^{(j)}_{n}
x_n=a^{(j)}_0$ holds under the given assignment.

The symbol $\models$ denotes the \emph{semantic implication} relation, that is,
for every collection $D_1\CommaDots D_m$ of disjunctions of linear equations,
$$D_1\CommaDots D_m \models D_0$$
 means that every assignment of $0,1$
values that satisfies all $D_1\CommaDots D_m$ also satisfies $D_0$.\footnotemark
~In this case we also say that $D_1\CommaDots D_m$ \emph{semantically
imply} $D_0$.
\footnotetext{Alternatively, we can consider assignments of any integral
values (instead of only Boolean values) to the variables in
$D_1\CommaDots D_m$, stipulating that the collection
$D_1\CommaDots D_m$ contains all disjunctions of the form $\,(x_j=0)\Or(x_j=1)$
for all the variables $x_j\in X$
(these formulas force any satisfying assignment to give only $0,1$
 values to the variables).}

The \emph{size of a linear equation} $a_1 x_1\PlusDots a_n x_n = a_0$
is $\sum_{i=0}^n{|a_i|}$, i.e., the sum of the bit sizes of all $a_i$ written in
\emph{unary} notation.
Accordingly, the \emph{size of the linear form} $a_1 x_1\PlusDots a_n x_n$
is $\sum_{i=1}^n{|a_i|}$.
The \emph{size of a disjunction of linear equations} is the total size of all
linear equations in it.

Since all linear equations considered in this paper are of integral coefficients,
we shall speak of \emph{linear equations} when we actually mean linear
equations with integral coefficients. %We shall sometime call disjunctions of
%linear equations with integer coefficients \emph{clauses of linear equations}.
Similar to resolution, the \emph{empty disjunction} is unsatisfiable and stands
for the truth value \textsc{false}. %The \emph{width} of a clause of linear equations is

\para{Translation of clauses.}
As described in the introduction,
we can translate
any CNF formula to a collection of disjunctions of linear equations in a
direct manner: Every clause $\BigOr_{i\in I} x_i \Or \BigOr_{j\in J}
\Not x_j$ (where $I$ and $J$ are sets of indices of variables)
%are the indices of variables with positive polarities and
%$J$ are the indices of variables with negative polarities)
pertaining to the CNF is translated into the disjunction
$\BigOr_{i\in I} (x_i=1) \Or \BigOr_{j\in J} (x_j=0)$.
For a clause $D$ we denote by $\widetilde{D}$ its
translation into a disjunction of linear equations. It is easy to verify that
any Boolean assignment to the variables $x_1\CommaDots x_n$ satisfies a clause
$D$ if and only if it satisfies $\widetilde{D}$ (where $\textsc{true}$ is
treated as $1$ and $\textsc{false}$ as $0$).

\subsection{Resolution over Linear Equations -- R(lin)}

Defined below is our basic proof system R(lin) that enables resolution to
reason with disjunctions of linear equations. As we wish to reason about
Boolean variables we augment the system with the axioms $\,(x_i=0)\Or(x_i=1)$,
for all $i\in[n]$, called the \emph{Boolean axioms}.
%The Boolean axioms force all the variables to take on only $0,1$ values.
%In order to be able to flip sides of constants in equations

\begin{definition}[\textbf{R(lin)}]\label{def-R(lin)}
Let $K:= \set{K_1,\ldots,K_m}$ be a collection of disjunctions of linear
equations. An \emph{R(lin)-proof from $K$ of a disjunction of linear equations
$D$} is a finite sequence $\pi =(D_1 ,...,D_\ell)$ of disjunctions of linear
equations, such that $D_\ell=D$ and for every $i\in[\ell]$,
either $\,D_i = K_j\,$ for some $j\in[m]$,
or $D_i$ is a Boolean axiom $\,(x_h=0)\Or(x_h=1)$ for some $h\in[n]$,
or $D_i$ was deduced by one of the following R(lin)-inference rules,
 using $D_j,D_k$ for some $j,k<i$:

\begin{description}
\item[\quad Resolution] Let $A,B$ be two disjunctions\footnotemark of linear equations
and let $L_1,L_2$ be two linear equations.

From $A\Or L_1$ and $B\Or L_2$ derive $A \Or B \Or(L_1+L_2)$.

Similarly, from $A\Or L_1$ and $B\Or L_2$ derive $A \Or B \Or(L_1-L_2)$.

\item[\quad Weakening] From a disjunction of linear equations $A$ derive
$A \Or L$\,, where $L$ is an arbitrary linear equation over $X$.

\item[\quad Simplification] From $A\Or (0=k)$ derive
$A$, where $A$ is a disjunction of linear equations and $k\ne 0$.
\footnotetext{Possibly the empty disjunction. This remark also applies
to the inference rules below.}
\end{description}
An \emph{R(lin) refutation of} a collection of disjunctions of linear equations
$K$ is a proof of the empty disjunction from $K$. The \emph{size} of an R(lin)-proof
$\pi$ is the total size of all the disjunctions of linear equations in
$\pi$, denoted $|\pi|$.
\end{definition}

Similar to resolution,
in case $A \Or B \Or(L_1+L_2)$ is derived from $A\Or L_1$ and $B\Or L_2$ by
the resolution rule, we say that $A\Or L_1$ and $B\Or L_2$ were \emph{resolved
over $L_1$ and $L_2$}, respectively, and we call $A \Or B \Or(L_1+L_2)$ the
\emph{resolvent} of $A\Or L_1$ and $B\Or L_2$ (and similarly, when $A \Or B
\Or(L_1-L_2)$ is derived from $A\Or L_1$ and $B\Or L_2$ by the resolution
rule; we use the same terminology for both addition and subtraction,
and it should be clear from the context which operation is actually applied).
We also describe such an application of the
resolution rule by saying that \emph{$L_1$ was added (resp., subtracted) to (resp.
from) $L_2$ in $A\Or L_1$ and $B\Or L_2$}.

In light of the direct translation between CNF formulas and collections of
disjunctions of linear equations (described in the previous subsection), we can
consider R(lin) to be a proof system for the set of unsatisfiable CNF formulas:

\begin{proposition}\label{prop-R(lin)-sim-resolution}
The R(lin) refutation system is a sound and complete Cook-Reckhow (see Section
\ref{sec-notation}) refutation system for unsatisfiable CNF formulas
(translated into unsatisfiable collection of disjunctions of linear equations).
\end{proposition}

\begin{proof}
Completeness of R(lin) (for the set of unsatisfiable CNF formulas) stems from a
straightforward simulation of resolution, as we now show.

\begin{claim}\label{cla-R(lin)-sim-res-straitforward}
R(lin) polynomially simulates resolution.
\end{claim}

\begin{proofclaim}
Proceed by induction on the length of the resolution refutation to show that
any resolution derivation of a clause $A$ can be translated
with only a linear increase in size into an R(lin)
derivation of the corresponding disjunction of linear equations
$\widetilde{A}$ (see the previous subsection for
 the definition of $\widetilde A$).

\emph{The base case:
}An initial clause $A$ is translated into its corresponding disjunction
of linear equations $\widetilde{A}$.

\emph{The induction step:
}If a resolution clause $A\Or B$ was derived by the resolution
rule from $A\Or x_i$ and $B\Or \Not x_i$, then in R(lin)
we subtract $(x_i=0)$ from $(x_i=1)$ in
$\widetilde{B}\Or (x_i=0)$ and $\widetilde{A}\Or (x_i=1)$, respectively,
to obtain
$\widetilde{A}\Or\widetilde{B}\Or(0=1)$.
Then, using
the Simplification rule, we can cut-off $(0=1)$ from
$\widetilde{A}\Or\widetilde{B}\Or(0=1)$,
and arrive at $\widetilde{A}\Or\widetilde{B}$.

If a clause $A\Or B$ was derived in resolution from $A$ by the Weakening rule,
then we derive $\widetilde{A}\Or\widetilde{B}$ from
$\widetilde{A}$ by the Weakening rule in R(lin).
\end{proofclaim}

Soundness of R(lin) stems from the soundness of the inference rules (which
means that: If $D$ was derived from $C,B$ by the R(lin) resolution rule then
any assignment that satisfies both $C$ and $B$ also satisfies $D$; and if $D$
was derived from $C$ by either the Weakening rule or the Simplification rule,
then any assignment that satisfies $C$ also satisfies $D$).\QuadSpace

The R(lin) proof system is a Cook-Reckhow proof system,
as it is easy to verify
in polynomial-time whether  an R(lin) proof-line is inferred, by an application of
one of R(lin)'s inference rules, from a previous proof-line (or proof-lines).
Thus, any sequence of disjunctions of linear equations,
can be checked in polynomial-time (in the size of the sequence)
to decide whether or not it is a legitimate R(lin) proof-sequence.
\end{proof}

In Section \ref{sec-impl-complete} we shall
see that a stronger notion of completeness (that is, implicational completeness)
holds for R(lin) and its subsystems.

\subsection{Fragment of Resolution over Linear Equations -- R$^{\mathbf 0}$(lin)}

Here we consider a restriction of \RL0, denoted \RZ0.
As discussed in the introduction section,
\RZ0 is roughly the fragment of R(lin) we know how to polynomially simulate with
depth-$3$ multilinear proofs.

By results established in the sequel (Sections \ref{sec-clique}
and \ref{sec-lower-bounds}) R(lin) is \emph{strictly stronger} than \RZ0,
which means that R(lin) polynomially simulates \RZ0,
while the converse does not hold.

\RZ0 operates with disjunctions of (arbitrarily many) linear equations with
constant coefficients (excluding the free terms), under the following
restriction: Every disjunction can be partitioned into a constant number of
sub-disjunctions, where each sub-disjunction either consists of linear
equations that differ only in their free-terms or is a (translation of a)
clause.% whose variables do not appear in other sub-disjunctions.

As mentioned in the introduction, every linear \emph{inequality}
with Boolean variables can be represented by a disjunction of linear
equations that differ only in their free-terms.
So the \RZ0 proof system resembles, to some extent, a proof system operating with
disjunctions of constant number of linear inequalities with constant integral
coefficients (on the other hand, it is probable that \RZ0 is stronger
than such a proof system, as a disjunction of linear equations that differ only
in their free terms is [expressively] stronger than a linear inequality
[or even a disjunction of linear inequalities]:
the former can define the \textsc{parity}
function while the latter cannot).\QuadSpace

\ind\emph{Example of an \RZ0-line:}
$$(x_1+\ldots+x_{\ell}=1)\Or\cdots\Or(x_1+\ldots+x_{\ell}=\ell)
\Or(x_{\ell+1}=1)\Or\cdots\Or(x_n=1),$$ for some $1\le \ell \le n$.
The next section contains other concrete (and natural) examples of \RZ0-lines.
\QuadSpace

Let us define formally what it means to be an \RZ0 proof-line, that is, a
proof-line inside an \RZ0 proof, called {\em \RZ0-line}:

\begin{definition}[R$^{\mathbf 0}$(lin)-line]\label{def-RZ0-clause}
Let $D$ be a disjunction of linear equations whose variables have constant integer
coefficients (the free-terms are unbounded).
Assume $D$ can be partitioned
into a constant number $k$ of sub-disjunctions $D_1\CommaDots D_k$, where each
$D_i$ either consists of (an unbounded) disjunction of linear equations that
differ only in their free-terms, or is a translation of a clause (as defined in
Subsection \ref{sec-clauses-of-linear-equations}).
%whose variables do not appear in any $D_j$ for $j\ne i$.
Then the disjunction $D$ is called an
\emph{\RZ0-line}.
\end{definition}

Thus, any \RZ0-line is of the following general form:

\begin{equation}\label{eq-R0-clause}
	\BigOr_{i\in I_1}
		\left(
			\vec a^{(1)}\cd\vec x = \ell^{(1)}_i
		\right)
	\OrDots
	\BigOr_{i\in I_k}
		\left(
			\vec a^{(k)} \cd \vec x =\ell^{(k)}_i
		\right)
	\Or
	\BigOr_{j\in J}(x_j=b_j)
\,,
\end{equation}
where $k$ and all $a^{t}_r$ (for  $r\in[n]$ and $t\in[k]$) are integer
constants and $b_j\in\zo$ (for all $j\in J$) (and $I_1\CommaDots I_k, J$ are
unbounded sets of indices). Note that a disjunction of clauses can be combined
into a single clause. Hence, without loss of generality we can assume that in
any \RZ0-line only a single (translation of a) clause occurs. This is depicted
in (\ref{eq-R0-clause}) (where in addition we have ignored in
(\ref{eq-R0-clause}) the possibility that the single clause obtained by
combining several clauses contains $x_j\Or\neg x_j$, for some $j\in[n]$).
\QuadSpace

\begin{definition}[R$^{\mathbf 0}$(lin)]
The \emph{\RZ0} proof system is a restriction of the \RL0 proof system in which
each proof-line is an \RZ0-line (as in Definition \ref{def-RZ0-clause}).
\end{definition}

For a completeness proof of \RZ0 see Section \ref{sec-impl-complete}.\footnotemark
\footnotetext{The simulation of resolution inside \RL0 (in the proof
of Proposition \ref{prop-R(lin)-sim-resolution}) is carried on with each \RL0
proof-line being in fact a translation of a clause, and hence, an \RZ0-line
(notice that the Boolean axioms of R(lin) are \RZ0-lines).
This already implies that \RZ0 is a complete refutation system
for the set of unsatisfiable CNF formulas. In section \ref{sec-impl-complete}
we give a proof of a stronger notion of completeness for \RZ0.}

%====================================================
\section{Reasoning and Counting inside R(lin) and its Subsystems}
%====================================================
In this section we illustrate a simple way to reason by
case-analysis inside R(lin) and its subsystems.
This kind of reasoning will simplify the presentation of proofs inside R(lin)
(and \RZ0) in the sequel (essentially, a similar -- though
weaker -- kind of reasoning is applicable already in resolution).
We will then demonstrate efficient and transparent proofs for simple
counting arguments
that will also facilitate us in the sequel.

\subsection{Basic Reasoning inside R(lin) and its Subsystems}\label{sec-case-anl}

Given $K$ a collection of disjunctions of linear equations
$\set{K_1,\ldots,K_m}$ and $C$ a disjunction of linear equations,
denote by $K\Or C$ the collection $\set{K_1\Or C,\ldots,K_m\Or C}$.
Recall that the formal variables in our proof system are $x_1\CommaDots x_n$.

\begin{lemma}\label{lem-R(lin)-case-analysis}
Let $K$ be a collection of disjunctions of linear equations,
and let $z$ abbreviate some linear form with integer coefficients.
Let $E_1\CommaDots E_\ell$ be $\ell$ disjunctions of linear equations.
Assume that for all $i\in[\ell]$ there is an R(lin) derivation of $E_i$
from $z=a_i$ and $K$ with size at most $s$ where $a_1,\ldots,a_\l$ are
distinct integers.
Then, there is an R(lin) proof of $\BigOr_{i=1}^{\ell} E_i$
from $K$ and $(z=a_1)\OrDots(z=a_\ell)$, with size polynomial in $s$ and $\ell$.
\end{lemma}

\begin{proof}
Denote by $D$ the disjunction $(z=a_1)\OrDots(z=a_\ell)$ and by
$\pi_i$ the \RL0 proof of $E_i$ from $K$ and $z=a_i$ (with size at most $s$),
for all $i\in[\ell]$.
It is easy to verify that for all $i\in[\ell]$ the sequence
$\pi_i\Or \BigOr_{j\in[\ell]\sm \set{i}}(z=a_j)$ is an \RL0 proof of
$E_i\Or \BigOr_{j\in[\ell]\sm \set{i}}(z=a_j)$ from $K$ and $D$.
So overall, given $D$ and $K$ as premises, there is an R(lin)
derivation of size polynomial in $s$ and $\ell$ of the following
collection of disjunctions of linear equations:
\begin{equation}\label{eq-case-analysis}
E_1\Or\BigOr_{j\in[\ell]\sm \set{1}}(z=a_j)\CommaDots
E_\ell\Or\BigOr_{j\in[\ell]\sm \set{\ell}}(z=a_j)\,.
\end{equation}

We now use the Resolution rule to cut-off all the equations
$(z=a_i)$ inside all the disjunctions in (\ref{eq-case-analysis}).
Formally, we prove that for every $1\le k \le\ell$ there is a
polynomial-size (in $s$ and $\ell$) R(lin) derivation from
(\ref{eq-case-analysis}) of
\begin{equation}\label{eq-case-analysis-ind1}
E_1\OrDots E_{k}\Or \BigOr_{j\in[\ell]\sm [k]}(z=a_j)\,,
\end{equation}
and so putting $k=\ell$, will conclude the proof of the lemma.

We proceed by induction on $k$. The base case for $k=1$ is immediate
(from (\ref{eq-case-analysis})).
For the induction case,
assume that for some $1\le k<\ell$ we already have an R(lin) proof
of (\ref{eq-case-analysis-ind1}),
with size polynomial in $s$ and $\ell$.

Consider the line
\begin{equation}\label{eq-case-anlsys-consider-k+1}
 E_{k+1}
	\Or
	\BigOr_{
			j\in[\ell]\sm \set{k+1}
			}
		(z=a_j)
\,.
\end{equation}
We can now cut-off the disjunctions
$\BigOr_{j\in[\ell]\sm [k]}(z=a_j)$
and $\BigOr_{j\in[\ell]\sm \set{k+1}}(z=a_j)$
from (\ref{eq-case-analysis-ind1})
and (\ref{eq-case-anlsys-consider-k+1}), respectively, using the Resolution rule
(since the $a_j$'s in (\ref{eq-case-analysis-ind1}) and in
(\ref{eq-case-anlsys-consider-k+1}) are disjoint).
We will demonstrate this derivation in some detail now,
in order to exemplify a proof carried inside R(lin).
We shall be less formal sometime in the sequel.

Resolve (\ref{eq-case-analysis-ind1}) with (\ref{eq-case-anlsys-consider-k+1})
over $(z=a_{k+1})$ and $(z=a_1)$, respectively, to obtain
\begin{equation}\label{eq-case-analysis-ind2}
(0=a_1-a_{k+1})\Or E_1\OrDots E_{k}\Or E_{k+1}\Or
\BigOr_{j\in[\ell]\sm\set{1,k+1}}(z=a_j)\,.
\end{equation}
Since $a_1\neq a_{k+1}$,
we can use the Simplification rule to cut-off $(0=a_1-a_{k+1})$
from (\ref{eq-case-analysis-ind2}), and we arrive at
\begin{equation}\label{eq-case-analysis-ind3}
E_1\OrDots E_{k}\Or E_{k+1}\Or\BigOr_{j\in[\ell]\sm\set{1,k+1}}(z=a_j)\,.
\end{equation}
Now, similarly, resolve (\ref{eq-case-analysis-ind1}) with
(\ref{eq-case-analysis-ind3}) over $(z=a_{k+1})$ and $(z=a_2)$,
respectively, and use Simplification to obtain
\begin{equation*}
E_1\OrDots E_{k}\Or E_{k+1}\Or\BigOr_{j\in[\ell]\sm\set{1,2,k+1}}(z=a_j)\,.
\end{equation*}
Continue in a similar manner until you arrive at
\begin{equation*}
E_1\OrDots E_{k}\Or E_{k+1}\Or\BigOr_{j\in[\ell]
\sm\set{1,2\CommaDots k,k+1}}(z=a_j)\,,
\end{equation*}
which is precisely what we need.
\end{proof}

Under the appropriate conditions, Lemma \ref{lem-R(lin)-case-analysis} also
holds for \RZ0 proofs. %and \RC0 proofs.
This is stated in the following lemma.

\begin{lemma}\label{lem-R0(lin)-case-anl}
Let $K$ be a collection of disjunctions of linear equations,
and let $z$ abbreviate a linear form with integer coefficients.
Let $E_1\CommaDots E_\ell$ be $\ell$ disjunctions of linear equations.
Assume that for all $i\in[\ell]$ there is an \RZ0
derivation of $E_i$ from $z=a_i$ and $K$ with size at most $s$,
where the $a_i$'s are distinct integers.
Then, assuming $\BigOr_{i=1}^{\ell} E_i$ is an \RZ0-line,
there is an \RZ0 proof of $\BigOr_{i=1}^{\ell} E_i$
from $K$ and $(z=a_1)\OrDots(z=a_\ell)$,
with size polynomial in $s$ and $\ell$.
\end{lemma}

\begin{proof}
It can be verified by simple inspection
that,
under the conditions spelled out in the statement of the lemma,
each proof-line in the R(lin) derivations in the proof of
Lemma \ref{lem-R(lin)-case-analysis} is actually an \RZ0-line.\footnotemark
\footnotetext{Note that when the proofs of $E_i$ from $z=a_i$,
for all $i\in[\ell]$, are all done inside \RZ0,
then the linear form $z$ ought to have \emph{constant} coefficients.}
\end{proof}

\para{Abbreviations.}
Lemmas \ref{lem-R(lin)-case-analysis} and \ref{lem-R0(lin)-case-anl}
will sometime facilitate us to
proceed inside R(lin) and \RZ0 with a slightly less formal manner.
For example, the situation in Lemma \ref{lem-R(lin)-case-analysis} above can be
depicted by saying that
``if $z=a_i$ implies $E_i$ (with a polynomial-size proof)
for all $i\in[\ell]$, then $\BigOr_{i=1}^\ell(z=a_i)$
implies $\BigOr_{i=1}^\ell E_i$ (with
a \ps0 proof)".

In case $\BigOr_{i=1}^\ell(z=a_i)$ above is just the
\emph{Boolean axiom} $(x_i=0)\Or (x_i=1)$, for some
$i\in[n]$, and $x_i=0$ implies $E_0$ and $x_i=1$ implies $E_1$ (both
with \ps0 proofs),
then to simplify the writing we shall sometime not mention the Boolean axiom at all.
For example, the latter situation can be depicted by saying
that ``if $x_i=0$ implies $E_0$ with a \ps0 proof and
$x_i=1$ implies $E_1$ with a \ps0 proof, then
we can derive $E_0\Or E_1$ with a \ps0 proof''.\FullSpace

\subsection{Basic Counting inside R(lin) and R$^{\mathbf 0}$(lin)}\label{sec-basic-count-R(lin)}

In this subsection we illustrate how to efficiently prove several basic counting
arguments inside R(lin) and \RZ0. This will facilitate us in showing short proofs for
hard tautologies in the sequel.
In accordance with the last paragraph in the previous subsection,
we shall carry the proofs inside R(lin) and \RZ0 with a slightly less rigor.

\begin{lemma}\label{lem-basic-cnt-combine-2-lines-all-possible}
%Let $r,t$ be non-negative integers,
Let $z_1$ abbreviate $\vec a \cd \vec x$ and $z_2$ abbreviate
 $\vec b \cd \vec x$.
Let $D_1$ be $\BigOr_{\alpha\in \mathcal A}(z_1 =\alpha)$
and let $D_2$ be
$\BigOr_{\beta\in\mathcal B}\left(z_2=\beta\right)$,
where $\mathcal A,\mathcal B$ are two (finite) sets of integers.
Then there is a polynomial-size (in the size of $D_1,D_2$) R(lin) proof from
$D_1,D_2$ of:
\begin{equation}\label{eq-basic-cnt-combine-all-possible}
		\BigOr\limits_
			{\alpha\in\mathcal A,
				\beta\in\mathcal B}
				\left(
					z_1+z_2
				  =\alpha+\beta
				\right)
\,.
\end{equation}
Moreover, if $\vec a$ and $\vec b$ consist of \emph{constant} integers
(which means that $D_1,D_2$ are \RZ0-lines),
then there is a polynomial-size (in the size of $D_1,D_2$) \RZ0 proof
of (\ref{eq-basic-cnt-combine-all-possible}) from $D_1,D_2$.
\end{lemma}

\begin{proof}
Denote the elements of $\mathcal A$ by $\alpha_1,\ldots,\alpha_k$.
In case $z_1=\alpha_i$, for some $i\in[k]$ then we can
add $z_1=\alpha_i$ to every equation in
$\BigOr_{\beta\in\mathcal B}\left(z_2=\beta\right)$ to get
$\BigOr_{\beta\in\mathcal B}(z_1+z_2=\alpha_i+\beta)$.
Therefore, there exist $k$\, R(lin) proofs,
each with \ps0 (in $|D_1|$ and $|D_2|$),
of
\begin{eqnarray*}
\BigOr_{\beta\in\mathcal B}\left(z_1+z_2=\alpha_1+\beta\right),\;
\BigOr_{\beta\in\mathcal B}\left(z_1+z_2=\alpha_2+\beta\right),
\; & \ldots & \;,
\BigOr_{\beta\in\mathcal B}\left(z_1+z_2=\alpha_k+\beta\right)
\end{eqnarray*}
from $z_1=\alpha_1$, $z_1=\alpha_2$ \CommaDots $z_1=\alpha_k$, respectively.

Thus, by Lemma \ref{lem-R(lin)-case-analysis}, we can derive
\begin{equation}\label{eq-conclude-combine-lines}
		\BigOr\limits_
			{\alpha\in\mathcal A,
				\beta\in\mathcal B}
				\left(
					z_1+z_2
				  =\alpha+\beta
				\right)
\,
\end{equation}
from $D_1$ and $D_2$ in a polynomial-size (in $|D_1|$ and $|D_2|$) R(lin)-proof.
This concludes the first part of the lemma.
\HalfSpace

Assume that $\vec a$ and $\vec b$ consist of constant coefficients only.
Then by inspecting the R(lin)-proof of (\ref{eq-conclude-combine-lines})
from $D_1$ and $D_2$ demonstrated above
(and by using Lemma \ref{lem-R0(lin)-case-anl} instead of
Lemma \ref{lem-R(lin)-case-analysis}),
one can verify that this proof is in fact carried inside \RZ0.
\end{proof}

An immediate corollary of Lemma \ref{lem-basic-cnt-combine-2-lines-all-possible}
is the  efficient formalization in R(lin) of the following
obvious counting argument:
If a linear form equals some value in the interval (of integer numbers)
$[a_0,a_1]$ and another linear form equals some value in
$[b_0,b_1]$ (for some $a_0\le a_1$ and $b_0\le b_1$),
then their addition equals some value in $[a_0+b_0,a_1+b_1]$.
More formally:
\begin{corollary}\label{lem-basic-count-combine-two-lines}
Let $z_1$ abbreviate $\vec a \cd \vec x$ and $z_2$ abbreviate
 $\vec b \cd \vec x$.
Let $D_1$ be
$(z_1 =a_0) \Or (z_1 =a_0+1) \ldots \Or (z_1=a_1)$,
and let $D_2$ be
$\left(z_2=b_0\right) \Or \left(z_2=b_0+1\right)
	\ldots \Or \left(z_2=b_1\right)$.
Then there is a polynomial-size (in the size of $D_1,D_2$) R(lin) proof from
$D_1,D_2$ of
\begin{equation}\label{eq-basic-cnt-combine}
	\left(
			z_1+z_2
		  =a_0+b_0
	\right)
		\Or
	\left(
			z_1+z_2
			=a_0+b_0+1
	\right)
			\Or \ldots \Or
	\left(
			z_1+z_2
		= a_1+b_1
	\right)
\,.
\end{equation}
Moreover, if $\vec a$ and $\vec b$ consist of \emph{constant} integers
(which means that $D_1,D_2$ are \RZ0-lines),
then there is a polynomial-size (in the size of $D_1,D_2$) \RZ0 proofs
of (\ref{eq-basic-cnt-combine}) from $D_1,D_2$.
\end{corollary}

\begin{lemma}\label{lem-basic-count-all-possibilities}
Let $\vec a \cd\vec x$ be a linear form with $n$ variables, and let
$\mathcal A:=\set{ \vec a \cd \vec x \such \vec x\in\zo^n }$ be the set
of all possible values of $\vec a\cd\vec x$ over Boolean assignments to $\vec x$.
Then there is a polynomial-size,
in the size of the linear form $\vec a\cd\vec x$,\footnotemark ~R(lin) proof of
\footnotetext{Recall that the size of $\vec a \cd \vec x$ is
$\sum_{i=1}^n |a_i|$, that is, the size of the unary representation of $\vec a$.}
\begin{equation}
\label{eq-basic-all-possiblities}
\BigOr\limits_{
				\alpha \in \mathcal A
				 }
				 {
				  (\vec a \cd \vec x = \alpha)
				 }
	\,.
\end{equation}
Moreover, if the coefficients in $\vec a$ are constants, then there is a polynomial-size
(in the size of $\vec a\cd\vec x$) \RZ0 proof of (\ref{eq-basic-all-possiblities}).
\end{lemma}

\begin{proof}
Without loss of generality, assume that all the
coefficients in $\vec a$ are nonzero.
Consider the Boolean axiom $(x_{1}=0)\Or(x_{1}=1)$ and the
(first) coefficient $a_1$ from $\vec a$.
Assume that $a_1\ge 1$.
Add $(x_1=0)$ to itself
$a_1$ times, and arrive at $(a_1 x_{1}=0)\Or(x_{1}=1)$.
Then, in the resulted line,
add $(x_{1}=1)$ to itself $a_1$ times,
until the following is reached:
%\begin{equation}\label{eq-all-possibilities-a1}
\[
(a_1 x_1 =0)\Or (a_1 x_1 = a_1)\,.
\]
%\end{equation}

Similarly, in case $a_1\le -1$
we can subtract ($\abs{a_1}+1$ many times) $(x_1=0)$ from itself in $(x_1=0)\Or(x_1=1)$,
and then subtract ($\abs{a_1}+1$ many times) $(x_1=1)$ from itself in the resulted line.

In the same manner, we can derive the disjunctions:
$(a_2 x_2=0)\Or(a_2 x_2=a_2),\ldots,(a_n x_n=0)\Or(a_n x_n=a_n)$.

Consider $(a_1 x_1 =0)\Or (a_1 x_1 = a_1)$ and $(a_2 x_2=0)\Or(a_2 x_2=a_2)$.
From these two lines,
by Lemma \ref{lem-basic-cnt-combine-2-lines-all-possible},
there is a polynomial-size in $|a_1|+|a_2|$ derivation of:
%(we assume here that both $a_1$ and $a_2$ are positive; the other cases are similar):
\begin{equation}\label{eq-all-possibilities-a1-AND-a2}
(a_1 x_1 + a_2 x_2 =0)
\Or(a_1 x_1 +a_2 x_2 =a_1)
\Or (a_1 x_1 + a_2 x_2 =a_2)
\Or (a_1 x_1 +a_2 x_2 = a_1 + a_2)\,.
\end{equation}
In a similar fashion, now consider $(a_3 x_3=0)\Or(a_3 x_3=a_3)$
and apply again Lemma \ref{lem-basic-cnt-combine-2-lines-all-possible},
to obtain
\begin{equation}\label{eq-all-possibilities-a1-a2-a3}
\BigOr\limits_{
			\alpha\in \mathcal A'
			}
		(a_1 x_1 + a_2 x_2 + a_3 x_3 = \alpha)
\,,
\end{equation}
where $\mathcal A'$ are all possible values to $a_1 x_1 + a_2 x_2 + a_3 x_3$ over
Boolean assignments to $x_1,x_2,x_3$.
The derivation of (\ref{eq-all-possibilities-a1-a2-a3}) is of size polynomial in
$|a_1|+|a_2|+|a_3|$.

Continue to consider, successively, all other lines
$(a_4 x_4=0)\Or(a_4 x_4=a_4),\ldots,(a_n x_n=0)\Or(a_n x_n=a_n)$,
and apply the same reasoning.
Each step uses a derivation of size at most polynomial in $\sum_{i=1}^n{|a_i|}$.
And so overall we reach the desired line (\ref{eq-basic-all-possiblities}),
with a derivation of size polynomial in the size of $\vec a\cd\vec x$.
This concludes the first part of the lemma.
\QuadSpace

Assume that $\vec a$ consists of constant coefficients only.
Then by inspecting the R(lin)-proof demonstrated above
(and by using the second part of Lemma \ref{lem-basic-cnt-combine-2-lines-all-possible}),
one can see that this proof is in fact carried inside \RZ0.
\end{proof}

\begin{lemma}\label{lem-counting-in-R0(lin)}
There is a polynomial-size (in $n$) \RZ0 proof from
\begin{equation}\label{eqR001}
(x_{1} = 1)\Or \cdots \Or(x_{n} = 1)
\end{equation}
of
\begin{equation}\label{eq-basic-cnt-for-PHP}
(x_1+\ldots+x_n=1)\Or\cdots\Or(x_1+\ldots+x_n=n)\,.
\end{equation}
\end{lemma}

\begin{proof}
We show that for every $i\in[n]$, there is a polynomial-size (in $n$) \RZ0
proof from $(x_{i}=1)$ of
$(x_1+\ldots+x_{n}=1)\OrDots(x_1+\ldots+x_{n}=n)$.
This concludes the proof since,
by Lemma~\ref{lem-R0(lin)-case-anl},
we then can derive
from (\ref{eqR001}) (with a \ps0 (in $n$) \RZ0 proof)
the disjunction (\ref{eqR001}) in which
each $(x_{i}=1)$ (for all $i\in[n]$) is replace by
$(x_1+\ldots+x_{n}=1)\OrDots(x_1+\ldots+x_{n}=n)$, which is
precisely the disjunction (\ref{eq-basic-cnt-for-PHP})
(note that (\ref{eq-basic-cnt-for-PHP}) is an \RZ0-line).

\begin{claim}\label{cla-simplified-counting-prf}
For every $i\in[n]$, there is a a polynomial-size (in $n$) \RZ0
proof from $(x_{i}=1)$ of
$(x_1+\ldots+x_{n}=1)\OrDots(x_1+\ldots+x_{n}=n)$.
\end{claim}

\begin{proofclaim}
By Lemma \ref{lem-basic-count-all-possibilities},
for every $i\in[n]$ there is a polynomial-size
(in $n$) \RZ0 proof (using only the Boolean axioms) of
\begin{equation}\label{eq-bsc-cnt-simplify-claim}
	(x_1+\ldots+x_{i-1}+x_{i+1}+\ldots+x_n=0)
		\OrDots
	(x_1+\ldots+x_{i-1}+x_{i+1}+\ldots+x_n=n-1)\,.
\end{equation}
Now add successively $(x_i=1)$ to every equation in
(\ref{eq-bsc-cnt-simplify-claim}) (note that this can be done
in \RZ0).
We obtain precisely
$(x_1+\ldots+x_{n}=1)\OrDots(x_1+\ldots+x_{n}=n)$.
\end{proofclaim}
\end{proof}

\begin{lemma}\label{lem-holes-counting-in-R0(lin)}
There is a polynomial-size (in $n$) \RZ0 proof of $(x_{1}+\ldots+ x_n = 0)\Or
(x_{1}+\ldots+ x_n = 1)$ from the collection of disjunctions consisting of
$(x_i=0)\Or(x_j=0)$, for all $1\le i<j\le n$.
\end{lemma}

\begin{proof}
We proceed by induction on $n$. The base case for $n=1$ is immediate from the
Boolean axiom $(x_1=0)\Or(x_1=1)$. Assume we already have a polynomial-size
proof of
\begin{equation}\label{eqR011}
(x_{1}+\ldots+ x_n = 0)\Or (x_{1}+\ldots+ x_n = 1).
\end{equation}
If $x_{n+1}=0$ we add $x_{n+1}=0$ to both of the equations in (\ref{eqR011}), and reach:
\begin{equation}\label{eqR012}
(x_{1}+\ldots+ x_{n+1} = 0)\Or (x_{1}+\ldots+ x_{n+1} = 1).
\end{equation}
Otherwise, $x_{n+1}=1$, and so we can cut-off
%resolve $(x_{n+1}=1)$ with
$(x_{n+1}=0)$ in all the initial disjunctions
$(x_{i}=0)\Or(x_{n+1}=0)$, for all $1\le i\le n$.
We thus obtain $(x_{1}=0),\ldots,(x_{n}=0)$.
Adding together $(x_{1}=0),\ldots,(x_{n}=0)$ and
$(x_{n+1}=1)$ we arrive at
\begin{equation}\label{eq-reason-hole-like-temp}
(x_{1}+\ldots+ x_{n+1} = 1)\,.
\end{equation}
So overall, either (\ref{eqR012}) holds
or (\ref{eq-reason-hole-like-temp}) holds;
and so (using Lemma \ref{lem-R0(lin)-case-anl}) we arrive at the disjunction of
(\ref{eq-reason-hole-like-temp}) and (\ref{eqR012}),
which is precisely (\ref{eqR012}).
\end{proof}

%====================================================
\section{Implicational Completeness of R(lin) and its Subsystems}
\label{sec-impl-complete}
%====================================================

In this section we provide a proof of the implicational completeness
of R(lin) and its subsystems.
We shall need this property in the sequel (see Section \ref{sec-Tseitin}).
The implicational completeness of a proof system
is a stronger property than mere completeness.
Essentially, a system is implicationally complete if whenever
something is \emph{semantically} implied by a set of initial premises,
then it is also \emph{derivable} from the initial premises.
In contrast to this, mere completeness means that any tautology
(or in case of a refutation system, any unsatisfiable set of initial premises)
has a proof in the system (respectively, a refutation in the system).
As a consequence, the proof of implicational completeness
in this section establishes an alternative completeness proof
to that obtained via simulating resolution
(see Proposition \ref{prop-R(lin)-sim-resolution}).
Note that we are not concerned in this section with the size of the proofs,
but only with their existence.

Recall the definition of the semantic implication relation $\models$
from Section \ref{sec-clauses-of-linear-equations}.
 Formally, we say that R(lin) is \emph{implicationally complete}
if for every collection of disjunctions of linear equations
$D_0, D_1\CommaDots D_m$, it holds that
$D_1\CommaDots D_m \models D_0$ implies that there is
an R(lin) proof of $D_0$ from $D_1\CommaDots D_m$.

\begin{theorem}\label{thm-impl-comp}
R(lin) is implicationally complete.
\end{theorem}

\begin{proof}
We proceed by induction on $n$,
the number of variables $x_1\CommaDots x_n$ in
$D_0, D_1\CommaDots D_m$.\QuadSpace

\emph{The base case $n=0$}. We need to show that
$D_1\CommaDots D_m \models D_0$ implies that there is
an R(lin) proof of $D_0$ from $D_1\CommaDots D_m$, where all $D_i$'s (for
$0\le i\le m$) have no variables but only constants.
This means that each $D_i$ is a disjunction of equations of the form
$(0 = a_0)$ for some integer $a_0$
(if a linear equation have no variables, then
the left hand side of this equation must be $0$; see Section
\ref{sec-clauses-of-linear-equations}).

There are two cases to consider.
In the first case \emph{$D_0$ is satisfiable}. Since $D_0$ has no variables,
this means precisely that $D_0$ is the equation $(0=0)$.
Thus, $D_0$ can be derived easily from any axiom in R(lin) (for instance,
by subtracting each equation in $(x_1=0)\Or(x_1=1)$ from itself, to reach
$(0=0)\Or(0=0)$, which is equal to $(0=0)$, since we discard duplicate equations
inside disjunctions).

In the second case \emph{$D_0$ is unsatisfiable}. %This means that $D_0$
%is a disjunction of unsatisfiable equations
Thus, since $D_1\CommaDots D_m \models D_0$,
there is no assignment satisfying all $D_1\CommaDots D_m$.
Hence, there must be at least one unsatisfiable
disjunction $D_i$ in $D_1\CommaDots D_m$ (as a disjunction with no variables
is either tautological or unsatisfiable).
%(again, since every disjunction in $D_1,\ldots,D_m$ is either
%unsatisfiable or tautological).
Such an unsatisfiable $D_i$ is a disjunction of zero or more unsatisfiable
equations of the form $(0=a_0)$,
for some  integer $a_0\ne 0$.
We can then use Simplification to cut-off all the unsatisfiable equations
in $D_i$ to reach the empty disjunction.
By the Weakening rule, we can now derive $D_0$ from the
empty disjunction.
\smallskip

\emph{The induction step}. Assume that the theorem holds for disjunctions with
$n$ variables. %$x_1\CommaDots x_n$.
Let the underlying variables of $D_0, D_1\CommaDots D_m$ be
$x_1\CommaDots x_{n+1}$, and assume that
\begin{equation}\label{eq-impl-compl-ind-hyp}
D_1\CommaDots D_m \models D_0\,.
\end{equation}
We write the disjunction $D_0$ as:
\begin{equation}\label{eq-impl-compl-general-line}
\BigOr_{j=1}^t\left(
\sum_{i=1}^n a^{(j)}_i x_i + a^{(j)}_{n+1} x_{n+1} = a^{(j)}_0
\right)\,,
\end{equation}
where the $a_i^{(j)}$'s are integer coefficients.
We need to show that there is an R(lin) proof of $D_0$
from $D_1\CommaDots D_m$.

Let $D$ be a disjunction of linear equations,
let $x_i$ be a variable and let $b\in\zo$.
We shall denote by $D\rst_{x_i=b}$ the
disjunction $D$,
where in every equation in $D$
the variable $x_i$ is substituted by $b$,
and the constant terms in the left hand sides of all resulting
equations (after substituting $b$ for $x_i$) switch sides (and change signs, obviously)
to the right
hand sides of the equations (we have to switch sides of constant terms, as
by definition linear equations in R(lin) proofs have all constant terms appearing only
on the right hand sides of equations).

We now reason (slightly) informally inside R(lin)
(as illustrated in Section \ref{sec-case-anl}).
Fix some $b\in\zo$, and assume that $x_{n+1}=b$.
Then, from $D_1\CommaDots D_m$ we can derive (inside R(lin)):
\begin{equation}\label{eq-impl-compl-subs}
D_1\rst_{x_{n+1}=b} \CommaDots D_m\rst_{x_{n+1}=b}\,.
\end{equation}
The only variables occurring in (\ref{eq-impl-compl-subs}) are
$x_1\CommaDots x_n$.
From assumption (\ref{eq-impl-compl-ind-hyp}) we clearly have
$D_1\rst_{x_{n+1}=b} \CommaDots D_m\rst_{x_{n+1}=b}
\,\models D_0\rst_{x_{n+1}=b}$.
And so by the induction hypothesis there is an R(lin) derivation of
$D_0\rst_{x_{n+1}=b}$
from $D_1\rst_{x_{n+1}=b} \CommaDots D_m\rst_{x_{n+1}=b}$.
So overall, assuming that $x_{n+1}=b$,
there is an R(lin) derivation of
$D_0\rst_{x_{n+1}=b}$ from $D_1\CommaDots D_m$.

We now consider the two possible cases: $x_{n+1}=0$ and $x_{n+1}=1$.

\emph{In case $x_{n+1}=0$}, by the above discussion,
we can derive $D_0\rst_{x_{n+1}=0}$ from $D_1\CommaDots D_m$.
For every $j\in[t]$, add successively ($a^{(j)}_{n+1}$ times) the
equation $x_{n+1}=0$ to the $j$th equation in $D_0\rst_{x_{n+1}=0}$
(see (\ref{eq-impl-compl-general-line})).
We thus obtain precisely $D_0$.

\emph{In case $x_{n+1}=1$}, again, by the above discussion,
we can derive $D_0\rst_{x_{n+1}=1}$ from $D_1\CommaDots D_m$.
For every $j\in[t]$, add successively ($a^{(j)}_{n+1}$ times) the
equation $x_{n+1}=1$ to the $j$th equation in $D_0\rst_{x_{n+1}=1}$
(recall that we switch sides of constant terms in every linear equation
after the substitution of $x_{n+1}$ by $1$ is performed in $D_0\rst_{x_{n+1}=1}$).
Again, we obtain precisely $D_0$.
\end{proof}

By inspecting the proof of Theorem \ref{thm-impl-comp},
it is possible to verify that if all the disjunctions $D_0,\CommaDots D_m$
are \RZ0-lines (see Definition
\ref{def-RZ0-clause}), then the proof of $D_0$ in R(lin) uses only
\RZ0-lines as well.
Therefore, we have:

\begin{corollary}\label{cor-impl-compl-RZ0}
\RZ0 is implicationally complete.
\end{corollary}

\begin{remark}
Corollary \ref{cor-impl-compl-RZ0} states that any \RZ0-line that
is semantically implied by a set of initial \RZ0-lines, is in fact
derivable in \RZ0 from the initial \RZ0-lines.
On the other hand, it is possible that a certain proof of the same
\RZ0-line inside R(lin) will be significantly shorter
than the proof inside \RZ0. %(and the same goes to \RZ0 proofs versus
%\RC0 proofs, and \RC0 proofs versus R(lin) proofs).
Indeed, we shall see in Section \ref{sec-lower-bounds} that for certain
CNF formulas R(lin) has a super-polynomial speed-up over \RZ0.
\end{remark}

%==========================================================
\section{Short Proofs for Hard Tautologies}\label{sec-hard}
%==========================================================

In this section we show that \RZ0 is already enough to admit small
proofs for ``hard" counting principles like the pigeonhole principle and
the Tseitin graph formulas for constant degree graphs.
On the other hand, as we shall see in Section \ref{sec-lower-bounds},
\RZ0 inherits the same weakness that cutting planes proofs have
with respect to the clique-coloring tautologies.
 Nevertheless, we can efficiently prove the
clique-coloring principle in (the stronger system) R(lin),
but not by using R(lin) ``ability to count",
rather by using its (straightforward) ability to simulate Res(2) proofs (that is,
resolution proofs extended to operate with 2-DNF formulas, instead of clauses).

%====================================================
\subsection{The Pigeonhole Principle Tautologies in R$^{\mathbf 0}$(lin)}\label{sec-php}
%====================================================

This subsection illustrates \ps0 \RZ0 proofs of the pigeonhole principle.
This will allow us to establish \ps0 \ml0 proofs operating with
depth-$3$ \ml0 formulas of the pigeonhole principle (in Section \ref{sec-multilinear}).

The \emph{$m$ to $n$ pigeonhole principle} states that $m$ pigeons cannot be mapped
one-to-one into $n<m$ holes. The negation of the pigeonhole principle,
denoted $\neg$PHP$_n^m$, is formulated as an unsatisfiable CNF formula as follows
(where clauses are translated to disjunctions of linear equations):

\begin{definition}\label{def-PHP}
The $\neg$PHP$_n^m$ is the following set of clauses:
\begin{enumerate}
\item Pigeons axioms:\;\; $(x_{i,1} = 1)\Or \cdots \Or(x_{i,n} = 1)$,\; for all $1\le i \le
m$;
\item Holes axioms: \;\; $(x_{i,k}=0)\Or(x_{j,k}=0)$,\;\; for all $1\le i < j\le m$
and for all $1\le k\le
n$.
\end{enumerate}
The intended meaning of each propositional variable $x_{i,j}$
is that the $i$th pigeon is mapped to the $j$th hole.
\end{definition}

We now describe a polynomial-size in $n$ refutation of $\neg$PHP$_n^m$ inside
\RZ0.
For this purpose it is sufficient to prove a polynomial-size refutation of the
pigeonhole principle when the number of pigeons $m$ equals $n+1$ (because the
set of clauses pertaining to $\neg$PHP$_n^{n+1}$ is already contained in the set of
clauses pertaining to $\neg$PHP$_n^m$, for any $m>n$). Thus, we fix $m=n+1$.
In this subsection we shall say a proof in \RZ0
is of \emph{polynomial-size}, always intending
\emph{polynomial-size in $n$} (unless otherwise stated).

By Lemma \ref{lem-counting-in-R0(lin)}, for all $i\in[m]$ we can derive from the
Pigeon axiom (for the $i$th pigeon):
\begin{equation}\label{eq-PHP-first-sum}
(x_{i,1}+\ldots+x_{i,n}=1)\Or\cdots\Or(x_{i,1}+\ldots+x_{i,n}=n)
\end{equation}
with a polynomial-size \RZ0 proof.

By Lemma \ref{lem-holes-counting-in-R0(lin)}, from the Hole axioms we can
derive, with a polynomial-size \RZ0 proof
\begin{equation}\label{eq-PHP-second-sum}
(x_{1,j}+\ldots+ x_{m,j} = 0)\Or (x_{1,j}+\ldots+ x_{m,j} = 1),
\end{equation}
for all $j\in[n]$.

Let $S$ abbreviate the sum of all formal variables $x_{i,j}$. In other words,
$$S:=\sum\limits_{i\in[m],j\in[n]}{x_{i,j}}\,.$$

\begin{lemma}\label{lem-PHP-first-sum}
There is a polynomial-size \RZ0 proof from (\ref{eq-PHP-first-sum}) (for all $i\in[m]$)
 of
$$(S=m)\Or(S=m+1)\cdots\Or(S=m\cd n).$$
\end{lemma}

\begin{proof}
For every $i\in[m]$ fix the abbreviation $z_i:=x_{i,1}\PlusDots x_{i,n}$. Thus,
by (\ref{eq-PHP-first-sum}) we have $(z_{i}=1)\Or\cdots\Or(z_{i}=n)$.

Consider $(z_{1}=1)\Or\cdots\Or(z_{1}=n)$ and $(z_{2}=1)\Or\cdots\Or(z_{2}=n)$.
By Corollary \ref{lem-basic-count-combine-two-lines},
we can derive from these two lines
\begin{equation}\label{eq-PHP-first-sum1}
(z_1+z_2=2) \Or(z_1+z_2=3)\OrDots (z_1+z_2=2n)
\end{equation}
with a polynomial-size \RZ0 proof.

Now, consider $(z_3=1)\OrDots(z_3=n)$ and (\ref{eq-PHP-first-sum1}).
By Corollary \ref{lem-basic-count-combine-two-lines} again,
from these two lines we can derive with a \ps0 \RZ0 proof:
\begin{equation}
(z_1+z_2+z_3=3)\Or(z_1+z_2+z_3=4) \OrDots (z_1+z_2+z_3=3n)
\,.
\end{equation}

Continuing in the same way, we eventually arrive at
\[
	(z_1\PlusDots z_m=m) \Or(z_1\PlusDots z_m=m+1)\OrDots (z_1\PlusDots z_m=m\cd n)\,,
\]
which concludes the proof,
since $S$ equals $z_1\PlusDots z_m$.
\end{proof}

\begin{lemma}\label{lem-PHP-second-sum}
There is a polynomial-size \RZ0 proof from (\ref{eq-PHP-second-sum}) of
$$(S=0)\Or\cdots\Or(S=n).$$
\end{lemma}

\begin{proof}
For all $j\in[n]$, fix the abbreviation $y_j:=x_{1,j}\PlusDots x_{m,j}$. Thus,
by (\ref{eq-PHP-second-sum}) we have $(y_{j}=0)\Or(y_{j}=1)$, for all $j\in[n]$.
Now the proof is similar to the proof of Lemma \ref{lem-basic-count-all-possibilities},
except that here single variables are abbreviations of linear forms.

If $y_1=0$ then we can add $y_1$ to the two sums in $(y_2=0)\Or(y_2=1)$, and
reach $(y_1+y_2=0)\Or(y_1+y_2=1)$ and if $y_1=1$ we can do the same and reach
$(y_1+y_2=1)\Or(y_1+y_2=2)$. So, by Lemma \ref{lem-R0(lin)-case-anl}, we
can derive with a polynomial-size \RZ0 proof
\begin{equation}\label{eq-PHP-three-cases}
(y_1+y_2=0)\Or(y_1+y_2=1)\Or(y_1+y_2=2)\,.
\end{equation}

Now, we consider the three cases in (\ref{eq-PHP-three-cases}):
$y_1+y_2=0$ or $y_1+y_2=1$  or $y_1+y_2=2$,
and the clause $(y_3=0)\Or(y_3=1)$. We arrive in a similar manner at
$(y_1+y_2+y_3=0)\OrDots(y_1+y_2+y_3=3)$. We continue in the same way
until we arrive at $(S=0)\Or\cdots\Or(S=n)$.
\end{proof}

\begin{theorem}\label{thmPHP}
There is a polynomial-size \RZ0 refutation of the $m$ to $n$ pigeonhole principle
$\neg$PHP$^m_n$.
\end{theorem}
\begin{proof}
By Lemmas \ref{lem-PHP-first-sum} and \ref{lem-PHP-second-sum} above,
all we need is to
show a polynomial-size refutation of $(S=m)\Or\cdots\Or(S=m\cd n)$ and
$(S=0)\Or\cdots\Or(S=n)$.

Since $n<m$, for all $0\le k\le n$, if $S=k$ then using the Resolution and
Simplification rules we can cut-off all the sums in
$(S=m)\Or\cdots\Or(S=m\cd n)$ and arrive at the empty clause.
Thus, by Lemma \ref{lem-R0(lin)-case-anl},
 there is a polynomial-size \RZ0 proof of the
empty clause from $(S=0)\Or\cdots\Or(S=n)$ and $(S=m)\Or\cdots\Or(S=m\cd n)$.
\end{proof}
\FullSpace

%====================================================
\subsection{Tseitin mod $p$ Tautologies in R$^{\mathbf 0}$(lin)}\label{sec-Tseitin}
%====================================================

This subsection establishes polynomial-size \RZ0 proofs of Tseitin
graph tautologies (for constant degree graphs).
%Our proofs here are different and more transparent than the multilinear proofs
%illustrated in \cite{RT05} for the same tautologies.
This will allow us (in  Section \ref{sec-multilinear})
to extend the multilinear proofs of the Tseitin mod $p$ tautologies
to any field of characteristic $0$
(the proofs in \cite{RT05}
required working over a field containing a primitive $p$th root of unity
when proving the Tseitin mod $p$ tautologies; for more details
see Section \ref{sec-multilinear}).

Tseitin mod $p$ tautologies (introduced in \cite{BGIP01})
are generalizations of the (original, mod $2$)
Tseitin graph tautologies (introduced in \cite{Tse68}).
To build the intuition for the generalized version,
we start by describing the (original) Tseitin mod $2$ principle.
Let $G=(V,E)$ be a connected undirected graph with an \emph{odd} number
of vertices $n$.
The Tseitin mod $2$ tautology states that there is no sub-graph $G'=(V, E')$,
where $E'\subseteq E$,
so that for \emph{every} vertex $v\in V$,
the number of edges from $E'$ incident to $v$ is odd.
This statement is valid, since otherwise,
summing the degrees of all the vertices in $G'$ would amount to an odd number
(since $n$ is odd), whereas this sum also counts every edge in $E'$ twice, and
so is even.

As mentioned above, the Tseitin mod $2$ principle was generalized
by Buss \emph{et al.} \cite{BGIP01} to obtain the Tseitin mod $p$ principle.
Let $p\geq 2$ be some fixed integer and let $G=(V,E)$ be
a connected undirected $r$-regular graph with $n$
vertices and no double edges. Let $G'=(V,E')$ be the corresponding
\emph{directed} graph that results from $G$ by replacing every (undirected)
edge in $G$ with two opposite directed edges.
Assume that $n\equiv 1$ (mod $p$).
Then, the Tseitin mod $p$ principle states that
there is no way to assign to every edge in $E'$ a value from
$\set{0, \ldots, p-1}$, so that:
\begin{description}
  \item[(i)] For every pair of opposite directed edges $e,
  \bar{e}$ in $E'$, with assigned values $a, b$,
respectively, $a+b \equiv 0$ (mod $p$); and
  \item[(ii)] For every vertex $v$ in $V$, the sum of the values assigned to the
  edges in $E'$ coming out of $v$ is congruent to $1$ (mod $p$).
\end{description}

The Tseitin mod $p$ principle is valid,
since if we sum the values assigned to all edges of $E'$
in pairs we obtain  $0$ (mod $p$) (by (i)), where summing them by vertices we
arrive at a total value of $1$ (mod $p$)
(by (ii) and since $n\equiv 1$ (mod $p$)).
We shall see in what follows, that this simple counting argument can be
carried on in a natural (and efficient) way already inside \RZ0.

As an unsatisfiable propositional formula (in CNF form) the negation of
the Tseitin mod $p$ principle is formulated by assigning a variable
$x_{e,i}$ for every edge $e\in E'$ and every residue $i$ modulo $p$.
The variable $x_{e,i}$ is an indicator variable for the
fact that the edge $e$ has an associated value $\,i$.
The following are the clauses of the Tseitin mod $p$ \,CNF formula
(as translated to disjunctions of linear equations).

\begin{definition}[\textbf{Tseitin mod $p$ formulas} (\Tse0)]
\label{def-Tse}
Let $p\geq 2$ be some fixed integer and let  $G=(V,E)$ be a connected
undirected $r$-regular graph with $n$ vertices and no double edges, and assume
that $\,n\equiv 1$ ({\rm{mod}} $p$). Let $G'=(V,E')$ be the corresponding
directed graph that results from $G$ by replacing every (undirected) edge in
$G$ with two opposite directed edges.

Given a vertex $v\in V$, denote the edges in $E'$ coming
out of $v$ by $e[v,1], \ldots, e[v,r]$
and define the following set of (translation of) clauses:
\[
{\rm{MOD}}_{{p,1}}(v)\!:=\left\{{\BigOr\limits_{k = 1}^r { (x_{e[v,k] ,i_k }=0) }\;
{\biggl|\biggr.}\; i_1 , \ldots , i_r  \in \{ 0, \ldots ,p - 1\} {\mbox{ and
}}\sum\limits_{k = 1}^r {i_k \not\equiv 1\!\!\mod p}}\right\}.
\]
The {\rm Tseitin mod $p$} formula, denoted  {\rm \Tse0},
consists of the following (translation) of clauses:
\[
\begin{array}{l}
 {\mbox{1.}}\,\BigOr\limits_{i = 0}^{p - 1} {(x_{e,i}=1)} \,{\mbox{,  for all }}e \in E'\;
 %\vspace{3pt}
 \\
 {\mbox{(expresses that every edge is assigned at least one value from $0,\ldots,p-1$);}}
  %\vspace{6pt}
  \\
 {\mbox{2.}}\; (x_{e,i}=0) \Or (x_{e,j}=0)\,{\mbox{,   for all  }}i \ne j \in
 \{ 0, \ldots ,p - 1\}
{\mbox{ and all }}\,
 e\in E' %\vspace{3pt}
 \\ {\mbox{(expresses that every edge is assigned at most one value
from
 $0,\ldots,p-1$);}} %\vspace{6pt}
 \\
{\mbox{3.}}\; (x_{e,i}=1)\Or (x_{\bar e,p - i}=0)\;
 {\mbox{     and     }}\; (x_{e,i}=0)\Or(x_{\bar e,p - i}=1),\footnotemark
  %\vspace{3pt}
  \\ \qquad{\mbox{    for all two opposite directed edges }}
  e,\bar e \in E' {\mbox{ and all }}i \in \{ 0, \ldots ,p - 1\}
   %\vspace{3pt}
   \\ {\mbox{    (expresses condition (i) of the Tseitin mod $p$ principle
above); }}
  %\vspace{6pt}
  \\
 {\mbox{4.}}\,{\mbox{\rm{ MOD}}}_{p,1} (v)\,,{\mbox{   for all }}v \in V %\vspace{3pt}
 \\
 {\mbox{(expresses condition (ii) of the Tseitin mod $p$ principle above).}}
 \end{array}
\]
\footnotetext{If $i=0$ then $x_{\bar e,p - i}$ denotes $x_{\bar e,0}$.}
\end{definition}
Note that for every edge $e\in E'$, the polynomials of (1,2) in Definition
\ref{def-Tse}, combined with the Boolean axioms of \RZ0,
force any collection of edge-variables $x_{e,0}, \ldots, x_{e,p-1}\,$
to contain exactly one $i\in\set{0, \ldots, p-1}$ so that $x_{e,i}=1$.
%This corresponds to edge $e$ assigned %the value $i$.
Also, it is easy to verify that, given a vertex $v\in V$, any assignment
$\sigma$ of $0,1$ values (to the relevant variables) satisfies both the disjunctions
of (1,2) and the disjunctions of MOD$_{p,1}(v)$ if and only if $\sigma$ corresponds
to an assignment of values from $\set{0, \ldots, p-1}$ to the edges coming out
of $v$ that sums up to $1$ (mod $p$).\QuadSpace

Until the rest of this subsection we fix
an integer $p\geq 2$ and a connected undirected $r$-regular graph
$G=(V,E)$ with $n$ vertices and no double edges,
such that $n\equiv 1 \mod p$ ~and $r$ is a constant.
As in Definition \ref{def-Tse},
we let $G'=(V,E')$ be the corresponding directed graph that results
from $G$ by replacing every (undirected) edge in $G$ with two opposite directed edges.
We now proceed to refute \Tse0 inside \RZ0 with a polynomial-size (in $n$)
refutation.

Given a vertex $v\in V$, and the edges in $E'$ coming
out of $v$, denoted $e[v,1], \ldots, e[v,r]$,
 define the following abbreviation:
\begin{equation}\label{eq-Tse-abbreviate}
\alpha_v := \sum_{j=1}^r \sum_{i=0}^{p-1} i\cd x_{e[v,j], i} \,.
\end{equation}

\begin{lemma}\label{lem-Tse-impl-compl}
Let $v\in V$ be any vertex in $G'$.
Then there is a constant-size \RZ0 proof from \Tse0
of the following disjunction:
\begin{equation}\label{eq-Tse-impl-compl}
\BigOr_{\ell=0}^{r-1}(\alpha_v = 1+\ell\cd p)\,.
\end{equation}
\end{lemma}

\begin{proof}
Let $T_v\subseteq \Tse0$ be the set of all disjunctions of the form
(1,2,4) from Definition \ref{def-Tse} that contain only variables pertaining to
vertex $v$ (that is, all the variables $x_{e,i}$, where $e\in E'$ is an edge
coming out of $v$, and $i\in\set{0\CommaDots p-1}$).

\begin{claim}\label{cla-Tse-impl-compl}
$T_v$ semantically implies (\ref{eq-Tse-impl-compl}), that is:\footnotemark
$$T_v \models \BigOr_{\ell=0}^{r-1}(\alpha_v = 1+\ell\cd p)\,.$$
\end{claim}
\footnotetext{Recall that we only consider assignments of $0,1$ values to variables
when considering the semantic implication relation $\models$.}
\begin{proofclaim}
Let $\sigma$ be an assignment of $0,1$ values to the variables in $T_v$
that satisfies both the disjunctions of (1,2) and
the disjunctions of MOD$_{p,1}(v)$ in Definition \ref{def-Tse}.
 As mentioned above (the comment after Definition \ref{def-Tse}),
such a $\sigma$ corresponds to an assignment of values from $\set{0, \ldots, p-1}$
to the edges coming out of $v$, that sums up to $1 \mod p$.
 This means precisely that $\alpha_v = 1 \mod p$ under the assignment $\sigma$.
Thus, there exists a nonnegative integer $k$,
such that $\alpha_v=1+kp$ under $\sigma$.

It remains to show that $k\le r-1$
(and so the only possible values that $\alpha_v$ can get under $\sigma$
are $1,1+p,1+2p,\ldots,1+(r-1)p$).
Note that because $\sigma$ gives the value $1$ to only
one variable from $x_{e[v,j],0},\ldots,x_{e[v,j],p-1}$ (for every $j\in[r]$),
then the maximal value that $\alpha_v$ can have under $\sigma$ is $r(p-1)$.
Thus, $1+kp\le rp-r$ and so $k\le r-1$.
\end{proofclaim}

From Claim \ref{cla-Tse-impl-compl} and from the implicational completeness
of \RZ0 (Corollary \ref{cor-impl-compl-RZ0}),
there exists an \RZ0 derivation of (\ref{eq-Tse-impl-compl}) from $T_v$.
It remains to show that this derivation is of constant-size.

Since the degree $r$ of $G'$ and the modulus $p$ are both constants,
both $T_v$ and (\ref{eq-Tse-impl-compl}) have constant number of variables and
constant coefficients (including the free-terms).
Thus, %by Corollary \ref{cor-impl-comp-constant-size-proof}
there is a constant-size \RZ0 derivation of (\ref{eq-Tse-impl-compl})
from $T_v$.
\end{proof}

%We now add successively all the equations pertaining to
%disjunctions (\ref{eq-Tse-impl-compl}), for all vertices $v\in V$.
%This is done in the next lemma.

\begin{lemma}\label{lem-Tse-sum1}
There is a \ps0 (in $n$) \RZ0 derivation from \Tse0 of the following disjunction:
\begin{equation*}%\label{eq-Tse-sum1}
\BigOr_{\ell=0}^{(r-1)\cd n}\left(\sum_{v\in V}\alpha_v = n+\ell\cd p\right)\,.
\end{equation*}
\end{lemma}

\begin{proof}
Simply add successively all the equations pertaining to
disjunctions (\ref{eq-Tse-impl-compl}), for all vertices $v\in V$.
Formally, we show that for every subset of vertices ${\mathcal V}\subseteq V$,
with $|{\mathcal V}|=k$, there is a \ps0 (in $n$) \RZ0 derivation from \Tse0 of
\begin{equation}\label{eq-Tse-sum1-ind}
\BigOr_{\ell=0}^{(r-1)\cd k}\left(\sum_{v\in{\mathcal V}}
\alpha_v = k+\ell\cd p\right)\,,
\end{equation}
and so putting ${\mathcal V}=V$, will conclude the proof.

We proceed by induction on the size of $\mathcal V$.
The base case, $|{\mathcal V}|=1$, is immediate from Lemma \ref{lem-Tse-impl-compl}.

Assume that we already derived (\ref{eq-Tse-sum1-ind}) with a \ps0 (in $n$)
\RZ0 proof,
for some ${\mathcal V}\subset V$,
such that $|{\mathcal V}|=k<n$.
Let $u\in V \sm {\mathcal V}$.
By Lemma \ref{lem-Tse-impl-compl}, we can derive
\begin{equation}\label{eq-Tse-one-vertex}
\BigOr_{\ell=0}^{r-1}(\alpha_u = 1+\ell\cd p)
\end{equation}
from \Tse0 with a constant-size proof.
Now, by Lemma \ref{lem-basic-cnt-combine-2-lines-all-possible},
each linear equation in (\ref{eq-Tse-one-vertex}) can be added
to each linear equation in (\ref{eq-Tse-sum1-ind}),
with a polynomial-size (in $n$) \RZ0 proof.
This results in the following disjunction:
\[
\BigOr_{\ell=0}^{(r-1)\cd(k+1)}
\left(\sum_{v\in{\mathcal V}\cup\set{u}}\alpha_v = k+1+\ell\cd p\right)\,,
\]
which is precisely what we need to conclude the induction step.
\end{proof}

\begin{lemma}\label{lem-Tse-sum2}
Let $e,\bar e$ be any pair of opposite directed edges in $G'$
and let $i\in\set{0\CommaDots p-1}$.
Let $T_e\subseteq \Tse0$ be the set of all disjunctions of the form
(1,2,3) from Definition \ref{def-Tse} that contain only
variables pertaining to edges $e,\bar e$
(that is, all the variables $x_{e,j}, x_{\bar e,j}$,
for all $j\in\set{0\CommaDots p-1}$).
Then, there is a constant-size \RZ0 proof from $T_e$
of the following disjunction:
\begin{equation}\label{eq-Tse-sum2}
\left(i\cd x_{e,i}+(p-i)\cd x_{\bar e, p-i} = 0\right) \Or
\left(i\cd x_{e,i}+(p-i)\cd x_{\bar e, p-i} = p \right)\,.
\end{equation}
\end{lemma}

\begin{proof}
First note that $T_e$ %(in fact, those disjunction in $T_e$ that are taken
%from axioms (1,2) from Definition \ref{def-Tse})
semantically implies
\begin{equation}\label{eq-Tse-opposite-edges}
(x_{e,i}+x_{\bar e,p-i} = 0) \Or (x_{e,i}+x_{\bar e,p-i} = 2)\,.
\end{equation}
The number of variables in $T_e$ and (\ref{eq-Tse-opposite-edges}) is constant.
Hence, %by Corollary \ref{cor-impl-comp-constant-size-proof}
there is a constant-size \RZ0-proof of (\ref{eq-Tse-sum2}) from $T_e$.
Also note that
\begin{equation}\label{eq-Tse-sum2-LHS}
\begin{array}{l}
(x_{e,i}+x_{\bar e,p-i} = 0) \Or (x_{e,i}+x_{\bar e,p-i} = 2)
\models \\
\qquad\qquad\qquad\qquad\qquad\left(i\cd x_{e,i}+(p-i)\cd x_{\bar e, p-i} = 0\right) \Or
\left(i\cd x_{e,i}+(p-i)\cd x_{\bar e,p-i} = p \right)\,.
\end{array}
\end{equation}
Therefore, there is also an \RZ0-proof of constant-size
%(again, by Corollary \ref{cor-impl-comp-constant-size-proof})
from $T_e$ of the lower line in
(\ref{eq-Tse-sum2-LHS}).
\end{proof}

We are now ready to complete the \ps0 \RZ0 refutation of \Tse0.
Using the two prior lemmas, the refutation idea is simple, as we now explain.
Observe that
\begin{equation}\label{eq-Tse-sum1=sum2}
\sum_{v\in V}\alpha_v
= \sum_{\{e,\bar e\}\subseteq E' \atop i\in\set{0\CommaDots p-1}}
 \left(i\cd x_{e,i}+(p-i)\cd x_{\bar e, p-i}\right) \,,
\end{equation}
where by $\{e,\bar e\}\subseteq E'$ we mean that $e,\bar e$ is
pair of opposite directed edges in $G'$.

Derive by Lemma \ref{lem-Tse-sum1} the disjunction
\begin{equation}\label{eq-Tse-sum1}
\BigOr_{\ell=0}^{(r-1)\cd n}\left(\sum_{v\in V}\alpha_v = n+\ell\cd p\right)\,.
\end{equation}
This disjunction expresses the fact that $\sum_{v\in V}\alpha_v = 1\mod p$
(since $n=1\mod p$).
 On the other hand, using Lemma \ref{lem-Tse-sum2},
we can ``sum together" all the equations (\ref{eq-Tse-sum2})
(for all $\{e,\bar e\}\subseteq E'$ and all $i\in\set{0,\ldots,p-1}$),
to obtain a disjunction expressing the statement that
\[
	\sum_{
		\{e,\bar e\}\subseteq E' \atop i\in
										\set{
											0\CommaDots p-1
											}
		}
		(i\cd x_{e,i}+(p-i)\cd x_{\bar e, p-i})
												= 0\mod p
\,.
\]
By Equation (\ref{eq-Tse-sum1=sum2}), we then obtain the desired contradiction.
This idea is formalized in the proof of the following theorem:

\begin{theorem}\label{thm-Tse}
Let $G=(V,E)$ be an $r$-regular graph with $n$ vertices, where $r$ is a constant.
Fix some modulus $p$.
Then, there are \ps0 (in $n$) \RZ0 refutations of \Tse0.
\end{theorem}

\begin{proof}
First, use Lemma \ref{lem-Tse-sum1} to derive
\begin{equation}\label{eq-Tse-sum1-thm}
\BigOr_{\ell=0}^{(r-1)\cd n}\left(\sum_{v\in V}\alpha_v = n+\ell\cd p\right)\,.
\end{equation}
Second, use Lemma \ref{lem-Tse-sum2} to derive
\begin{equation}\label{eq-Tse-sum2-thm}
\left(i\cd x_{e,i}+(p-i)\cd x_{\bar e, p-i} = p\right) \Or
\left(i\cd x_{e,i}+(p-i)\cd x_{\bar e,p-i} = 0 \right)\,,
\end{equation}
for every pair of opposite directed edges in $G'=(V,E')$ (as in Definition
\ref{def-Tse}) and every
residue $i\in\set{0\CommaDots p-1}$.

We now reason inside \RZ0. Pick a pair of opposite directed edges
$e,\bar e$ and a residue $i\in\set{0\CommaDots p-1}$.
If $i\cd x_{e,i}+(p-i)\cd x_{\bar e,p-i} = 0$, then
subtract this equation successively from every equation in
(\ref{eq-Tse-sum1-thm}).
We thus obtain a new disjunction, similar to that of (\ref{eq-Tse-sum1-thm}),
but which does not contain the $x_{e,i}$ and $x_{\bar e,p-i}$ variables,
and with the same free-terms.

Otherwise, $i\cd x_{e,i}+(p-i)\cd x_{\bar e,p-i} = p$,
then subtract this equation successively from every equation in
(\ref{eq-Tse-sum1-thm}).
Again, we obtain a new disjunction, similar to that of (\ref{eq-Tse-sum1-thm}),
but which does not contain the $x_{e,i}$ and $x_{\bar e,p-i}$ variables,
and such that $p$ is subtracted from every free-term in every equation.
Since, by assumption, $n\equiv 1\mod p$, the free-terms in every equation are
(still) equal $1\mod p$.

So overall, in both cases ($i\cd x_{e,i}+(p-i)\cd x_{\bar e,p-i} = 0$ and
$i\cd x_{e,i}+(p-i)\cd x_{\bar e,p-i} = p$) we obtained a new
disjunction with all the free-terms in equations equal $1\mod p$.

We now continue the same process for every pair $e,\bar e$
of opposite directed edges in $G'$ and every residue $i$.
Eventually, we discard all the variables $x_{e,i}$ in the
equations, for every $e\in E'$ and $i\in\set{0\CommaDots p-1}$,
while all the free-terms in every equation remain to be equal $1\mod p$.
 Therefore, we arrive at a disjunction of equations of the form
$(0=\gamma)$ for some $\gamma = 1\mod p$. By using the Simplification
rule we can cut-off all such equations, and arrive finally at the
empty disjunction.
\end{proof}

%====================================================
\subsection{The Clique-Coloring Principle in R(lin)}
%====================================================
\label{sec-clique}

In this section we observe that there are polynomial-size R(lin)
proofs of the clique-coloring principle (for certain, weak, parameters).
This implies, in particular, that R(lin) does not possess the feasible monotone
interpolation property (see more details on the interpolation method in Section
\ref{sec-interpo}).

Atserias, Bonet \& Esteban \cite{ABE02} demonstrated
polynomial-size \Res 2 refutations of the clique-coloring formulas
(for certain weak parameters; Theorem \ref{thm-ABE02}).
Thus, it is sufficient to show that R(lin) polynomially-simulates \Res 2 proofs
(Proposition \ref{prop-R(lin)-sim-Res2}).
This can be shown in a straightforward manner.
As noted in the first paragraph of Section \ref{sec-hard},
because the proofs of the clique-coloring formula we discuss here
only follow the proofs inside \Res 2,
then in fact these proofs do not take any advantage of the capacity
``to count'' inside R(lin) (this capacity is exemplified, for instance,
 in Section \ref{sec-basic-count-R(lin)}).

We start with the clique-coloring formulas (these formulas will also be used
in Section \ref{sec-lower-bounds}).
These formulas express the clique-coloring principle that
has been widely used
in the proof complexity literature (cf.,
\cite{BPR97}, \cite{Pud97}, \cite{Kra97-Interpolation},
\cite{Kra98-Discretely}, \cite{ABE02}, \cite{Kra07}).
This principle is based on the following basic combinatorial idea.
Let $G=(V,E)$ be an undirected graph with $n$ vertices and let $k'<k$ be
two integers.
Then, one of the following must hold:
\begin{description}
\item[(i)] The graph $G$ does not contain a \emph{clique with $k$ vertices};

\item[(ii)] The graph $G$ is not a \emph{complete $k'$-partite graph}.
In other words, there is no way to partition $G$ into $k'$ subgraphs
$G_1\CommaDots G_{k'}$,
such that every $G_i$ is an independent set, and for all $i\ne j \in[k']$,
all the vertices in $G_i$ are connected by edges (in $E$) to all the vertices in $G_j$.
\end{description}

Obviously, if Item (ii) above is false
(that is, if $G$ is a complete $k'$-partite graph),
then there exists a $k'$-coloring of
the vertices of $G$; hence the name \emph{clique-coloring} for the principle.

The propositional formulation of the (negation of the)
clique-coloring principle is as follows.
Each variable $\vr p i j$, for all $i\ne j \in [n]$,
is an indicator variable for the fact that there is an edge in $G$
between vertex $i$ and vertex $j$.
Each variable $\vr q \ell i$, for all $\ell\in[k]$ and all $i\in[n]$,
is an indicator variable for the fact that the vertex $i$ in $G$
is the $\ell$th vertex in the $k$-clique.
Each variable $\vr r \ell i$, for all $\ell\in[k']$ and all $i\in[n]$,
is an indicator variable for the fact that the vertex $i$ in $G$
pertains to the independent set $G_\ell$.

\begin{definition}\label{def-clique-color}
The negation of the clique-coloring principle consists of the following
unsatisfiable collection of clauses (as translated to disjunctions of
linear equations), denoted \clique n k {k'}:
\renewcommand{\theenumi}{\roman{enumi}}   % change temporarily to Roman items
\begin{enumerate}
\item $(q_{\ell,1}=1)\OrDots (q_{\ell,n}=1), \mbox{ for all } \ell\in[k]$

(expresses that there exists at least one vertex in $G$ which constitutes
the $\ell$th vertex of the $k$-clique); \label{eq-clique-q-pigeons}%\vspace{-6pt}

\item $(q_{\ell,i}=0)\Or(q_{\ell,j}=0),  \mbox{ for all } i\ne j \in [n],\;
	  \ell\in[k]$

(expresses that there exists at most one vertex in $G$ which constitutes
the $\ell$th vertex of the $k$-clique); \label{eq-clique-functional-q}%\vspace{-6pt}

\item $(q_{\ell,i}=0)\Or(q_{\ell',i}=0), \mbox{ for all } i\in [n],\;
	  \ell\ne \ell'\in[k]$

(expresses that the $i$th vertex of $G$ cannot be both the $\ell$th
and the $\ell'$th vertex of the $k$-clique);\label{eq-clique-holes-q}%\vspace{-6pt}

\item $(q_{\ell,i}=0)\Or(q_{\ell',j}=0)\Or(p_{i,j}=1), \mbox{ for all }
\ell\ne\ell'\in[k], i\ne j\in[n]$

(expresses that if both the vertices $i$ and $j$ in $G$ are in the $k$-clique,
then there is an edge in $G$ between $i$ and $j$);
\label{eq-clique-p-and-q-vars}%\vspace{-6pt}

\item $(r_{1,i}=1)\OrDots(r_{k',i}=1), \mbox{ for all } i\in[n]$

(expresses that every vertex of $G$ pertains to at least one independent set);
\label{eq-clique-r-pigeons}%\vspace{-6pt}

\item  $(r_{\ell,i}=0)\Or(r_{\ell',i}=0), \mbox{ for all } \ell\ne\ell\in[k'], i\in[n]$

(expresses that every vertex of $G$ pertains to at most one independent set);
\label{eq-clique-holes-r}%\vspace{-6pt}

\item $(p_{i,j}=0)\Or(r_{t,i}=0)\Or(r_{t,j}=0), \mbox{ for all } i\ne j\in[n],
t\in[k']$

(expresses that if there is an edge between vertex $i$ and $j$ in $G$, then
$i$ and $j$ cannot be in the same independent set);
\label{eq-clique-final-contradictory-clauses}
\end{enumerate}
\renewcommand{\theenumi}{(\Roman{enumi})}   % change back to Roman items
\end{definition}

\begin{remark}
Our formulation of the clique-coloring formulas above is similar to the one
used by \cite{BPR97}, except that we consider also the $\vr p i j$ variables
(we added the (\ref{eq-clique-p-and-q-vars}) clauses and changed accordingly
the (\ref{eq-clique-final-contradictory-clauses}) clauses).
This is done for the sake of clarity of the contradiction itself,
and also to make it clear that the formulas are in the appropriate form
required by the interpolation method
(see Section \ref{sec-interpo} for details on the interpolation method).
By resolving over the $p_{i,j}$ variables in (\ref{eq-clique-p-and-q-vars}) and
(\ref{eq-clique-final-contradictory-clauses}), one can obtain precisely
the collection of clauses in \cite{BPR97}.
\end{remark}

Atserias, Bonet \& Esteban \cite{ABE02} demonstrated \ps0 (in $n$) \Res 2
refutations of \clique n k {k'},
when $k=\sqrt n$ and $k'=(\log n )^2/8 \log\log n$.
These are rather weak parameters,
but they suffice to establish the fact that \Res 2 does not possess the
feasible monotone interpolation property. %(see Section \ref{sec-mono-interp} for
%details on the feasible monotone interpolation).

The \Res 2 proof system (also called \emph{$2$-DNF resolution}), first considered
in \cite{Kra01-Fundamenta},
is resolution extended to operate with $2$-DNF formulas, defined as follows.

A \emph{$2$-term} is a conjunction of up to two literals.
A $2$-DNF is a disjunction of $2$-terms.
The size of a $2$-term is the number of literals in it (that is, either
$1$ or $2$).
The \emph{size of a $2$-DNF} is the total size of all the $2$-terms in it.

\begin{definition}[\Res 2]\label{def-res2}
A \emph{\Res {$2$} proof of a $2$-DNF $D$ from a collection $K$ of $2$-DNFs}
is a sequence of $2$-DNFs $D_1,D_2,\ldots,D_s\,$,
such that $D_s=D$, and every $D_j$ is either from $K$ or
was derived from previous line(s)
in the sequence by the following inference rules:
\begin{description}
\item[\quad Cut] Let $A,B$ be two $2$-DNFs.

From $A\Or \BigAnd_{i=1}^2 l_i$ and $B\Or \BigOr_{i=1}^2 \neg l_i$
derive $A \Or B$, where the $l_i$'s are (not necessarily distinct)
literals (and $\neg l_i$ is the negation of the literal $l_i$).

\item[\quad AND-introduction] Let $A,B$ be two $2$-DNFs and $l_1,l_2$ two literals.

From $A\Or l_1$ and $B\Or l_2$ derive
$A\Or B\Or \BigAnd_{i=1}^2 l_i$.

\item[\quad Weakening] From a $2$-DNF $A$ derive
$A \Or \BigAnd_{i=1}^2 l_i$\,, where the $l_i$'s are (not necessarily distinct)
literals.
\end{description}
A \Res {$2$} \emph{refutation} of a collection of $2$-DNFs $K$ is
a \Res {$2$} proof of the empty disjunction $\Box$ from $K$
(the empty disjunction stands for \textsc{false}).
The \emph{size} of a \Res {$2$} proof is the total size of all
the $2$-DNFs in it.
\end{definition}

Given a collection $K$ of $2$-DNFs we translate it into a collection of
disjunctions of linear
equations via the following translation scheme.
For a literal $l$, denote by $\widehat l$
the translation that maps a variable $x_i$ into $x_i$, and $\neg x_i$
into $1-x_i$.
A $2$-term $l_1\And l_2$ is first transformed into the equation
$\widehat l_1 + \widehat l_2 =2$,
and then moving the free-terms in the left
hand side of $\;\widehat l_1 + \widehat l_2 =2$
(in case there are such free-terms) to the right hand side;
So that the final translation of $l_1\And l_2$ has only a single
free-term in the right hand side.
A disjunction of $2$-terms (that is, a $2$-DNF)
$D= \BigOr_{i\in I}(l_{i,1}\And l_{i,2})$ is translated into
the disjunction of the translations of the $2$-terms,
%$\BigOr_{i\in I}(\widehat l_{i,1}+ \widehat l_{i,2}=2)$,
denoted by $\widehat D$.
It is clear that every assignment satisfies a $2$-DNF $D$ if and only if it satisfies
$\widehat D$.

\begin{proposition}\label{prop-R(lin)-sim-Res2}
R(lin) polynomially simulates \Res {$2$}.
In other words, if $\pi$ is a \Res {$2$} proof of $D$ from a collection of $2$-DNFs
$K_1,\ldots,K_t$,
then there is an R(lin) proof of $\widehat D$ from
$\widehat K_1,\ldots,\widehat K_t$ whose size is polynomial in the size of $\pi$.
\end{proposition}

The proof of Proposition \ref{prop-R(lin)-sim-Res2}
proceeds by induction on the length (that is, the number of proof-lines)
in the \Res 2 proof.
This is pretty straightforward and similar to the simulation of
resolution by R(lin), as illustrated in the proof of
Proposition \ref{prop-R(lin)-sim-resolution}.
We omit the details.

\begin{theorem}[\cite{ABE02}]\label{thm-ABE02}
Let $k=\sqrt{n}$ and $k'=(\log n )^2/8 \log\log n$. Then \clique n k { k'} has
\Res {$2$} refutations of size polynomial in $n$.
\end{theorem}

Thus, Proposition \ref{prop-R(lin)-sim-Res2} yields the following:
\begin{corollary}\label{cor-R(lin)-proof-clique}
Let $k,k'$ be as in Theorem \ref{thm-ABE02}.
Then \clique n k { k'} has R(lin) refutations of size polynomial in $n$.
\end{corollary}

The following corollary is important
(we refer the reader to Section \ref{sec-mono-interp} in the Appendix for the
necessary relevant definitions concerning
the \emph{feasible monotone interpolation
property}
and to Section \ref{sec-interpo} for explanation and definitions concerning
the general [non-monotone] interpolation method).

\begin{corollary}\label{cor-RC0-no-interp}
\RL0 does not possess the feasible monotone interpolation property.
\end{corollary}

\begin{remark}
The proof of \clique n k {k'} inside \Res {$2$}
demonstrated in \cite{ABE02}
 (and hence, also
the corresponding proof inside R(lin))
proceeds along the following lines.
First reduce \clique n k {k'} to the $k$ to $k'$ pigeonhole principle.
For the appropriate values of the parameters $k$ and $k'$
--- and specifically, for the values in Theorem \ref{thm-ABE02} ---
there is a short \emph{resolution} proof of the $k$ to $k'$ pigeonhole principle
(this was shown by Buss \& Pitassi \cite{BP97});
(this resolution proof is polynomial in the number of pigeons $k$,
but not in the number of holes $k'$, which is exponentially smaller than $k$).\footnotemark
~Therefore, in order to conclude the refutation of \clique n k {k'} inside \Res 2
(or inside R(lin)),
it suffices to simulate the short resolution refutation of the $k$ to $k'$
pigeonhole principle.\footnotetext{Whenever $k\ge 2k'$ the $k$ to $k'$ pigeonhole principle is
referred to as the \emph{weak pigeonhole principle}.}
It is important to emphasize this point:
After reducing, inside R(lin), \clique n k {k'} to the pigeonhole principle,
one simulates the \emph{resolution} refutation of the pigeonhole principle,
and this has nothing to do with the small-size \RZ0 refutations of the
pigeonhole principle
demonstrated in Section \ref{sec-php}.
This is because, the reduction (inside R(lin)) of
\clique n k {k'} to the $k$ to $k'$
pigeonhole principle, results in a \emph{substitution instance} of the
pigeonhole principle formulas; in other words, the reduction results in a
collection of disjunctions that are similar to the pigeonhole principle
disjunctions \emph{where each original pigeonhole principle variable is
substituted by some big formula}
(and, in particular, these disjunctions are not \RZ0-lines at all).
(Note that \RZ0 does not admit short proofs
of the clique-coloring formulas as we show in Section \ref{sec-lower-bounds}.)
\end{remark}

%====================================================
\section{Interpolation Results for R$^{\mathbf 0}$(lin)}\label{sec-interpo}
%====================================================

In this section we study the applicability of the
feasible (non-monotone) interpolation technique to \RZ0 refutations.
In particular, we show that \RZ0 admits a polynomial
(in terms of the \RZ0-proofs) upper bound on the (non-monotone)
circuit-size of interpolants.
In the next section we shall give a polynomial upper bound on the \emph{monotone}
circuit-size of interpolants,
but only in the case that the interpolant corresponds
to the clique-coloring formulas (whereas, in this section we are interested in
the general case; that is, upper bounding circuit-size of
interpolants corresponding to any formula [of the prescribed type; see below]).
First, we shortly describe the feasible interpolation method %, as applied in
%propositional proof complexity for the sake of obtaining lower bounds on proof size.
and explain how this method can be applied to obtain (sometime, conditional)
lower bounds on proof size.
Explicit usage of the interpolation method in proof complexity
goes back to \cite{Kra94-Lower}.

Let $A_i(\vec p, \vec q)$, $i\in I$,
and $B_j(\vec p, \vec r)$, $j\in J$, ($I$ and $J$ are sets of indices)
be a collection of formulas (for instance, a collection of
disjunctions of linear equations)
in the displayed variables only.
Denote by $A(\vec p,\vec q)$ the conjunction of all $A_i(\vec p, \vec q)$, $i\in I$,
and by $B(\vec p, \vec r)$, the conjunction of all $B_j(\vec p, \vec r)$, $j\in J$.
%(equivalently, we can consider the $A_i$'s and $B_j$'s to be disjunctions of linear equations).
Assume that $\vec p,\vec q,\vec r$ are pairwise disjoint
sets of distinct variables,
and that there is no assignment that satisfies
both $A(\vec p,\vec q)$ and $B(\vec p, \vec r)$.
Fix an assignment $\vec \alpha$ to the variables in $\vec p$. The $\vec p$ variables
are the\emph{ only common variables} of the $A_i$'s and the $B_j$'s.
Therefore, either $A(\vec \alpha, \vec q)$ is unsatisfiable
or $B(\vec \alpha, \vec r)$ is unsatisfiable.

The interpolation technique transforms a refutation of
$A(\vec p,\vec q)\And B(\vec p,\vec r)$,
in some proof system, into a circuit (usually a Boolean circuit)
separating those assignments $\vec \alpha$ (for $\vec p$) for which
$A(\vec \alpha,\vec q)$ is unsatisfiable,
from those assignments $\vec \alpha$  for which  $B(\vec \alpha,\vec r)$ is
unsatisfiable (the two cases are not necessarily exclusive,
so if both cases hold for an assignment, the circuit
can output either that the first case holds or that the second case holds).
In other words, given a refutation of $A(\vec p,\vec q)\And B(\vec p,\vec r)$,
we construct a circuit $C(\vec p)$, called \emph{the interpolant}, such that
\begin{equation}\label{eq-interp-circuit-C}
	\begin{array}{lll}
		C(\vec \alpha ) = 1 \quad & \Longrightarrow &\quad A(\vec \alpha ,\vec q)
		\,\,\,{\rm{is}}\,{\rm{unsatisfiable}}, {\rm and}\\
		C(\vec \alpha ) = 0 \quad & \Longrightarrow  &\quad B(\vec \alpha ,\vec r)
		\,\,\,{\rm{is}}\,{\rm{unsatisfiable}}{\rm{.}}
	\end{array}
\end{equation}
(Note that if $U$ denotes the set of those assignments $\vec \alpha$
for which $A(\vec \alpha, \vec q)$ is \emph{satisfiable},
and $V$ denotes the set of those assignments $\vec \alpha$ for which
$B(\vec \alpha, \vec r)$ is \emph{satisfiable},
then $U$ and $V$ are disjoint
[since $A(\vec p,\vec q)\And B(\vec p,\vec r)$ is unsatisfiable],
and $C(\vec p)$ separates $U$ from $V$; see Definition
\ref{def-separating-circ} below.)

Assume that for a proof system $\mathcal P$ the transformation from refutations
of $A(\vec p,\vec q), B(\vec p,\vec r)$ into the corresponding interpolant
circuit $C(\vec p)$ results in a circuit whose size is
polynomial in the size of the refutation.
%In such a case, %we say that the proof system has the \emph{feasible interpolation property}.
%a proof system ${\mathcal P}$ has the feasible interpolation property,
Then, an exponential lower bound on circuits for which
(\ref{eq-interp-circuit-C}) holds, implies an exponential lower bound on $\mathcal
P$-refutations of $A(\vec p,\vec q), B(\vec p,\vec r)$.
\QuadSpace

\subsection{Interpolation for Semantic Refutations}\label{sec-interp-sem-refs}

We now lay out the basic concepts needed to formally describe the feasible
interpolation technique.
We use the general notion of \emph{semantic refutations}
(which generalizes any standard propositional refutation system).
We shall use a close terminology to that in \cite{Kra97-Interpolation}.

\begin{definition}[Semantic refutation]\label{def-semantic-refs}
Let $N$ be a fixed natural number and let $E_1,\ldots, E_k\se\zo^N$,
where $\bigcap_{i=1}^k E_i=\emptyset$.
A \emph{semantic refutation} from $E_1,\ldots, E_k$
is a sequence $D_1,\ldots,D_m\se\zo^N$ with $D_m=\emptyset$ and such that
for every $i\in[m]$, $D_i$ is either one of the $E_j$'s or
is deduced from two previous $D_j,D_\ell$, $1\le j,\ell<i$, by the following
\emph{semantic inference rule}:%\vspace{-7pt}

\begin{itemize}
\item From $A,B\se\zo^N$ deduce any $C$, such that $C\supseteq (A\cap B)$.
\end{itemize}
\end{definition}

Observe that any standard propositional refutation (with inference rules
that derive from at most two proof-lines, a third line) can be regarded as a
semantic refutation: just substitute each refutation-line by the set
of its satisfying assignments; and by the soundness of the inference rules applied
in the refutation, it is clear that each refutation-line (considered as the set
of assignments that satisfy it) is deduced by the semantic inference rule
from previous refutation-lines.

\begin{definition}[Separating circuit]\label{def-separating-circ}
Let $\mathcal U,\mathcal V\se\zo^n$, where $\mathcal U\cap \mathcal V=\emptyset$,
be two disjoint sets.
A Boolean circuit $C$ with $n$ input variables is said
to \emph{separate $\mathcal U$ from $\mathcal V$} if $C(\vec x)=1$ for
every $\vec x\in \mathcal U$,
and $C(\vec x)=0$ for every $\vec x\in \mathcal V$.
In this case we also say that
$\mathcal U$ and $\mathcal V$ are \emph{separated by $C$}.
\end{definition}

\begin{convention}
In what follows we sometime identify a Boolean formula with the set of its
satisfying assignments.
\end{convention}

\begin{notation}
For two (or more) binary strings $u,v\in\zo^*$, we write $(u,v)$ to denote
the concatenation of the $u$ with $v$ (where $v$ comes to the right of $u$,
obviously).
\end{notation}

Let $N=n+s+t$ be fixed from now on.
Let $A_1,\ldots,A_k\se\zo^{n+s}$ and let $B_1,\ldots, B_\ell\se\zo^{n+t}$.
Define the following two sets of assignments of length $n$ (formally,
$0,1$ strings of length $n$) that can be extended to satisfying assignments of
$A_1,\ldots,A_k$ and $B_1,\ldots, B_\ell$, respectively
(formally, those $0,1$ string of length $n+s$ and $n+t$,
that are contained in all $A_1,\ldots,A_k$ and $B_1,\ldots, B_\ell$, respectively):
\[\mathcal U_A:=\set{u\in\zo^n \;{\biggl|\biggr.}\; \exists q\in\zo^s,\,(u,q)\in
\bigcap_{i=1}^k A_i}\,,\]
\[\mathcal V_B:=\set{v\in\zo^n \;{\biggl|\biggr.}\; \exists r\in\zo^t,\,(v,r)\in
\bigcap_{i=1}^\ell B_i}\,.\]

\begin{definition}[\textbf{polynomial upper bounds on interpolants}]
\label{def-quasi-interp-UB}
Let $\mathcal P$ be a propositional refutation system.
Assume that $\vec p,\vec q,\vec r$ are pairwise disjoint
sets of distinct variables, where $\vec p$ has $n$ variables,
$\vec q$ has $s$ variables and $\vec r$ has $t$ variables.
Let $A_1(\vec p, \vec q),\ldots,A_k(\vec p, \vec q)$ and
$B_1(\vec p, \vec r),\ldots, B_\ell(\vec p, \vec r)$
%$A(\vec p, \vec q), B(\vec p, \vec r)$
be two collections of formulas with the displayed variables only.
Assume that for any such $A_1(\vec p, \vec q),\ldots,A_k(\vec p, \vec q)$
and $B_1(\vec p, \vec r),\ldots, B_\ell(\vec p, \vec r)$,
if there exists a $\mathcal P$-refutation of size $S$
for
$A_1(\vec p, \vec q)\And\cdots\And A_k(\vec p, \vec q)\And
B_1(\vec p, \vec r)\And\ldots\And B_\ell(\vec p, \vec r)$
then there exists
a Boolean circuit separating $\mathcal U_A$
from $\mathcal V_B$ of size polynomial in $S$.\footnotemark
~In this case we say that $\mathcal P$ has
a \emph{polynomial upper bound on interpolant circuits}.
\end{definition}
\footnotetext{Here $\mathcal U_A$ and $\mathcal V_B$ are defined as above,
by identifying the $A_i(\vec p, \vec q)$'s and the
$B_i(\vec p, \vec r)$'s with the sets of assignments that satisfy them.}

\subsubsection{The Communication Game Technique}
The \emph{feasible interpolation via communication game technique}
is based on transforming proofs into Boolean circuits, where the size of the resulting
circuit depends on the communication complexity of each proof-line.
This technique goes back to \cite{IPU94} and \cite{Razb95-Unprovability}
and was subsequently applied and extended in \cite{BPR97} and \cite{Kra97-Interpolation}
(\cite{IPU94} and \cite{BPR97} did not use explicitly the notion of interpolation of
tautologies or contradictions).
We shall employ the interpolation theorem of Kraj{\'i}\v{c}ek in \cite{Kra97-Interpolation},
that demonstrates how to transform a small semantic refutation with
each proof-line having low communication complexity into a small Boolean circuit
separating the corresponding sets.

The underlying idea of the interpolation via communication game technique
is that a (semantic) refutation,
where each proof-line is of small (that is, logarithmic) communication complexity,
can be transformed into an efficient communication protocol for the
\emph{Karchmer-Wigderson game} (following \cite{KW88}) for two players.
In the Karchmer-Wigderson game the first player knows some binary
string $u\in U$ and the second player knows
some different binary string $v\in V$,
where $U$ and $V$ are disjoint sets of strings.
The two players communicate by sending information bits to one another (following
a protocol previously agreed on).
The goal of the game is for the two players to decide on an index $i$ such that
the $i$th bit of $u$ is different from the $i$th bit of $v$.
An efficient Karchmer-Wigderson protocol (by which we mean a protocol that
requires the players to exchange at most a logarithmic number of bits in the
worst-case)
can then be transformed into a small circuit separating $U$ from $V$
(see Definition \ref{def-separating-circ}).
This efficient transformation from protocols for Karchmer-Wigderson games (described
in a certain way) into circuits,
was demonstrated by Razborov in \cite{Razb95-Unprovability}.
 So overall, given a semantic refutation with proof-lines of low communication
complexity, one can obtain a small circuit for separating the corresponding sets.

First, we need to define the concept of \emph{communication complexity}
in a suitable way for the interpolation theorem.

\begin{definition}[Communication complexity]\label{def-CC}
Let $N=n+s+t$ and $A\se\zo^N$. Let $u,v\in\zo^n$, $q^u\in\zo^s$, $r^v\in\zo^t$.
Denote by $u_i$, $v_i$ the $i$th bit of $u$, $v$, respectively, and let
$(u,q^u,r^v)$ and $(v,q^u,r^v)$ denote the concatenation of strings
$u,q^u,r^v$ and $v,q^u,r^v$, respectively.
Consider the following three tasks:%\vspace{-4pt}

\begin{enumerate}
\item Decide whether $(u,q^u,r^v)\in A$;\label{it-task1}%\vspace{-5pt}

\item Decide whether $(v,q^u,r^v)\in A$;\label{it-task2}%\vspace{-5pt}

\item If one of the following holds:

{\rm(i)} $(u,q^u,r^v)\in A$ and $(v,q^u,r^v)\not\in A$; or

{\rm(ii)} $(u,q^u,r^v)\not\in A$ and $(v,q^u,r^v)\in A$,

then find an $i\in[n]$, such that $u_i\neq v_i$;\label{it-task3}
\end{enumerate}%\vspace{-3pt}
Consider a game between two players, Player I and Player II,
where Player I knows $u\in\zo^n,q^u\in\zo^s$
and Player II knows $v\in\zo^n,r^v\in\zo^t$.
The two players communicate by exchanging bits of information between them
(following a protocol previously agreed on).
The \emph{communication complexity of $A$},
denoted $CC(A)$, is the minimal (over all protocols) number of bits that players
I and II need to exchange in the worst-case in solving each of Tasks
\ref{it-task1}, \ref{it-task2} and \ref{it-task3} above.\footnotemark
\footnotetext{In other words, $CC(A)$ is the minimal number $\zeta$, for which there exists
a protocol, such that for every input ($u,q^u$ to Player I and $v,r^v$ to Player II)
and every task (from Tasks \ref{it-task1}, \ref{it-task2} and \ref{it-task3}),
the players need to exchange at most $\zeta$ bits in order to solve the task.}
\end{definition}

For $A\se\zo^{n+s}$ define
\[\dot A:=\set{(a,b,c) \;{\bigl|\bigr.}\;  (a,b)\in A \mbox{ and } c\in\zo^t}\,,\]
where $a$ and $b$ range over $\zo^n$ and $\zo^s$, respectively.
Similarly, for $B\se\zo^{n+t}$ define
\[\dot B:=\set{(a,b,c) \;{\bigl|\bigr.}\; (a,c)\in B \mbox{ and } b\in\zo^t}\,,\]
where $a$ and $c$ range over $\zo^n$ and $\zo^t$, respectively.\QuadSpace

\begin{theorem}[\cite{Kra97-Interpolation}]\label{thm-CC-interpolation}
Let $A_1,\ldots,A_k\se\zo^{n+s}$ and $B_1,\ldots, B_\ell\se\zo^{n+t}$.
Let $D_1,\ldots, D_m$ be a semantic refutation from
$\dot A_1,\ldots,\dot A_k$ and $\dot B_1,\ldots, \dot B_\ell$.
Assume that $\,CC(D_i)\le \zeta$, for all $i\in[m]$.
Then, the sets $\mathcal U_A$ and $\mathcal V_B$ (as defined above)
can be separated by a Boolean circuit of size ~$(m+n)2^{O(\zeta)}$.
\end{theorem}

In light of Theorem \ref{thm-CC-interpolation},
to demonstrate that a certain propositional refutation system $\mathcal P$
possesses a polynomial upper bound on interpolant circuits (see Definition
\ref{def-quasi-interp-UB})
it suffices to show that any proof-line of $\mathcal P$ induces a set of
assignments
with at most a logarithmic (in the number of variables) communication
complexity (Definition \ref{def-CC}).

%===================================================================================
\subsection{Polynomial Upper Bounds on Interpolants for R$^{\mathbf 0}$(lin)}
%===================================================================================
\label{sec-quasi-poly-interp-UB}

Here we apply Theorem \ref{thm-CC-interpolation} to show that \RZ0 has
polynomial upper bounds on its interpolant circuits.
Again, in what follows we sometime identify a disjunction of linear equations
with the set of its satisfying assignments.

\begin{theorem}\label{thm-RZ0-interpolation}
\RZ0 has a polynomial upper bounds on interpolant circuits
(Definition \ref{def-quasi-interp-UB}).
\end{theorem}

According to the paragraph after Theorem \ref{thm-CC-interpolation},
all we need in order to establish Theorem \ref{thm-RZ0-interpolation} is
the following lemma:

\begin{lemma}\label{lem-low-CC}
Let $D$ be an \RZ0-line with $N$ variables and let $\widetilde D$ be
the set of assignments that satisfy $D$.\footnotemark
Then, $CC(\widetilde D) \le O(\log N)$.
\end{lemma}
\footnotetext{The notation $\widetilde D$ has nothing to do with
the same notation used in Section \ref{sec-systems-definitions}.}
\begin{proof}
Let $N=n+s+t$ (and so $\widetilde D\in\zo^{n+s+t}$).
For the sake of convenience we shall assume that the $N$ variables in $D$ are
partitioned into (pairwise disjoint)
three groups $\vec p := (p_1\ldots, p_n)$,\;
$\vec q:=(q_1,\ldots,q_s)$ and $\vec r:=(r_1,\ldots, r_t)$.
Let $u,v\in\zo^n$, $q^u\in\zo^s$, $r^v\in\zo^t$.
Assume that Player I knows $u, q^u$ and Player II knows
$v, r^v$.

By the definition of an \RZ0-line (see Definition \ref{def-RZ0-clause})
we can partition the disjunction $D$ into a \emph{constant number} of disjuncts,
where one disjunct is a (possibly empty, translation of a) clause
in the $\vec p,\vec q,\vec r$ variables
(see Section \ref{sec-clauses-of-linear-equations}),
and all other disjuncts have the following form:
\begin{equation}\label{eq-interp-MCC-disjunct}
\BigOr_{i\in I}
 \left(
   \vec a \cd \vec p + \vec b \cd \vec q + \vec c \cd \vec r
   =\ell_i
 \right)\,,
\end{equation}
where $I$ is (an unbounded) set of indices,
$\ell_i$ are integer numbers, for all $i\in I$, and
$\vec a, \vec b, \vec c$ denote vectors of $n,s$ and $t$ constant coefficients,
respectively.%
%and where $\vec a \cd \vec p$ abbreviates the scalar product $a_1 p_1 \PlusDots a_n p_n$
%(and similarly for $\vec b \cd \vec q$ and $\vec c \cd \vec r$).

Let us denote the (translation of the) clause from $D$
in the $\vec p,\vec q,\vec r$ variables by
\[
	P\Or Q\Or R\,,
\]
where $P$, $Q$ and $R$ denote the (translated) sub-clauses consisting
of the $\vec p$, $\vec q$ and $\vec r$ variables, respectively.
\HalfSpace

We need to show that by exchanging $O(\log N)$ bits, the players
can solve each of Tasks \ref{it-task1}, \ref{it-task2} and
\ref{it-task3} from Definition \ref{def-CC}, correctly.

\para{Task \ref{it-task1}:}
The players need to decide whether $(u,q^u,r^v)\in \widetilde D$.
Player II, who knows $r^v$, computes the numbers $\vec c \cdot\, r^v$,
for every $\vec c$ pertaining to every disjunct of the form
shown in Equation (\ref{eq-interp-MCC-disjunct}) above.
Then, Player II sends the (binary representation of) these numbers
to Player I.
Since there are only a constantly many such numbers and the coefficients in every
$\vec c$ are also constants,
this amounts to $O(\log t)\le O(\log N)$ bits that Player II sends to Player I.
Player II also computes the truth value of the sub-clause $R$, and sends this (single-bit)
value to Player I.

Now, it is easy to see that Player I has sufficient data to
compute by herself/himself whether $(u,q^u,r^v)\in \widetilde D$
(Player I can then send a single bit informing Player II
whether $(u,q^u,r^v)\in \widetilde D$).

\para{Task \ref{it-task2}:} This is analogous to Task \ref{it-task1}.

\para{Task \ref{it-task3}:}
%Essentially, the players apply binary search to find an index $i\in[n]$ such that
%$u_i\ne v_i$, as we now explain more formally.

Assume that $(u,q^u,r^v)\in \widetilde D$ and $(v,q^u,r^v)\not\in \widetilde D$
(the case $(u,q^u,r^v)\not\in \widetilde D$ and $(v,q^u,r^v)\in \widetilde D$
is analogous).

The first rounds of the protocol are completely similar to that described in
Task \ref{it-task1} above:
Player II, who knows $r^v$, computes the numbers $\vec c \cdot\, r^v$,
for every $\vec c$ pertaining to every disjunct of the form
shown in Equation (\ref{eq-interp-MCC-disjunct})
above. Then, Player II sends the (binary representation of) these numbers
to Player I.
Player II also computes the truth value of the sub-clause $R$,
and sends this (single-bit) value to Player I.
Again, this amounts to $O(\log N)$ bits
that Player II sends to Player I.

By assumption (that $(u,q^u,r^v)\in \widetilde D$ and
$(v,q^u,r^v)\not\in \widetilde D$) the players need to deal
only with the following two cases:\smallskip

\emph{Case 1:} The assignment $(u,q^u,r^v)$ satisfies the
clause $P\Or Q\Or R$ while
$(v,q^u,r^v)$ falsifies $P\Or Q\Or R$.
Thus, it must be that $\vec u$ satisfies the sub-clause $P$ while $\vec v$
falsifies $P$. This means that for any $i\in[n]$ such that $u_i$
sets to $1$ a literal in $P$ (there ought to exist at least one such $i$),
it must be that $u_i\neq v_i$.
Therefore, all that
Player I needs to do is to send the (binary representation of)
index $i$ to Player II. (This amounts to $O(\log N)$ bits
that Player I sends to Player II.)\smallskip

\emph{Case 2:}
There is some linear equation
\begin{equation}\label{eq-CC-satisfied-eqn}
\vec a \cd \vec p + \vec b \cd \vec q + \vec c \cd \vec r   =\ell
\end{equation}
in $D$, such that
$ \vec a \cd u + \vec b \cd q^u + \vec c \cd r^v =\ell$.
Note that (by assumption that $(v,q^u,r^v)\not\in \widetilde D$)
it must also hold that:
$ \vec a \cd v + \vec b \cd q^u + \vec c \cd r^v \ne \ell$
(and so there is an $i\in[n]$, such that $u_i\ne v_i$).
Player I can find linear equation (\ref{eq-CC-satisfied-eqn}),
as he/she already received from
Player II all the possible
values of $\vec c\cd \vec r$ (for all possible $\vec c$\,'s in $D$).

Recall that the left hand side of a linear equation $\vec d \cd \vec x=\ell$
is called the \emph{linear form} of the equation.
By the definition of an \RZ0-line there are only constant many distinct
linear forms in $D$.
Since both players know these linear forms, we can assume that each linear form
has some index associated to it by both players.
 Player I sends to Player II the index of the linear form
$\vec a \cd \vec p + \vec b \cd \vec q + \vec c \cd \vec r$ from
(\ref{eq-CC-satisfied-eqn}) in $D$.
Since there are only \emph{constantly} many such linear forms
in $D$, it takes only constant number of bits to send this index.

Now both players need to apply a protocol for finding
an $i\in [n]$ such that $u_i\ne v_i$, where
$\vec a \cd \vec u + \vec b \cd q^u + \vec c \cd r^v = \ell$ and
$\vec a \cd \vec v + \vec b \cd q^u + \vec c \cd r^v \neq \ell$.
Thus, it remains only to prove the following claim:
\begin{claim}
There is a communication protocol in which Player I and Player II need at most
$O(\log N)$ bits of communication in order to find an $i\in [n]$ such
that $u_i\ne v_i$ (under the above conditions).
\end{claim}

\begin{proofclaim}
We invoke the well-known connection between Boolean circuit-depth and communication
complexity.
Let $f:\zo^N\to\zo$ be a Boolean function. Denote by ${\rm dp}(f)$
the minimal depth of a Boolean circuit computing $f$.
Consider a game between two players:
Player I knows some $\vec x\in\zo^N$ and
Player II knows some other $\vec y\in\zo^N$, such that $f(\vec x)=1$ while $f(\vec y)=0$.
The goal of the game is to find an $i\in[N]$ such that $x_i\neq y_i$.
Denote by ${\rm CC'}(f)$ the minimal number of bits needed for the two players to
communicate (in the worst case\footnote{Over all inputs $\vec x,\vec y$ such that
$f(\vec x)=1$ and $f(\vec y)=0$.}) in order to
solve this game.\footnote{The measure $CC'$ is basically the
same as $CC$ defined earlier.}
Then, for any function $f$
it is known  that ${\rm dp}(f) =
{\rm CC'}(f)$ (see \cite{KW88}).

Therefore, to conclude the proof of the
claim it is enough to establish that the function
$f:\zo^N\to\zo$ that receives the input variables $\vec p,\vec q,\vec r$ and
computes the truth value of
$\vec a \cd \vec p + \vec b \cd \vec q + \vec c \cd \vec r = \ell$ has Boolean
circuit of depth $O(\log N)$.
In case all the coefficients in $\vec a,\vec b,\vec c$ are $1$,
it is easy to show\footnote{Using the known $O(\log N)$-depth Boolean
circuits for the threshold functions.} that there is a Boolean
circuit of depth $O(\log N)$ that computes the function $f$.
In the case that
the coefficients in $\vec a,\vec b,\vec c$ are all constants,
it is easy to show, by a reduction to the case where all coefficients
are $1$,\footnote{For instance, consider the
simple case where we have only a single variable.
That is, let $c$ be a constant and assume that
we wish to construct a circuit that computes $c\cd x =\ell$,
for some integer $\ell$.
Then, we take a circuit that computes the
function $f:\zo^c\to\zo$ that outputs the truth value of
$y_1+\ldots+y_c=\ell$ (thus, in $f$ all coefficients are $1$'s);
and to compute $c\cd x =\ell$ we only have to substitute each $y_i$ in the
circuit with the variable $x$.}
   that there is a Boolean
circuit of depth $O(\log N)$ that computes the function $f$.
We omit the details.
\end{proofclaim}
\end{proof}

%
%====================================================
\section{Size Lower Bounds}\label{sec-lower-bounds}
%====================================================

In this section we establish an exponential-size lower bound on \RZ0 refutations
of the clique-coloring formulas.
We shall employ the theorem of Bonet, Pitassi \&
Raz  in \cite{BPR97} that provides exponential-size lower bounds
for any semantic refutation of the
clique-coloring formulas, having low communication complexity in each
refutation-line.

First we recall the strong lower bound obtained by Alon \& Boppana \cite{AB87}
(improving over \cite{Razb85}; see also \cite{And85})
for the (monotone) \emph{clique separator} functions, defined as follows
(a function $f:\zo^n\to\zo$ is called \emph{monotone} if for all
$\alpha \in\zo^n$, $\alpha'\ge\alpha$ implies  $f(\alpha')\ge f(\alpha)$):

\begin{definition}[Clique separator]
A monotone boolean function $Q^n_{k,k'}$ is called a \emph{clique separator} if it
interprets its inputs as the edges of a graph on $n$ vertices,
and outputs $1$ on every input representing a $k$-clique, and $0$ on every
input representing a complete $k'$-partite graph (see Section \ref{sec-clique}).
\end{definition}

Recall that a \emph{monotone Boolean circuit} is a circuit that uses
only monotone Boolean gates (for instance, only the fan-in two gates $\And,\Or$).

\begin{theorem}[\cite{AB87}]\label{thm-Alon-Boppana87}
Let $k,k'$ be integers such that $3\le k'< k$ and $k\sqrt{k'}\le n/(8\log n)$,
then every monotone Boolean circuit %(that is, a circuit that uses only the gates $\And,\Or$)
that computes a clique separator function $Q^n_{k,k'}$ requires size at least
$$\frac{1}{8}\left(\frac{n}{4k\sqrt{k'}\log n}\right)^{\left(\sqrt{k'}+1\right)/2}\,.$$
%$2^{\Omega(\sqrt{k'})}$.
\end{theorem}

For the next theorem,
we need a slightly different (and weaker) version of communication complexity,
than that in Definition \ref{def-CC}.

\begin{definition}[Communication complexity (second definition)]
\label{def-CC-BPR97}
Let $X$ denote $n$ Boolean variables $x_1,\ldots,x_n$,
and let $S_1, S_2$ be a partition of $X$ into two disjoint sets of variables.
The communication complexity of a Boolean function $f:\zo^n\to\zo$ is the number of bits needed to
be exchanged by two players, one knowing the values given to the
$S_1$ variables and the other knowing the values given
to $S_2$ variables, in the worst-case,
over all possible partitions $S_1$ and $S_2$.
\end{definition}

\begin{theorem}[\cite{BPR97}]\label{thm-BPR97}
Every semantic refutation of \clique n k {k'} (for $k'<k$)
with $m$ refutation-lines and where each refutation-line
(considered as a the characteristic function of the line)
has communication complexity (as in Definition \ref{def-CC-BPR97}) $\zeta$,
can be transformed into a monotone circuit of size $m\cd 2^{3\zeta+1}$
that computes a separating function $Q^n_{k,k'}$.
\end{theorem}

In light of Theorem \ref{thm-Alon-Boppana87},
in order to be able to apply Theorem \ref{thm-BPR97} to \RZ0,
and arrive at an exponential-size lower bound for \RZ0 refutations of the
clique-coloring formulas,
it suffices to show that \RZ0 proof-lines have logarithmic communication
complexity:

\begin{lemma}\label{lem-lowCC-for-BPR97}
Let $D$ be an \RZ0-line with $N$ variables.
Then, the communication complexity (as in Definition \ref{def-CC-BPR97})
of $D$ is at most $O(\log N)$ (where $D$ is identified here with
the characteristic function of $D$). %a function that receives
%assignments on $N$ variables and outputs $1$ if and only if the input assignment
%satisfies $D$).
\end{lemma}

\begin{proof}
The proof is similar to the proof of Lemma \ref{lem-low-CC} for solving Task
\ref{it-task1} (and the analogous Task \ref{it-task2}) in Definition \ref{def-CC}.
\end{proof}

By direct calculations we obtain the following lower bound from Theorems \ref{thm-Alon-Boppana87},
\ref{thm-BPR97} and Lemma \ref{lem-lowCC-for-BPR97}:

\begin{corollary}
Let $k$ be an integer such that $3\le k'=k-1$ and assume that
$\frac{1}{2}\cd n/(8\log n) \le k\sqrt{k}\le n/(8\log n)$.
Then, for all $\eps<1/3$,
every \RZ0 refutation of \clique n k {k'} is of size at least
$2^{\Omega(n^\eps)}$.
\end{corollary}

When considering the parameters of Theorem \ref{thm-ABE02},
we obtain a super-polynomial separation between \RZ0 refutations and
R(lin) refutations, as described below.

From Theorems \ref{thm-Alon-Boppana87},\ref{thm-BPR97} and Lemma \ref{lem-lowCC-for-BPR97} we have (by
direct calculations):

\begin{corollary}\label{cor-LBforSmallParameters}
Let $k=\sqrt{n}$ and $k'=(\log n )^2/8 \log\log n$.
%and let $\eps<1/3$.
Then, every \RZ0 refutation of \clique n k {k'} has size at least
~$n^{\Omega\left(\frac{\log n}{\sqrt{\log\log n}}\right)}$.
%$\;\exp\left(\Omega\left(\frac{\log^2 n}{\sqrt{\log\log n}}\right)\right)$.
\end{corollary}

By Corollary \ref{cor-R(lin)-proof-clique},
R(lin) admits polynomial-size in $n$ refutations of \clique n k {k'}
under the parameters in Corollary \ref{cor-LBforSmallParameters}.
Thus we obtain the following separation result:
\begin{corollary}
{\rm R(lin)} is super-polynomially stronger than \RZ0.
\end{corollary}

\begin{comment}
Note that we do not need to assume that the coefficients in \RZ0-lines
are constants for the lower bound argument.
If the coefficients in \RZ0-lines are only
polynomially bounded (in the number of variables) then the same lower bound
as in Corollary \ref{cor-LBforSmallParameters} also applies.
This is because \RZ0-lines in which coefficients are polynomially bounded
integers,
still have low (that is, logarithmic)
communication complexity (as in Definition \ref{def-CC-BPR97}).
\end{comment}

%====================================================
\section{Applications to Multilinear Proofs}\label{sec-multilinear}
%====================================================

In this section we arrive at one of the main benefits
of the work we have done so far;
Namely, applying results on resolution over linear equations
in order to obtain new results for multilinear proof systems.
Subsection \ref{sec-fMC-background} that follows,
contains definitions, sufficient for the current paper,
concerning the notion of multilinear proofs introduced in
\cite{RT05}. %We start with some preliminary definitions.

%==============================================================
\subsection{Background on Algebraic and Multilinear Proofs}
\label{sec-fMC-background}
%==============================================================

%========================================================================
\subsubsection{Arithmetic and Multilinear Formulas}\label{secDefnArit-circ}
%========================================================================

\begin{definition}[Arithmetic formula]
Fix a field $\F$. An \emph{arithmetic formula} is a tree,
with edges directed from the leaves to the root,
and with unbounded (finite) fan-in.
Every leaf of the tree (namely, a node of fan-in $0$)
is labeled with either an input variable or a field element.
A field element can also label an edge of the tree.
Every other node of the tree is labeled with either $+$ or $\times$
(in the first case the node is a \emph{plus gate} and in the second case
a \emph{product gate}).
We assume that there is only one node of out-degree zero, called the \emph{root}.
The \emph{size }of an arithmetic formula $F$ is the total number of nodes in its graph
and is denoted by $|F|$.
An arithmetic formula computes a polynomial in the ring of polynomials
$\mathbb{F}[x_1,\ldots,x_n]$ in the following way.
A leaf just computes the input variable or field element that labels it.
A field element that labels an edge means that the polynomial computed at its
tail (namely, the node where the edge is directed from) is multiplied by this field element.
A plus gate computes the sum of polynomials computed by the tails of all incoming edges.
A product gate computes the product of the polynomials computed by the tails of
all incoming edges. (Subtraction is obtained using the constant $-1$.) The
output of the formula is the polynomial computed by the root.
The \emph{depth} of a formula $F$ is the maximal number of edges in a path from
a leaf to the root of $F$.
\end{definition}
We say that an arithmetic formula has a \emph{plus (resp., product) gate at
the root} if the root of the formula is labeled with a plus (resp., product) gate.
\QuadSpace

A polynomial is \emph{multilinear} if in each of its monomials the power of
every input variable is at most one.
\begin{definition}[Multilinear formula]\label{def-ml-fmla}
An arithmetic formula is \emph{a multilinear formula} (or equivalently,
\emph{multilinear arithmetic formula}) if the polynomial computed by
\emph{each} gate of the formula is multilinear (as a formal polynomial, that is,
as an element of $\,\mathbb{F}[x_1,\ldots,x_n]$).
\end{definition}

An additional definition we shall need is the following linear operator, called the
\emph{multilinearization operator}:

\begin{definition}[Multilinearization operator]\label{def-ml-operator}
Given a field $\;\mathbb{F}$ and a polynomial $q\in
\mathbb{F}[x_1,\ldots,x_n]$, we denote by $\ML{q}$ the unique multilinear
polynomial equal to $q$ modulo the ideal generated by all the polynomials
$\,x_i^2-x_i$, for all variables $x_i$.
\end{definition}

For example, if $q=x_1^2x_2+a x_4^3\,$ (for some $a\in \mathbb{F}$)
then $\,\ML{q}=x_1x_2+a x_4\,$. \bigs

The simulation of \RZ0 by multilinear proofs will rely
heavily on the fact that multilinear symmetric polynomials
have small depth-$3$ multilinear formulas over fields of characteristic $0$
(see \cite{SW01} for a proof of this fact).
To this end we define precisely the concept of symmetric polynomials.

A \emph{renaming} of the variables $x_1,\ldots,x_n$ is a permutation
$\sigma\in S_n$ (the symmetric group on $[n]$) such that $x_i$
is mapped to $x_{\sigma(i)}$ for every $1\le i\le n$.

\begin{definition}[Symmetric polynomial]\label{def-symetric-polys}
Given a set of variables $X=\set{x_1,\ldots,x_n}$,
a \emph{symmetric polynomial} $f$ over $X$ is a polynomial in
(all the variables of)
$X$ such that renaming of variables does not change the polynomial (as a formal
polynomial).
\end{definition}

\subsubsection{Polynomial Calculus with Resolution}

Here we define the PCR proof system, introduced by Alekhnovich \emph{et al.}
in \cite{ABSRW99}.
%of \emph{generic algebraic proof systems}, which is essentially a semantic
%proof system for reasoning with polynomials in the ring of polynomials
%over some fixed field $\F$.

\begin{definition}[Polynomial Calculus with Resolution (PCR)]\label{def-PCR}
Let $\mathbb{F}$ be some fixed field and let $Q:= \set{Q_1,\ldots,Q_m}$ be
a collection of multivariate polynomials from the ring of polynomials
$\F[x_1,\ldots, x_n, \bar{x}_1, \ldots, \bar{x}_n]$.
The variables $\bar{x}_1, \ldots, \bar{x}_n$ are treated as new formal variables.
 Call the set of polynomials $x^2-x$,
for $x\in\set{x_1,\ldots, x_n, \bar{x}_1, \ldots, \bar{x}_n}$,
plus the polynomials $x_i+\bar{x}_i-1$,
for all $1\le i\le n$, the set of \emph{Boolean axioms of PCR}.
 A \emph{PCR proof} from  $Q$ of a polynomial $g$ is a finite sequence
$\pi =(p_1 ,...,p_\ell)$ of multivariate polynomials from
$\mathbb{F}[x_1,\ldots, x_n, \bar{x}_1, \ldots, \bar{x}_n]$
(each polynomial $p_i$ is interpreted as the polynomial equation $p_i=0$),
where $p_\ell=g$ and for each $i\in[\ell]$,
either $p_i = Q_j\,$ for some $j\in[m]$, or $p_i$ is a Boolean axiom,
or $p_i$ was deduced from $p_j,p_k\,$, where $j,k<i$,
by one of the following inference rules:
\begin{description}
  \item[Product]

  From $p$ deduce $x_i\cdot p$\,, for some variable $x_i\,$;

  From $p$ deduce $\bar x_i\cdot p$\,, for some variable $\bar x_i\,$;

  \item[Addition] From $p$ and $q$ deduce $\alpha \cd p + \beta \cd
  q$, for some $\alpha, \beta \in \F$.
\end{description}
A \emph{PCR refutation of} $Q$ is a proof of $\;1$ (which is interpreted as $1=0$)
from $Q$.
The \emph{number of steps} in a PCR proof is the number of proof-lines
in it (that is, $\l$ in the case of $\pi$ above).
\end{definition}
%The \emph{size} of a PCR proof is defined to be the total number of monomials appearing in the
%polynomials of the proof. The \emph{degree} of a PC proof is the maximum degree of the
%polynomials in the proof. \med
Note that the Boolean axioms of PCR have only $0,1$ solutions,
where $\bx_i=0$ if $x_i=1$ and $\bx_i=1$ if $x_i=0$.

\subsubsection{Multilinear Proof Systems}\label{sec-fMC-definitions}

In \cite{RT05} the authors introduced a natural (semantic)
algebraic proof system that
operates with multilinear arithmetic formulas
denoted fMC (which stands for \emph{formula multilinear calculus}),
defined as follows:

\begin{definition}[Formula Multilinear Calculus (fMC)]
\label{def-fMC}
 Fix  a field $\mathbb{F}$ and let $Q:= \set{Q_1,\ldots,Q_m}$ be a collection of multilinear
polynomials from $\F[x_1,\ldots, x_n, \bar{x}_1, \ldots, \bar{x}_n]$ (the variables $\bar{x}_1,
\ldots, \bar{x}_n$ are treated as formal variables). Call the set of polynomials consisting of
$x_i+\bar{x}_i-1$ and $x_i\cdot\bar{x}_i\,$ for $\,1\le i\le n\,$, the \emph{Boolean axioms of
fMC}.
 An \emph{\,fMC proof from  $Q$ of a polynomial $g$} is a finite sequence $\pi =(p_1
,...,p_\ell)$ of multilinear polynomials from $\F[x_1,\ldots, x_n, \bar{x}_1, \ldots,
\bar{x}_n]\,$, such that $p_\ell=g$ and for each $i\in[\ell]$, either $\,p_i = Q_j\,$ for some
$j\in[m]$, or $p_i$ is a Boolean axiom of \emph{fMC}, or $p_i$ was deduced by one of the
following inference rules using $p_j,p_k$ for $j,k<i$:
\begin{description}
  \item[\quad Product] from $p$ deduce $q\cdot p$\,,
  for some polynomial $q\in \F[x_1,\ldots, x_n, \bar{x}_1, \ldots, \bar{x}_n]$
  \emph{such that $p\cdot q$ is multilinear};

  \item[\quad Addition]   from $p$, $q$ deduce $\alpha\cdot p + \beta\cdot
  q$, for some $\alpha, \beta \in \mathbb{F}$.
\end{description}
All the polynomials in an \emph{fMC} proof are represented as multilinear formulas.
 (A polynomial $p_i$ in an \emph{fMC} proof is interpreted as the polynomial equation $p_i=0$.)
An \emph{fMC refutation of} $Q$ is a proof of $\;1$ (which is interpreted as $1=0$) from $Q$.
The \emph{size} of an {\rm fMC} proof $\pi$ is  defined as the total sum of all the formula sizes in
$\pi$ and is denoted by $|\pi|$.
\end{definition}

Note that the Boolean axioms have only $0,1$ solutions, where $\bx_i=0$ if $x_i=1$ and $\bx_i=1$
if $x_i=0$, for each $1\le i\le n\,$.

\begin{definition}[Depth-$k$ Formula Multilinear Calculus (depth-$k$ fMC)]
\label{def-depth-k-fMC}
For a natural number $k$, \emph{depth-$k$ fMC} denotes a restriction of the \emph{fMC} proof
system, in which proofs consist of multilinear polynomials from $\mathbb{F}[x_1,\ldots, x_n,
\bar{x}_1, \ldots, \bar{x}_n]$ represented as multilinear formulas of depth at most $k$.
\end{definition}

%==========================================================
\subsection{From R(lin) Proofs to PCR Proofs}\label{sec-R(lin)-to-PCR}
%==========================================================
We now demonstrate a general and straightforward translation
from R(lin) proofs into PCR proofs over fields of characteristic $0$.
We use the term ``translation'' in order to distinguish it
from a \emph{simulation};
since here we are not interested in the size of PCR proofs.
In fact we have not defined the size of PCR proofs at all.
We shall be interested only in the \emph{number of steps }in PCR proofs.

\emph{From now on, all polynomials and arithmetic formulas are considered
over some fix field $\F$ of characteristic $0$}.
Recall that any field of characteristic $0$ contains
(an isomorphic copy of) the integer numbers,
and so we can use integer coefficients in the field.

\begin{definition}[Polynomial translation of R(lin) proof-lines]\label{def-poly-trnas-of-R(lin)-lines}
Let $D$ be a disjunction of linear equations:
\begin{equation}\label{eq-R(lin)-sim-proof-line}
 \left(a^{(1)}_1 x_1+\ldots + a^{(1)}_{n} x_n=a^{(1)}_0\right) \OrDots
 \left(a^{(t)}_1 x_1+\ldots + a^{(t)}_{n} x_n=a^{(t)}_0\right).
\end{equation}
We denote by $\widehat D$ its translation into the following polynomial:\footnotemark
\footnotetext{This notation should not be confused with the same notation
in Section \ref{sec-clique}.}
\begin{equation}\label{eq-R(lin)-sim-algebraic-line}
 \left(a^{(1)}_1 x_1+\ldots + a^{(1)}_{n} x_n-a^{(1)}_0\right) \cdots
 \left(a^{(t)}_1 x_1+\ldots + a^{(t)}_{n} x_n-a^{(t)}_0\right).
\end{equation}
If $D$ is the \emph{empty disjunction}, we define $\widehat D$ to be the polynomial
$1$.
\end{definition}

It is clear that every $0,1$ assignment to the variables in $D$,
satisfies $D$, if and only if $\widehat D$ evaluates to $0$ under the assignment.
\QuadSpace

%\para{Translating R(lin) proofs to PCR proofs.}
\begin{proposition}\label{prop-PCR-sim-R(lin)}
Let $\pi=(D_1,\ldots,D_\l)$ be an R(lin) proof sequence of $D_\l$,
from some collection of initial disjunctions of linear equations
$Q_1,\ldots,Q_m$.
Then, there exists a PCR proof of $\widehat D_\l$ from $\widehat Q_1,\ldots,
\widehat Q_m$ with at most a polynomial in $|\pi|$ number of steps.
\end{proposition}

\begin{proof}
We proceed by induction on the number of lines in $\pi$.% to show that

\emph{The base case} is the translation of the axioms of R(lin) via the translation
scheme in Definition \ref{def-poly-trnas-of-R(lin)-lines}.
An R(lin) Boolean axiom $(x_i=0)\Or(x_i=1)$ is translated into
$x_i\cd(x_i-1)$ which is already a Boolean axiom of PCR.
\QuadSpace

For \emph{the induction step}, we translate every R(lin) inference rule
application into a polynomial-size
PCR proof sequence as follows.
We use the following simple claim:

\begin{claim}\label{cla-simple-PCR-prod}
Let $p$ and $q$ be two polynomials and let $s$ be the
minimal size of an arithmetic formula computing $q$.
Then one can derive in PCR,
with only a polynomial in $s$ number of steps,
from $p$ the product $q\cd p$.\footnotemark
\end{claim}
\footnotetext{Again, note that we only
require that the number of steps in the proof is polynomial.
We do not consider here the \emph{size} of the PCR proof.}

\begin{proofclaim}
By induction on $s$.
\end{proofclaim}

Assume that $D_i = D_j\Or L$ was derived from $D_j$ using
the Weakening inference rule of R(lin),
where $j<i\le\l$ and $L$ is some linear equation.
Then, by Claim \ref{cla-simple-PCR-prod},
$\widehat D_i=\widehat D_j\cd \widehat L$ can be derived from
$\widehat D_j$ with a derivation of at most polynomial in $|D_j\Or L|$ many steps.

Assume that $D_i$ was derived from $D_j$ where $D_j$ is $D_i\Or (0=k)$,
using the Simplification inference
rule of R(lin),
where $j<i\le\l$ and $k$ is a non-zero integer.
Then, $\widehat D_i$ can be derived from $\widehat D_j = \widehat
D_i \cd -k$ by multiplying with $-k^{-1}$ (via the Addition rule
of PCR).

Thus, it remains to simulate the \emph{resolution rule} application of R(lin).
Let $A,B$ be two disjunctions of linear equations and
assume that $A\Or B\Or((\vec a +\vec b)\cd \vec x = a_0+b_0)$
was derived in $\pi$ from $A\Or(\vec a\cd\vec x = a_0)$ and
$B\Or(\vec b\cd\vec x = b_0)$
(the case where $A\Or B\Or((\vec a -\vec b) \cd\vec x = a_0-b_0)$
was derived from
$A\Or(\vec a\cd\vec x = a_0)$ and
$B\Or(\vec b\cd\vec x = b_0)$, is similar).

We need to derive
$\widehat{A}\cd\widehat{B}\cd((\vec a+\vec b)\cd \vec x -a_0-b_0)$ from
$\widehat{A}\cd(\vec a\cd\vec x-a_0)$ and
$\widehat{B}\cd(\vec b\cd\vec x-b_0)$.
This is done by multiplying $\widehat{A}\cd(\vec a\cd\vec x-a_0)$
with $\widehat{B}$ and multiplying
$\widehat{B}\cd(\vec b\cd\vec x-b_0)$
with $\widehat{A}$ (using Claim \ref{cla-simple-PCR-prod}),
and then adding the resulted polynomials together.
\end{proof}

\begin{remark}
When translating R(lin) proofs into PCR proofs we actually do
not make any use of the ``negative" variables $\bar x_1,\ldots, \bar x_n$.
Nevertheless, the multilinear proof systems make use of these variables
in order to polynomially simulate PCR proofs
(see Theorem \ref{thm-RT05-PCRsim} and its proof in \cite{RT05}).
\end{remark}

We shall need the following corollary in the sequel:
\begin{corollary}\label{cor-PCR-RZ0-lines-translation}
Let $\pi = D_1,\ldots,D_\l$ be an \RZ0 proof of $D_\l$,
and let $s$ be the maximal size of an \RZ0-line in $\pi$.
Then there is a PCR proof $\pi'$ of $\widehat D_\l$
with polynomial-size in $|\pi|$ number
of steps and such that every line of $\pi'$ is a translation
(via Definition \ref{def-poly-trnas-of-R(lin)-lines})
of an $\RZ0$-line (Definition \ref{def-RZ0-clause}),
where the size of the \RZ0-line is polynomial in $s$.
\end{corollary}

\begin{proof}
The simulation of R(lin) by PCR shown above,
can be thought of as, first, considering
$\widehat D_1,\ldots,\widehat D_\l$ as the ``skeleton" of a PCR proof
of $\widehat D_\l$.
And second, for each $D_i$ that was deduced by one of R(lin)'s
inference rules from previous lines,
one inserts the corresponding PCR proof sequence that simulates the appropriate
inference rule application
(as described in the proof of Proposition \ref{prop-PCR-sim-R(lin)}).
By definition, those PCR proof-lines that correspond to lines in the skeleton
$\widehat D_1,\ldots,\widehat D_\l$ are translations of \RZ0-lines (with size
at most polynomial in $s$).
Thus, to conclude the proof of the corollary,
one needs only to check that for any
\RZ0-line $D_i$ that was deduced by one of R(lin)'s inference rules
from previous \RZ0-lines (as demonstrated in the proof
of Proposition \ref{prop-PCR-sim-R(lin)}),
the inserted corresponding PCR proof sequence uses only translations of \RZ0-lines
(with size polynomial in $s$).
This can be verified by a straightforward inspection.
\end{proof}

%=========================================================
\subsection{From PCR Proofs to Multilinear Proofs}\label{sec-PCR-to-fMC}
%=========================================================

We now recall the general simulation result proved in \cite{RT05}
stating the following:
Let $\pi$ be a PCR
refutation of some initial collection of multilinear polynomials $Q$ over some
fixed field.
Assume that $\pi$ has polynomially many steps
(that is, the number of proof lines in the PCR proof sequence
is polynomial).
If the `multilinearization' (namely, the result of applying the $\ML{\cd}$
operator -- see Definition \ref{def-ml-operator})
of each of the polynomials in $\pi$ has a polynomial-size depth $d$
multilinear formula (with a plus gate at the root),
then there is a polynomial-size depth-$d$ fMC refutation of $Q$.
More formally, we have:

\begin{theorem}[\cite{RT05}]\label{thm-RT05-PCRsim}
Fix a field $\F$ (not necessarily of characteristic $0$)
and let $Q$ be a set of
multilinear polynomials from $\,\F[x_1, \ldots, x_n, \bx_1, \ldots, \bx_n]$.
Let $\pi=(p_1,\ldots, p_m)$ be a \emph{PCR} refutation of $\,Q$.
For each $p_i\in \pi$, let $\Phi_i$ be a multilinear formula for the
polynomial $\ML{p_i}$.
Let $s$ be the total size of all formulas $\Phi_i$, that is,
$s=\Sigma_{i=1}^{m}{|\Phi_i|}$, and let $d\ge 2$ be the maximal depth of all
formulas $\Phi_i$.
Assume that the depth of all the formulas $\Phi_i$ that have
a product gate at the root is at most $d-1$.
Then there is a depth-$d$
\emph{fMC} refutation of $\,Q\,$ of size polynomial in $s$.
\end{theorem}

%====================================================
\subsubsection{Depth-$3$ Multilinear Proofs}
%====================================================

Here we show that multilinear proofs operating with depth-$3$ multilinear
formulas (that is, depth-$3$ fMC) over fields of characteristic $0$
polynomially simulate \RZ0 proofs. In light of Proposition
\ref{cor-PCR-RZ0-lines-translation} and Theorem \ref{thm-RT05-PCRsim},
to this end it suffices to show that any \RZ0-line $D$ %(Definition \ref{def-RZ0-clause}),
translates into a corresponding polynomial $p$
(via the translation in Definition \ref{def-poly-trnas-of-R(lin)-lines})
such that $\ML{p}$ has a multilinear formula of size polynomial (in the number of variables)
and depth at most $3$ (with a plus gate at the root) over fields of
characteristic $0$.

We need the following proposition from \cite{RT05}:

\begin{proposition}[\cite{RT05}]\label{prop-RT05-depth3-ml}
Let $\,\F$ be a field of characteristic $0$. For a constant $\,c$, let
$X_1,\ldots, X_c$ be $\,c\,$ finite sets of variables (not necessarily
disjoint), where $\Sigma_{i=1}^c{|X_i|}=n\,$. Let $\,f_1,\ldots, f_c\,$ be
$\,c\,$ symmetric polynomials over $\,X_1, \ldots, X_c\,$ (over the field
$\,\F$), respectively.  Then, there is a depth-$3$ multilinear formula for
$\,\ML{f_1\cdots f_c}\,$ of size polynomial (in $n$), with a plus gate at
the root.
\end{proposition}

The following is the key lemma of the simulation:

\begin{lemma}\label{lem-RZ0-line-into-3ml-fmla}
Let $D$ be an \RZ0-line with $n$ variables %(see Definition \ref{def-RZ0-clause}).
and let $p= \widehat D$ (see Definition \ref{def-poly-trnas-of-R(lin)-lines}).
Then, $\ML{p}$ has a depth-$3$ multilinear formula over fields of
characteristic $0$, with a plus gate at the root and size at most polynomial in
the size of $D$.
\end{lemma}

\begin{proof}
Assume that the underlying variables of $D$ are $\vec x = x_1\ldots,x_n$.
By the definition of an \RZ0-line (see Definition \ref{def-RZ0-clause})
we can partition the disjunction $D$ into a \emph{constant number} of disjuncts,
where one disjunct is a (possibly empty, translation of a)
clause $C$,\footnote{If there is more than one clause in $D$, we
simply combine all the clauses into a single clause.}
%(see Section \ref{sec-clauses-of-linear-equations}),
and all other disjuncts have the following form:
\begin{equation}\label{eq-3ml-disjunct}
\BigOr_{i=1}^m
 \left(
   \vec a \cd \vec x =\ell_i
 \right)\,,
\end{equation}
where the $\ell_i$'s are integers, $m$ is not necessarily bounded and
$\vec a$ denotes a vector of $n$ \emph{constant} integer coefficients.

Let us denote by $q$ the polynomial representing the clause $C$.\footnotemark
\footnotetext{$C$ is a translation of a clause (that is, disjunction of literals)
into a disjunction of linear equations, as
defined in Section \ref{sec-clauses-of-linear-equations}.
The polynomial $q$ is then the polynomial translation of this disjunction
of linear equations,
as in Definition \ref{def-poly-trnas-of-R(lin)-lines}.}

Consider a disjunct as shown in (\ref{eq-3ml-disjunct}).
Since the coefficients $\vec a$ are constants,
$\vec a \cd \vec x$ can be written as a sum of constant number of linear
forms, each with the \emph{same} constant coefficient.
In other words, $\vec a \cd \vec x$ can be written as
$z_1 \PlusDots z_d$, for some constant $d$,
where for all $i\in[d]$:
\begin{equation}\label{eq-Z-partition-of-line}
z_i:=b\cd \sum_{j\in J} x_j\,,
\end{equation}
for some $J\subseteq [n]$ and some constant integer $b$.
We shall assume without loss of generality that $d$ is the same constant
for every disjunct of the form (\ref{eq-3ml-disjunct}) inside $D$
(otherwise, take $d$ to be the maximal such $d$).

Thus, (\ref{eq-3ml-disjunct}) is translated (via the translation scheme
in Definition \ref{def-poly-trnas-of-R(lin)-lines})
into:
\begin{equation}\label{eq-prod-Z-substitution}
\prod\limits_{i = 1}^m {(z_1  + ... + z_d  - \ell _i )}\,.
\end{equation}

By fully expanding the product in (\ref{eq-prod-Z-substitution}),
we arrive at:
\begin{equation}\label{eq-expand-prod-disjunct-3ML}
 \sum\limits_{
				\scriptstyle r_1  + \ldots + r_{d + 1}  = m
				}
	{
		\left(
			{\alpha _{r_{d+1}}\cdot\prod\limits_{k = 1}^d
				 {
					z_k ^{r_k}
				  }
			 }
		\right)
	}
	\,,
\end{equation}
where the $r_i$'s are non-negative integers,
and where the $\alpha_{r}$'s, for every $0\le r\le m$ are just integer coefficients,
formally defined as follows (this definition is not essential;
we present it only for the sake of concreteness):
\begin{equation}\label{eq-coeff-for-sake-complete}
\alpha_r:=
 \sum\limits_{\scriptstyle
   \,U\subseteq [m] \hfill\atop\scriptstyle |\,U| = r\hfill}
{\prod\limits_{j \in U}{( - \ell _j ) }}
\,.
\end{equation}

\begin{claim}\label{cla-tedious-summands}
The polynomial $\widehat D$ (the polynomial translation of $D$)
is a linear combination (over $\F$) of polynomially (in $|D|$) many terms,
such that each term can be written as
\[
	q\cd\prod_{k\in K} z_k^{r_k}\,,
\]
where $K$ is a collection of a constant number of indices,
$r_k$'s are non-negative integers,
and the $z_k$'s and  $q$ are as above
(that is, the $z_k$'s are linear forms, where
each $z_k$ has a single coefficient for all variables in it,
as in (\ref{eq-Z-partition-of-line}),
and $q$ is a polynomial translation of a clause).
%with no common variables with the $z_k$'s, for all $k\in K$).
\end{claim}

\begin{proofclaim}
Denote the total number of disjuncts of the form
(\ref{eq-3ml-disjunct}) in $D$ by $h$.
By definition (of \RZ0-line), $h$ is a constant.
Consider the polynomial (\ref{eq-expand-prod-disjunct-3ML}) above.
In $\widehat D$,
we actually need to multiply $h$ many polynomials of the form
shown in (\ref{eq-expand-prod-disjunct-3ML}) and the polynomial $q$.
%a single polynomial translation of the clause $C$, denoted $q$
%(where the variables in $q$ do not appear in other polynomials).

For every $j\in[h]$ we write the (single) linear
form in the $j$th disjunct as a sum of constantly many linear
forms $z_{j,1}+\ldots+z_{j,d}$,
where each linear form $z_{j,k}$ has the same coefficient for every variable in it.
%for some integer constant $d$,
%and the $z_{j,k}$'s are as shown in (\ref{eq-Z-partition-of-line}).
Thus, $\widehat D$ can be written as:
\begin{equation}\label{eq-3ML-big-prod}
q\cd \prod\limits_{\scriptstyle j = 1}^h
 {\left(
  \sum\limits_{
			   \scriptstyle r_1+\ldots+r_{d+1}=m_j
			  }
	{\underbrace{
				 \left({\alpha ^{(j)} _{
									r_{
										d+1
									  }
								}\cdot\prod\limits_{k = 1}^{d}
									{
										z_{j,k}^{
												r_{k}
											}
									}
						}
				 \right)
				 }_{(\star)}
	}
  \right)
 }\,,
\end{equation}
(where the $m_j$'s are not bounded, and the coefficients
$\alpha ^{(j)} _{r_{d+1}}$
are as defined in (\ref{eq-coeff-for-sake-complete}) except that here we add the
index $(j)$ to denote that they depend on the $j$th disjunct in $D$).
Denote the maximal $m_j$, for all $j\in [h]$, by $m_0$.
The size of $D$, denoted $|D|$, is at least $m_0$.
Note that since $d$ is a constant,
the number of summands in each (middle) sum in (\ref{eq-3ML-big-prod})
is polynomial in $m_0$, which is at most polynomial in $|D|$.
 Thus, by expanding the outermost product in (\ref{eq-3ML-big-prod}),
we arrive at a sum of polynomially in $|D|$ many summands.
Each summand in this sum is a product of $h$ terms of the form $(\star)$
multiplied by $q$.
\end{proofclaim}

It remains to apply the multilinearization operator
(Definition \ref{def-ml-operator}) on $\widehat D$,
and verify that the resulting polynomial has a depth-$3$ \ml0 formula
with a plus gate at the root and of polynomial-size (in $|D|$).
Since $\ML \cd$ is a linear operator, it suffices to show that
when applying $\ML \cd$ on each summand in $\widehat D$,
as described in Claim \ref{cla-tedious-summands}, one obtains
a (multilinear) polynomial that has a depth-$3$ \ml0 formula
with a plus gate at the root,
and of polynomial-size in the number of variables $n$ (note that clearly $n\le |D|$).
This is established in the following claim:

\begin{claim}\label{cla-summand-3ml}
%Let $K$ be a collection of constant number of indices,
%$r_k$ are non-negative integers (for all $k\in K$),
%and $q$ is a polynomial translation of a clause with
%no common variables with the $z_k$'s (for every $k\in K$).
The polynomial $\ML{q\cd\prod_{k\in K} z_k^{r_k}}$
has a depth-$3$ multilinear formula of polynomial-size in
$n$ (the overall number of variables) and with a plus gate at the root (over
fields of characteristic $0$),
under the same notation as in Claim \ref{cla-tedious-summands}.
%A product of constant terms of the form $(\star)$ and $q$
%has a depth-$3$ \ml0 formula of polynomial-size in $n$ with
%a plus gate at the root.
\end{claim}

\begin{proofclaim}
Recall that a power of a symmetric polynomial is a symmetric polynomial in itself.
Since each $z_k$ (for all $k\in K$) is a symmetric polynomial,
then its power $z_k^{r_k}$ is also symmetric.
The polynomial $q$ is a translation of a clause, hence it is a product
of two symmetric polynomials:
the symmetric polynomial that is the translation
of the disjunction of literals with positive signs,
and the symmetric polynomial that is the translation
of the disjunction of literals with negative signs.
Therefore, $q\cd\prod_{k\in K} z_k^{r_k}$ is a product of constant number
of symmetric polynomials.
By Proposition \ref{prop-RT05-depth3-ml},
$\ML{q\cd\prod_{k\in K} z_k^{r_k}}$
(where here the $\ML \cd$ operator operates on the
$\vec x$ variables in the $z_k$'s and $q$)
is a polynomial for which there is
a polynomial-size (in $n$) depth-$3$ multilinear formula
with a plus gate at the root (over fields of characteristic $0$).
\end{proofclaim}
\end{proof}

We now come to the main corollary of this section.

\begin{corollary}\label{cor-3fMC-simulates-RZ0}
Multilinear proofs operating with depth-$3$ multilinear formulas
(that is, depth-$3$ {\rm fMC} proofs)
polynomially-simulate \RZ0 proofs.
\end{corollary}

\begin{proof}
Immediate from Corollary \ref{cor-PCR-RZ0-lines-translation}, Theorem
\ref{thm-RT05-PCRsim} and Proposition \ref{lem-RZ0-line-into-3ml-fmla}.

For the sake of clarity we repeat the chain of transformations needed to
prove the simulation.
Given an \RZ0 proof $\pi$, we first use Corollary \ref{cor-PCR-RZ0-lines-translation}
to transform $\pi$ into a PCR proof $\pi'$,
with number of steps that is at most polynomial in $|\pi|$,
and where each line in $\pi'$ is a polynomial translation of some \RZ0-line with
size at most polynomial in the maximal line in $\pi$
(which is clearly at most polynomial in $|\pi|$).
Thus, by Proposition \ref{lem-RZ0-line-into-3ml-fmla} each polynomial in $\pi'$
has a corresponding multilinear polynomial with a \ps0 in $|\pi|$
depth-$3$ multilinear formula (and a plus gate at the root).
Therefore, by Theorem \ref{thm-RT05-PCRsim}, we can transform
$\pi'$ into a depth-$3$ fMC proof with only a polynomial (in $|\pi|$) increase in size.
\end{proof}

\subsection{Small Depth-$3$ Multilinear Proofs}
Since \RZ0 admits polynomial-size (in $n$)
refutations of the $m$ to $n$ pigeonhole principle (for any $m>n$)
(as defined in \ref{def-PHP}),
Corollary \ref{cor-3fMC-simulates-RZ0} and Theorem \ref{thmPHP}
yield:

\begin{theorem}\label{thm-3fMC-PHP-ref}
For any $m>n$ there are polynomial-size (in $n$) depth-$3$ {\rm fMC}
refutations of the $m$ to $n$ pigeonhole principle {\rm PHP}$^m_n$
(over fields of characteristic $0$).
\end{theorem}

This improves over the result in \cite{RT05} that demonstrated a polynomial-size
(in $n$) depth-$3$ fMC refutations of a weaker principle,
namely the $m$ to $n$ \emph{functional} pigeonhole principle.

Furthermore, corollary \ref{cor-3fMC-simulates-RZ0} and Theorem \ref{thm-Tse}
yield:

\begin{theorem}\label{thm-3fMC-Tseitin-ref}
Let $G$ be an $r$-regular graph with $n$ vertices, where $r$ is a constant,
and fix some modulus $p$.
Then there are polynomial-size (in $n$) depth-$3$ {\rm fMC} refutations of
Tseitin {\rm mod} $p$ formulas \Tse0 (over fields of characteristic $0$).
\end{theorem}

The \ps0 refutations of Tseitin graph tautologies here are different
than those demonstrated in \cite{RT05}.
Theorem \ref{thm-3fMC-Tseitin-ref} establishes \ps0 refutations over
any field of characteristic $0$ of Tseitin mod $p$ formulas,
whereas \cite{RT05} required the field to contain a primitive
$p$th root of unity.
On the other hand, the refutations in \cite{RT05} of Tseitin mod $p$ formulas
do not make any use of the semantic nature
of the fMC proof system, in the sense that they
do not utilize the fact that the base field is of characteristic $0$
(which in turn enables one to efficiently represent any symmetric [multilinear]
polynomial by a depth-$3$ multilinear formula).

%===========================================================
\section{Relations with Extensions of Cutting Planes}\label{sec-CP}
%===========================================================
In this section we tie some loose ends by showing
that, in full generality, R(lin) polynomially simulates R(CP)
with polynomially bounded coefficients, denoted R(CP*).
First we define the R(CP*) proof system
-- introduced in \cite{Kra98-Discretely} --
which is a common extension of resolution
and CP* (the latter is cutting planes with polynomially bounded coefficients).
The system R(CP*), thus, is essentially resolution operating
with disjunctions of linear inequalities (with polynomially bounded
integral coefficients) augmented with the cutting planes
inference rules.

A linear inequality is written as
\begin{equation}\label{eq-inequality-line}
\vec a \cd \vec x \ge a_0\,,
\end{equation}
where $\vec a$ is a vector of integral coefficients $a_1,\ldots,a_n$,
$\vec x$ is a vector of variables $x_1,\ldots,x_n$,
%$\vec a\cd \vec x$ denotes the scalar multiplication $\sum_{i=1}^n{a_i x_i}$,
and $a_0$ is an integer.
The \emph{size} of the linear inequality (\ref{eq-inequality-line})
is the sum of all $a_0,\ldots,a_n$ written in \emph{unary notation}
%is $\sum_{i=0}^n{\abs{a_i}}$,
%where $\abs{a_i}$ is the bit size of $a_i$ written in binary notation
(this is similar to the size of linear equations in R(lin)).
A \emph{disjunction of linear inequalities}
is just a disjunction of inequalities of the form in (\ref{eq-inequality-line}).
The semantics of a disjunction of inequalities is the natural one,
that is, a disjunction is true under an assignment of integral values to $\vec x$
if and only if at least one of the inequalities is true under the assignment.
The \emph{size of a disjunction of linear inequalities} is the total size of all
linear inequalities in it.
We can also add in the obvious way linear inequalities, that is,
if $L_1$ is the linear inequality $\vec a\cd \vec x \ge a_0$ and $L_2$ is the
linear inequality $\vec b\cd\vec x\ge b_0$,
then $L_1+L_2$ is the linear inequality $(\vec a +\vec b)\cd \vec x \ge a_0+b_0$.
%(where $\vec a+\vec b$ denotes the vector addition of $\vec a$ and $\vec b$).
\QuadSpace

The proof system R(CP*) operates with disjunctions of linear inequalities
with integral coefficients (written in \emph{unary }representation),
and is defined as follows
(our formulation is similar to that in \cite{Koj07}):\footnotemark
\footnotetext{When we allow coefficients to be written in \emph{binary representation}, instead
of unary representation, the resulting proof system is denoted R(CP).}
\begin{definition}[R(CP*)]\label{def-R(CP)}
Let $K:= \set{K_1,\ldots,K_m}$ be a collection of disjunctions of linear inequalities
(whose coefficients are written in unary representation).
An \emph{R(CP*)-proof from $K$ of a disjunction of linear inequalities $D$}
is a finite sequence $\pi =(D_1 ,...,D_\ell)$ of disjunctions of linear
inequalities,
such that $D_\ell=D$
and for each $i\in[\l]$:
either $\,D_i = K_j\,$ for some $j\in[m]$;
or $D_i$ is one of the following \emph{R(CP*)-axioms}:
\begin{enumerate}
\item $x_i \ge 0$, for any variable $x_i$;
\item $-x_i\ge -1$,  for any variable $x_i$;
\item $(\vec a\cd \vec x \ge a_0)\Or ( -\vec a\cd \vec x \ge 1-a_0)$, where all
coefficients (including $a_0$) are integers;
\end{enumerate}
or $D_i$ was deduced from previous lines by one of the following
\emph{R(CP*)-inference rules}:
\begin{enumerate}
\item Let $A,B$ be two disjunctions of linear inequalities
and let $L_1,L_2$ be two linear inequalities.\footnotemark
\footnotetext{In all R(CP*)-inference rules, $A,B$ are possibly the empty
disjunctions.}
From $A\Or L_1$ and $B\Or L_2$ derive $A \Or B \Or(L_1+L_2)$.

\item Let $L$ be some linear equation.

From a disjunction of linear equations $A$ derive $A \Or L$.

\item  Let $A$ be a disjunction of linear equations

From $A\Or (0\ge 1)$ derive $A$.

\item Let  $c$ be a non-negative integer.

From $(\vec a\cd \vec x \ge a_0)\Or A$ derive
$(c\vec a\cd \vec x \ge c a_0) \Or A$.

\item Let $A$ be a disjunction of linear inequalities, and let $c\ge 1$ be an integer.

From $(c\vec a \cd \vec x \ge a_0) \Or A$
derive $\left(a\cd \vec x \ge \lceil a_0/c\rceil\right) \Or A$.

\end{enumerate}
An \emph{R(CP*) refutation of} a collection of disjunctions of linear inequalities
$K$ is a proof of the empty disjunction from $K$.
The \emph{size} of a proof $\pi$ in R(CP*) is the total size of all the
disjunctions of linear inequalities in $\pi$, denoted $|\pi|$.
\end{definition}

In order for R(lin) to simulate R(CP*) proofs,
we need to fix the following translation scheme.
%Let $K$ be a disjunction of linear inequalities.
Every linear inequality $L$ of the form
$\vec a \cd \vec x \ge a_0$
is translated into the following disjunction, denoted $\widehat L$:
\begin{equation}\label{eq-trans-inequality}
\left(
		\vec a \cd \vec x = a_0
\right)
	\Or
 \left(
		\vec a \cd \vec x = a_0+1
 \right)
	\OrDots
 \left(
		\vec a \cd \vec x = a_0+k
 \right)
\,,
\end{equation}
where $k$ is such that $a_0+k$ equals the sum of all positive coefficients in $\vec a$,
that is, $a_0+k = \max\limits_{\vec x\in\zo^n}\left(\vec a\cd \vec x\right)$
(in case the sum of all positive coefficients in $\vec a$ is less than $a_0$, then
we put $k=0$).
An inequality with no variables of the form $0\ge a_0$ is translated into
$0=a_0$ in case it is false (that is, in case $0<a_0$),
and into $0=0$ in case it is true (that is, in case $0\ge a_0$).
Note that since the coefficients of linear inequalities (and linear equations)
are written in \emph{unary} representation,
any linear inequality of size $s$ translates into
a disjunction of linear equations of size $O(s^2)$.
Clearly, every $0,1$ assignment to the variables $\vec x$
satisfies $L$ if and only if it satisfies its translation $\widehat L$.
A disjunction of linear inequalities $D$
%$\BigOr\limits_{i=1}^\ell \left(\vec a^{(i)} \cd\vec x \ge d^{(i)} \right)$\,
is translated into the disjunction of the translations of all the linear inequalities
in it, denoted $\widehat D$.
A collection $K:= \set{K_1,\ldots,K_m}$ of disjunctions of linear inequalities,
is translated into the collection $\set{\widehat K_1,\ldots,\widehat K_m}$.

\begin{theorem}\label{R(lin)-sim-R(CP)}
R(lin) polynomially-simulates R(CP*).
In other words, if $\pi$ is an R(CP*) proof of a linear inequality
$D$ from a collection of disjunctions of linear inequalities
$K_1,\ldots,K_t$,
then there is an R(lin) proof of $\widehat D$ from
$\widehat K_1,\ldots,\widehat K_t$ whose size is polynomial in $|\pi|$.
\end{theorem}

\begin{proof}
By induction on the number of proof-lines in $\pi$.

\emph{Base case:} Here we only need to show that the
axioms of R(CP*) translates into axioms of R(lin),
or can be derived with polynomial-size
(in the size of the original R(CP*) axiom)
R(lin) derivations (from R(lin)'s axioms).

R(CP*) axiom number (1): $x_i\ge 0$ translates into the R(lin) axiom $(x_i=0)\Or(x_i=1)$.

R(CP*) axiom number (2): $-x_i\ge -1$, translates into $(-x_i=-1)\Or(-x_i=0)$.
From the Boolean axiom $(x_i=1)\Or(x_i=0)$ of R(lin),
one can derive with a constant-size R(lin) proof the line
$(-x_i=-1)\Or(-x_i=0)$ (for instance, by subtracting twice
each equation in $(x_i=1)\Or(x_i=0)$ from itself).
%for instance, by Lemma \ref{cor-impl-comp-constant-size-proof}).

R(CP*) axiom number (3):
$(\vec a\cd \vec x \ge a_0)\Or ( -\vec a\cd \vec x \ge 1-a_0)$.
The inequality $(\vec a\cd \vec x \ge a_0)$ translates into
$$\BigOr\limits_{b=a_0}^{h} (\vec a\cd \vec x = b)\,,$$
where $h$ is the maximal value of $\vec a \cd \vec x$ over $0,1$ assignments to $\vec x$
(that is, $h$ is just the sum of all positive coefficients in $\vec a$).
The inequality $(-\vec a\cd \vec x \ge 1-a_0)$ translates into
$$\BigOr\limits_{b=1-a_0}^{f} (-\vec a\cd \vec x = b)\,,$$
where $f$ is the maximal value of $-\vec a \cd \vec x$ over $0,1$ assignments to $\vec x$
(that is, $f$ is just the sum of all negative coefficients in $\vec a$).
Note that one can always flip the sign of any equation $\vec a \cd \vec x =b$
in R(lin).
This is done, for instance, by subtracting twice  $\vec a \cd \vec x =b$
from itself.
So overall R(CP*) axiom number (3) translates into
$$
\BigOr\limits_{b=a_0}^{h}
	(\vec a\cd \vec x = b)
		\Or
\BigOr\limits_{b=1-a_0}^{f}
	(-\vec a\cd \vec x = b)\,,
$$
that can be converted inside R(lin) into
\begin{equation}\label{eq-R(CP)-axiom3-trans}
\BigOr\limits_{b=-f}^{a_0-1}
	(\vec a\cd \vec x = b)
		\Or
\BigOr\limits_{b=a_0}^{h}
	(\vec a\cd \vec x = b)\,.
\end{equation}
Let $\mathcal A' :=\set{-f,-f+1,\ldots,a_0-1,a_0,a_0+1,\ldots,h}$
and let $\mathcal A$ be the set of all possible values that
$\vec a \cd \vec x$ can get over all possible Boolean assignments to $\vec x$.
Notice that $\mathcal A\subseteq \mathcal A'$.
%Thus, disjunction (\ref{eq-R(CP)-axiom3-trans}) is the disjunction expressing
%that $\vec a \cd\vec x$ equals (at least) one value from all possible values
%that $\vec a \cd\vec x$ can get over $0,1$ assignments to $\vec x$.
By Lemma \ref{lem-basic-count-all-possibilities},
for any $\vec a\cd\vec x$, there is a polynomial-size
(in the size of the linear form $\vec a\cd\vec x$)
%; recall that this size is the sum of all coefficients in $\vec a$ written in unary representation)
derivation of $\BigOr_{\alpha\in\mathcal A}(\vec a\cd\vec x =\alpha)$.
By using the R(lin) Weakening rule we can then derive
$\BigOr_{\alpha\in\mathcal A'}(\vec a\cd\vec x =\alpha)$
which is equal to (\ref{eq-R(CP)-axiom3-trans}).
\QuadSpace

\emph{Induction step:}
Here we simply need to show how to polynomially simulate
inside R(lin) every inference rule application of R(CP*).\QuadSpace

\para{Rule  (1):}
Let $A,B$ be two disjunctions of linear inequalities
and let $L_1,L_2$ be two linear inequalities.
Assume we already have a R(lin) proofs of
$\widehat A\Or \widehat L_1$ and $\widehat B\Or \widehat L_2$.
We need to derive $\widehat A \Or \widehat B \Or \widehat {L_1+L_2}$.
Corollary \ref{lem-basic-count-combine-two-lines} shows that there is a
polynomial-size (in the size of $\widehat L_1$ and $\widehat L_2$; which
is polynomial in the size of $L_1$ and $L_2$)
 derivation of $\widehat {L_1+L_2}$ from $\widehat L_1$ and
$\widehat L_2$, from which the desired derivation immediately follows.

\para{Rule  (2):} The simulation of this rule in R(lin) is done using the
R(lin) Weakening rule.
\para{Rule  (3):} The simulation of this rule in R(lin) is done using the
R(lin) Simplification rule (remember that $0\ge 1$ translates into $0=1$ under
our translation scheme).

\para{Rule  (4):}
Let  $c$ be a non-negative integer.
We need to derive
$\widehat{(c\vec a\cd \vec x \ge c a_0)}\Or \widehat A$
from $\widehat{(\vec a\cd \vec x \ge a_0)}\Or \widehat A$
in R(lin).
This amounts only to ``adding together" $c$ times
the disjunction $\widehat{(\vec a\cd \vec x \ge  a_0)}$ in
$\widehat{(\vec a\cd \vec x \ge a_0)}\Or \widehat A$.
This can be achieved by $c$ many applications of
Corollary \ref{lem-basic-count-combine-two-lines}.
We omit the details.

\para{Rule  (5):}
We need to derive
$\widehat{
	\left(
		\vec a\cd \vec x \ge \lceil a_0/c\rceil
	\right)}
	 \Or
	\widehat A
$, from
$\widehat{(c\vec a \cd \vec x \ge a_0)} \Or \widehat A$.
Consider the disjunction of linear equations
$\widehat{(c\vec a \cd \vec x \ge a_0)}$, which can be written as:
\begin{equation}\label{eq-R(CP*)-rule(5)-sim}
	(c\vec a\cd\vec x =a_0) \Or
	(c\vec a\cd\vec x =a_0+1)
		\Or \ldots \Or
	(c\vec a\cd\vec x =a_0+r)
\,,
\end{equation}
where $a_0+r$ is the maximal value $c\vec a\cd \vec x$ can get over $0,1$ assignments
to $\vec x$.
By Lemma \ref{lem-basic-count-all-possibilities} there is a \ps0 (in the size of
$\vec a\cd\vec x$) R(lin) proof of
\begin{equation}\label{eq-R(CP*)-rule(5)-all-possibilities}
\BigOr\limits_{
				\alpha \in \mathcal A
				 }
				 {
				  (\vec a \cd \vec x = \alpha)
				 }
	\,,
\end{equation}
where $\mathcal A$ is the set of all possible values of $\vec a\cd\vec x$ over
$0,1$ assignments to $\vec x$.

We now use (\ref{eq-R(CP*)-rule(5)-sim}) to cut-off from
(\ref{eq-R(CP*)-rule(5)-all-possibilities}) all equations $(\vec a\cd\vec x
=\beta)$ for all $\beta<\lceil a_0/c\rceil$
(this will give us the desired disjunction of linear equations).
Consider the equation $(\vec a\cd\vec x =\beta)$ in
(\ref{eq-R(CP*)-rule(5)-all-possibilities})
for some fixed $\beta<\lceil a_0/c\rceil$.
Use the resolution rule of R(lin) to add this equation to itself
$c$ times inside (\ref{eq-R(CP*)-rule(5)-all-possibilities}).
We thus obtain
\begin{equation}\label{eq-R(CP*)-rule(5)-all-poss-beta}
(c\vec a \cd \vec x = c\beta)
	\Or
\BigOr\limits_{
				\alpha \in \mathcal A \setminus \set{\beta}
				 }
				 {
				  (\vec a \cd \vec x = \alpha)
				 }
	\,.
\end{equation}
Since $\beta$ is an integer and $\beta<\lceil a_0/c\rceil$,
we have $c\beta < a_0$. Thus, the equation
$(c\vec a \cd \vec x = c\beta)$ does not appear in (\ref{eq-R(CP*)-rule(5)-sim}).
We can then successively resolve $(c\vec a \cd \vec x = c\beta)$ in
(\ref{eq-R(CP*)-rule(5)-all-poss-beta}) with each equation
$(c\vec a\cd\vec x=a_0),\ldots,(c\vec a\cd\vec x=a_0+r)$ in
(\ref{eq-R(CP*)-rule(5)-sim}). Hence, we arrive at
$\BigOr_{
		 \alpha \in \mathcal A \setminus \set{\beta}
		}
		{
		 (\vec a \cd \vec x = \alpha)
		}
$.
Overall, we can cut-off all equations $(\vec a\cd\vec x=\beta)$, for
$\beta<\lceil a_0/c\rceil$, from (\ref{eq-R(CP*)-rule(5)-all-possibilities}).
We then get the disjunction
\[
\BigOr\limits_{
			\alpha\in \mathcal A'
			  }
			(
			\vec a \cd \vec x = \alpha
			)
\,,
\]
where $\mathcal A'$ is the set of all elements of $\mathcal A$ greater or
equal to $\lceil a_0/c\rceil$
(in other words, all values greater or equal to $\lceil a_0/c\rceil$
that $\vec a \cd\vec x$
can get over $0,1$ assignments to $\vec x$).
Using the Weakening rule of R(lin) (if necessary) we
can arrive finally at the desired disjunction
$\widehat{
	\left(
		\vec a\cd \vec x \ge \lceil a_0/c\rceil
	\right)}
$,
which concludes the R(lin) simulation of R(CP*)'s inference Rule (5).
\end{proof}

\appendix

%==================================================================
\section{Feasible Monotone Interpolation}
%==================================================================
\label{sec-mono-interp}

Here we formally define the feasible
monotone interpolation property.
The definition is taken mainly from \cite{Kra97-Interpolation}.

Recall that for two binary strings of length $n$ (or equivalently, Boolean
assignments for $n$ propositional variables) $\alpha,\alpha'$, we denote by
$\alpha'\ge \alpha$ that $\alpha'$ is \emph{bitwise} greater than $\alpha$,
that is, that for all $i\in[n]$, $\alpha'_i\ge \alpha_i$ (where $\alpha'_i$ and
$\alpha_i$ are the $i$th bits of $\alpha'$ and $\alpha$, respectively). Let
$A(\vec p,\vec q), B(\vec p,\vec r)$ be two collections of formulas in the
displayed variables only, where $\vec p,\vec q,\vec r$ are pairwise disjoint
sequences of distinct variables (similar to the notation at the beginning of
Section \ref{sec-interpo}). Assume that there is no assignment that satisfies
both $A(\vec p,\vec q)$ and $B(\vec p,\vec r)$. We say that $A(\vec p,\vec q),
B(\vec p,\vec r)$ are \emph{monotone }if one of the following conditions hold:
\begin{enumerate}
\item\label{it-cond-A}
If $\vec\alpha$ is an assignment to $\vec p$ and $\vec\beta$ is an assignment
to $\vec q$ such that $A(\vec\alpha, \vec\beta)=1$, then for any assignment
$\vec\alpha'\ge\vec\alpha$ it holds that $A(\vec\alpha', \vec\beta)=1$.

\item\label{it-cond-B}
 If $\vec\alpha$ is an assignment to $\vec p$ and
$\vec\beta$ is an assignment to $\vec r$ such that $B(\vec\alpha,
\vec\beta)=1$, then for any assignment $\vec\alpha'\le\vec\alpha$ it holds that
$B(\vec\alpha', \vec\beta)=1$.
\end{enumerate}

Fix a certain proof system $\mathcal P$.
Recall the definition of the interpolant function (corresponding
to a given unsatisfiable $A(\vec p,\vec q)\And B(\vec p,\vec r)$;
that is, functions for which
(\ref{eq-interp-circuit-C}) in Section \ref{sec-interpo} hold).
Assume that for every monotone $A(\vec p,\vec q),
B(\vec p,\vec r)$ there is a transformation from every $\mathcal P$-refutation
of $A(\vec p,\vec q)\And B(\vec p,\vec r)$ into the corresponding
interpolant \emph{monotone} Boolean circuit $C(\vec p)$
(that is, $C(\vec p)$ uses only monotone
gates\footnotemark) and whose size is polynomial in the size of the
refutation\footnotetext{For instance, a \emph{monotone Boolean circuit }is
a circuit that uses only $\And, \Or$ gates of fan-in two (see also Section
\ref{sec-lower-bounds}).
In certain cases, the monotone
interpolation technique is also applicable for a larger class of circuits, that
is, circuits that compute with real numbers and that can use any nondecreasing
real functions as gates (this was proved by Pudl\'ak in \cite{Pud97}).}
(note that for every monotone $A(\vec p,\vec q), B(\vec p,\vec r)$
the corresponding
interpolant circuit must compute a monotone function;\footnotemark
\footnotetext{That is, if $\alpha'\ge\alpha$ then $C(\alpha')\ge C(\alpha)$.}
the interpolant circuit itself, however, might not be
monotone, namely, it may use
non-monotone gates).
In such a case, we say that $\mathcal P$ has the
\emph{feasible monotone interpolation property}.
This means that, if a proof system ${\mathcal P}$ has the feasible monotone
interpolation property,
then an exponential lower bound on monotone circuits that compute the interpolant
function corresponding to $A(\vec p,\vec q)\And B(\vec p,\vec r)$ %(and specifically,
%circuits for which (\ref{eq-interp-circuit-C}) in Section \ref{sec-interpo} holds),
implies an exponential-size lower bound on $\mathcal
P$-refutations of $A(\vec p,\vec q)\And B(\vec p,\vec r)$. \QuadSpace

\begin{definition}[\textbf{Feasible monotone interpolation property}]
\label{def-mono-interp}
Let $\mathcal P$ be a propositional refutation system. Let $A_1(\vec p, \vec
q),\ldots,A_k(\vec p, \vec q)$ and $B_1(\vec p, \vec r),\ldots, B_\ell(\vec p,
\vec r)$
be two collections of formulas with the displayed variables only (where $\vec
p$ has $n$ variables, $\vec q$ has $s$ variables and $\vec r$ has $t$
variables), such that \emph{either}
(the set of satisfying assignments of)
$A_1(\vec p, \vec q),\ldots, A_k(\vec p, \vec q)$ meet condition \ref{it-cond-A}
above
\emph{or}
(the set of satisfying assignments of)
$B_1(\vec p, \vec r),\ldots, B_\ell(\vec p, \vec r)$ meet
condition \ref{it-cond-B} above.
 Assume that for any such $A_1(\vec p, \vec q),\ldots,A_k(\vec p, \vec q)$
and $B_1(\vec p, \vec r),\ldots, B_\ell(\vec p, \vec r)$,
if there exists a $\mathcal P$-refutation
for
$A_1(\vec p, \vec q)\And\cdots\And A_k(\vec p, \vec q)\And
B_1(\vec p, \vec r)\And\ldots\And B_\ell(\vec p, \vec r)$
of size $S$
then there exists
a monotone Boolean circuit separating $\,\mathcal U_A$
from $\mathcal V_B$ (as defined in Section \ref{sec-interp-sem-refs})
of size polynomial in $S$.
In this case we say that $\mathcal P$ possesses the \emph{feasible monotone
interpolation property}.
\end{definition}

%====================================================
\section*{Acknowledgments}
%====================================================
We wish to thank Arist Kojevnikov for useful correspondence on his paper. This
work was carried out in partial fulfillment of the requirements for the Ph.D.\
degree of the second author.

\bibliographystyle{alpha} %use "abbrv" of "plain" for numbers [2] and "alpha" for e.g., [BPR97]
%\bibliography{PrfCmplx}

\def\cprime{$'$} \def\cprime{$'$} \def\cprime{$'$} \def\cprime{$'$}

\end{document}